\documentclass[%
reprint,
superscriptaddress,
frontmatterverbose, 
preprintnumbers,
longbibliography,
amsmath,amssymb,
aps,
pra,
notitlepage,
nofootinbib,
onecolumn
]{revtex4-1}

\usepackage{titletoc}
\makeatletter
\def\l@subsection#1#2{}
\def\l@subsubsection#1#2{}
\makeatother

\usepackage[utf8]{inputenc}
\usepackage[english]{babel}
\usepackage[T1]{fontenc}
\usepackage{amsmath}
\makeatletter
\newcommand{\pushright}[1]{\ifmeasuring@#1\else\omit\hfill$\displaystyle#1$\fi\ignorespaces}
\newcommand{\pushleft}[1]{\ifmeasuring@#1\else\omit$\displaystyle#1$\hfill\fi\ignorespaces}
\makeatother

\usepackage{tikz}
\usepackage{amssymb}
\usepackage{color}
\usepackage{mathrsfs}
\usepackage{amsbsy}
\usepackage{amsthm}
\usepackage{thmtools,thm-restate}
\usepackage{graphicx}
\usepackage{tikz}
\usepackage{comment}

\usepackage{bbm}
\usepackage{bm}
\usepackage{epsfig}
\usepackage{xfrac}
\usepackage{xcolor}
\usepackage{enumerate}
\usepackage[shortlabels]{enumitem}
\usepackage{graphicx}
\usepackage{bm}
\usepackage{hyperref}
\usepackage{array}
\usepackage{makecell}
\usepackage{subfigure}
\hypersetup{
colorlinks=true,
linkcolor=blue,
filecolor=blue,
citecolor=blue,  
urlcolor=blue,
}

\newtheorem{theorem}{Theorem}
\newtheorem{lemma}[theorem]{Lemma}
\newtheorem{proposition}[theorem]{Proposition}
\newtheorem{corollary}[theorem]{Corollary}

\usepackage{braket}
\renewcommand{\v}[1]{\ensuremath{\mathbf{#1}}} 
\newcommand{\gv}[1]{\ensuremath{\text{\boldmath$ #1 $}}}
\newcommand{\abs}[1]{\left| #1 \right|} 
\newcommand{\norm}[1]{\left\| #1 \right\|} 

\newcommand{\trace}{\mathrm{Tr}}
\newcommand{\poly}{\mathrm{poly}}

\newcommand{\avg}{\mathrm{avg}}
\newcommand{\group}{{\mathrm{group}}} 
\newcommand{\point}{{\mathrm{local}}}
\newcommand{\charge}{{\mathrm{charge}}}

\newcommand{\txr}{{\mathscr{R}}}

\newcommand{\tPi}{{\tilde\Pi}}

\newcommand{\mC}{{\mathcal{C}}}

\newcommand{\mU}{{\mathcal{U}}}

\newcommand{\mN}{{\mathcal{N}}}

\newcommand{\tB}{{\widetilde{B}}}
\newcommand{\mR}{{\mathcal{R}}}
\newcommand{\frakF}{{\mathfrak{F}}}

\newcommand{\frakJ}{{\mathfrak{J}}}
\newcommand{\frakc}{{\mathfrak{c}}}

\newcommand{\mD}{{\mathcal{D}}}
\newcommand{\mE}{{\mathcal{E}}}

\newcommand{\id}{{\mathbbm{1}}}

\newcommand{\barF}{{F^\infty}}
\newcommand{\bgamma}{\overline{\gamma}}
\newcommand{\bepsilon}{\overline{\epsilon}}
\newcommand{\bdelta}{\overline{\delta}}
\newcommand{\barP}{\overline{P}}
\newcommand{\vac}{{\emptyset}}
\newcommand{\rep}{{\mathrm{rep}}}

\newcommand{\frakB}{{\mathfrak{B}}}

\newcommand{\is}{{\mathfrak{s}}}

\newcommand{\vj}{{\gv{j}}}

\newcommand{\vK}{{\v{K}}}

\newcommand{\bR}{{\mathbb{R}}}
\newcommand{\bZ}{{\mathbb{Z}}}
\newcommand{\bV}{{\mathbb{V}}}
\newcommand{\bN}{{\mathbb{N}}}

\renewcommand{\Re}{{\mathrm{Re}}}

\newcommand{\optL}{{\mathrm{opt}}}
\newcommand{\optG}{{\mathrm{opt(G)}}}
\newcommand{\optCG}{{\mathrm{opt(C)}}}
\newcommand{\cov}{{\mathrm{cov}}}
\newcommand{\covG}{{\mathrm{cov(G)}}}
\newcommand{\covCG}{{\mathrm{cov(C)}}}
\newcommand{\supp}{{\mathrm{supp}}}
\renewcommand{\epsilon}{\varepsilon}
\newcommand{\appropto}{\mathrel{\vcenter{
  \offinterlineskip\halign{\hfil$##$\cr
    \propto\cr\noalign{\kern2pt}\sim\cr\noalign{\kern-2pt}}}}}

\newcommand{\LtoS}{{S\leftarrow L}}
\newcommand{\LtoB}{{B\leftarrow L}}
\newcommand{\StoL}{{L\leftarrow S}}
\newcommand{\StoB}{{B\leftarrow S}}

\newcommand{\CtoSA}{{SA\leftarrow C}}
\newcommand{\SAtoC}{{C \leftarrow SA}}
\newcommand{\CtoLA}{{LA \leftarrow C}}
\newcommand{\LAtoC}{{C \leftarrow LA}}

\let\baraccent=\= 
\renewcommand{\=}[1]{\stackrel{#1}{=}} 

\newcommand{\thmref}[1]{\hyperref[#1]{Theorem~\ref{#1}}}
\newcommand{\lemmaref}[1]{\hyperref[#1]{Lemma~\ref{#1}}}
\newcommand{\corollaryref}[1]{\hyperref[#1]{Corollary~\ref{#1}}}
\newcommand{\propref}[1]{\hyperref[#1]{Proposition~\ref{#1}}}
\newcommand{\figref}[1]{\hyperref[#1]{Fig.~\ref{#1}}}
\newcommand{\figaref}[1]{\hyperref[#1]{Fig.~\ref{#1}a}}
\newcommand{\figbref}[1]{\hyperref[#1]{Fig.~\ref{#1}b}}
\newcommand{\figcref}[1]{\hyperref[#1]{Fig.~\ref{#1}c}}
\newcommand{\figdref}[1]{\hyperref[#1]{Fig.~\ref{#1}d}}
\newcommand{\figeref}[1]{\hyperref[#1]{Fig.~\ref{#1}e}}
\renewcommand{\eqref}[1]{\hyperref[#1]{Eq.~(\ref{#1})}}
\newcommand{\secref}[1]{\hyperref[#1]{Sec.~\ref{#1}}}
\newcommand{\eqsref}[2]{\hyperref[#1]{Eqs.~(\ref{#1})-(\ref{#2})}}
\newcommand{\appref}[1]{\hyperref[#1]{Appx.~\ref{#1}}}

\newcommand{\isometric}{Consider an isometric quantum code defined by $\mE_{\LtoS}$. Consider  physical Hamiltonian $H_S$,  logical Hamiltonian $H_L$, and noise channel $\mN_S$. Suppose the HKS condition is satisfied. }
\newcommand{\nonisometric}{Consider a quantum code defined by $\mE_{\LtoS}$. Consider  physical Hamiltonian $H_S$,  logical Hamiltonian $H_L$, and noise channel $\mN_S$. Suppose the HKS condition is satisfied.  }

\begin{document}

\title{
Quantum error correction meets continuous symmetries: 
fundamental trade-offs and case studies}

\author{Zi-Wen Liu}
\thanks{\href{mailto:zwliu0@mail.tsinghua.edu.cn}{zwliu0@mail.tsinghua.edu.cn} (Zi-Wen Liu); \href{mailto:sisi.zhou26@gmail.com}{sisi.zhou26@gmail.com} (Sisi Zhou); The author names are in alphabetical order. }
\affiliation{Perimeter Institute for Theoretical Physics, Waterloo, Ontario N2L 2Y5, Canada}
\affiliation{Yau Mathematical Sciences Center, Tsinghua University, Beijing 100084, China}
\author{Sisi Zhou}
\thanks{\href{mailto:zwliu0@mail.tsinghua.edu.cn}{zwliu0@mail.tsinghua.edu.cn} (Zi-Wen Liu); \href{mailto:sisi.zhou26@gmail.com}{sisi.zhou26@gmail.com} (Sisi Zhou); The author names are in alphabetical order. }
\affiliation{Perimeter Institute for Theoretical Physics, Waterloo, Ontario N2L 2Y5, Canada}
\affiliation{Institute for Quantum Information and Matter, California Institute of Technology, Pasadena, CA 91125, USA}
\affiliation{Pritzker School of Molecular Engineering, The University of Chicago, Illinois 60637, USA}
\date{\today}


\begin{abstract}
    
We systematically study the fundamental competition between quantum error correction (QEC) and continuous symmetries, two key notions in quantum information and physics, in a quantitative manner. Three meaningful measures of approximate symmetries in quantum channels and, in particular, QEC codes, based on covariance violation over the entire symmetry group, covariance violation at a local point (closely related to quantum Fisher information), and the violation of charge conservation, respectively, are introduced and studied.  Each measure induces a corresponding characterization of approximately covariant codes.  We explicate a host of different ideas and techniques that enable us to derive various forms of trade-off relations between the  QEC accuracy and all symmetry  measures.  More specifically, we introduce two frameworks for understanding and analyzing the trade-offs, based on the notions of charge fluctuation and gate implementation error (which may be of interest {in their own rights}) respectively, and employ methods including the Knill--Laflamme conditions as well as quantum metrology and quantum resource theory for the derivation. From the perspective of fault-tolerant quantum computing, our bounds on symmetry violation indicate limitations on the precision or density of transversally implementable logical gates for general QEC codes, refining the Eastin--Knill theorem. To exemplify nontrivial approximately covariant codes and understand the achievability of the above fundamental limits, we analyze two explicit types of codes: 1) a parametrized extension of the thermodynamic code, which gives a construction of a code family that continuously interpolates between exact QEC and exact symmetry, and 2) the quantum Reed--Muller codes, which represents a prominent example of approximately covariant exact QEC code.  We show that both codes can saturate the scaling of the bounds for group-global covariance and charge conservation asymptotically, indicating the near-optimality of both our bounds and codes. 
\end{abstract}

\maketitle

\newpage 
\tableofcontents

\newpage

\section{Introduction}

Quantum error correction (QEC) is one of the most important and widely studied ideas in quantum information processing~\cite{shor1995scheme,nielsen2002quantum,gottesman2010introduction,lidar2013quantum}.  The spirit of QEC is to protect  quantum information against noise and errors by suitably encoding logical quantum systems into quantum codes living in a larger physical Hilbert space. Since quantum systems are highly susceptible to noise effects such as decoherence so that errors easily occur, it is clear that  QEC is of vital importance to the practical realization of quantum computing and other quantum technologies.  Interestingly, besides the enduring efforts on the study of QEC and quantum codes for quantum information processing purposes,  in recent years, they are also found to play fundamental roles in many important physical scenarios  in e.g., holographic quantum gravity \cite{almheiri2015bulk,pastawski2015holographic} and many-body physics \cite{Kitaev2003,zeng2019quantum,brandao2019quantum,PhysRevLett.120.200503}, and have consequently drawn great interest in physics.

When considering the practical implementation  of QEC as well as its connections to physical problems, it is important to take symmetries and conservation laws into account as they are ubiquitous in physical systems.  More explicitly, symmetries may constrain the encoders in the way that they must be covariant with respect to the symmetry group, i.e., commute with certain representations of group actions, generating the so-called \emph{covariant codes} \cite{hayden2017error,faist2019continuous,woods2019continuous,wang2019quasi}.
A principle of fundamental significance in both quantum information and physics is that (finite-dimensional) covariant codes for continuous symmetries (mathematically modeled by Lie groups)\footnote{In what follows, we assume that the associated symmetry group is  continuous and that the relevant Hilbert spaces are finite-dimensional when using the term ``covariant codes''.} are in a sense fundamentally 
incompatible with exact QEC~\cite{hayden2017error,eastin2009restrictions}.
A well known no-go theorem that unfolds this principle from a quantum computation perspective is the Eastin--Knill theorem \cite{eastin2009restrictions}, which indicates that any  QEC code covariant with respect to any continuous symmetry group in the sense that the logical group actions are mapped to \emph{transversal} physical actions (that are tensor products on physical subsystems) cannot correct {local} errors perfectly.  An intuitive explanation of this phenomenon is that, due to the conservation laws and transversality, physical subsystems necessarily contain logical charge information that gets leaked into the environment upon errors, so that the perfect recovery of logical information is prohibited.
 Crucially, transversal actions are highly desirable for the ``fault tolerance''~\cite{shor1996fault,nielsen2002quantum,gottesman2010introduction,lidar2013quantum} of practical quantum computation schemes because they do not spread errors across physical subsystems within each code block. 
It is also worth noting that the transversality
property is widely important in physics as a fundamental feature of internal symmetries in many-body scenarios.  More specifically, they are normally generated by sums of disjoint local charge observables, or in particular, on-site (transversal with respect to sites).  Note that whether the symmetries are on-site is linked to whether they can be gauged or are anomaly-free, which plays important roles in the physics of quantum many-body systems and field theories \cite{Wen13}.  In AdS/CFT, transversality also plays fundamental roles \cite{harlow2018symmetries,PhysRevLett.117.021601,MaySorceYoshida}.

Due to the Eastin--Knill theorem, unfortunately, it is impossible to find an exact QEC code that implements a universal set of gates transversally, or namely achieves the full power of quantum computation while maintaining transversality.   However, it may still be feasible to perform QEC approximately under these constraints, and a natural task is then to characterize the optimal degree of accuracy.  Recently, several such  bounds on the QEC accuracy achievable by covariant codes (which give rise to ``robust'' or ``approximate'' versions of the Eastin--Knill theorem) as well as explicit constructions of near-optimal covariant codes are found using many different techniques and insights from various areas in quantum information \cite{faist2019continuous,woods2019continuous,wang2019quasi,kubica2020using,zhou2020new,yang2020covariant,fang2020no,tajima2021symmetry,wang2021theory,KongLiu21:random}, showcasing the fundamental nature of the problem.  
Remarkably, covariant codes have also found interesting applications to several important areas in physics already, including quantum many-body physics \cite{brandao2019quantum,PhysRevLett.120.200503,wang2019quasi}, AdS/CFT correspondence \cite{harlow2018symmetries,harlow2018constraints,kohler2019toy,faist2019continuous,woods2019continuous}, and quantum information \cite{hayden2017error,woods2019continuous,KongLiu21:random}.

These existing studies on covariant codes are mostly concerned with the precision of QEC under exact symmetry conditions.
Indeed, when symmetry principles arise, they are exactly respected by default.  However, especially for continuous symmetries, it is often important or even necessary to consider  approximate forms of symmetries or conservation laws in physical and practical scenarios. First of all, realistic quantum many-body systems are often dirty or defective so that the exact symmetry conditions and conservation laws could generally be violated to a certain extent. 
Furthermore, there are many important situations in physics where non-exact symmetries need to be considered for fundamental reasons. There are various symmetry breaking mechanisms that play key roles in wide-ranging physical scenarios including spontaneous symmetry breaking, anomalies, and non-renormalizable effects \cite{symmetry-breaking}. In particle physics, many important symmetries are known to be only approximate \cite{Witten2018}.   
More notably,  it has long been believed that global
symmetries cannot be exact in a unified theory of quantum
mechanics and gravity \cite{Misner1957,Giddings1988,PhysRevD.52.912,Arkani_Hamed_2007,BanksSeiberg11,Witten2018} (justified in more concrete terms in AdS/CFT \cite{harlow2018constraints,harlow2018symmetries}).  
Considering the need for large quantum systems to boost the advantages of quantum technologies and also the broad connections between QEC and physics, it would be important and fruitful to have a quantitative theory of QEC codes with approximate symmetries, or \emph{approximately covariant} codes. 
For example,  given that the QEC accuracy of exactly covariant codes is limited, one may wonder whether for codes that achieve exact QEC there are ``dual'' bounds on the degree of symmetry or covariance.
It is particularly worth noting that the no-global-symmetry arguments in AdS/CFT indeed have deep connections to covariant codes, and in particular this question  \cite{faist2019continuous,harlow2018symmetries}.
However, our understanding of approximate symmetries, especially characterizations and applications on a quantitative level, is very limited to date.

Our work aims to establish a quantitative theory of the interplay between the degree of continuous symmetries and QEC accuracy, which in particular allows us to understand symmetry violation in exact QEC codes. (Note that our discussion here mainly proceeds in terms of the most fundamental $U(1)$ symmetry which is sufficient to reveal the key phenomena.) 
To this end, we first formally define three different meaningful measures of  symmetry violation,  in terms of the violation of covariance conditions globally over the entire symmetry group or locally at a specific point in the group, and  the violation of charge conservation respectively, which induce corresponding quantitative notions of approximately covariant codes. Our main results are a series of trade-off bounds between QEC accuracy and the above different symmetry measures under a general condition called \emph{Hamiltonian-in-Kraus-span (HKS) condition} which
subsumes transversality in our setup, each of which may suit certain scenarios the best. (For readers' convenience,  we provide in \appref{app:comparison} a table that identifies the key theorems and summarizes their respective strengths and weaknesses.) We introduce two concepts---charge fluctuation and gate implementation error---each providing a framework for analyzing the QEC-symmetry trade-off and could be useful in their own rights.  Furthermore, our derivations feature ideas and techniques from several different fields. More explicitly, various different forms of the trade-off relations are derived by analyzing the ``perturbation'' of the Knill--Laflamme conditions~\cite{knill1997theory,beny2010general}, as well as by leveraging insights and techniques from the fields of quantum metrology~\cite{giovannetti2011advances,zhou2020theory} and quantum resource theory~\cite{chitambar2019quantum,marvian2020coherence,FangLiu19:nogo,fang2020no}. 
Our theory provides a complete understanding of the transition between exact QEC and exact symmetry.
On the exact symmetry end, the previous limits on covariant codes (often referred to as ``approximate Eastin--Knill theorems''~\cite{faist2019continuous,woods2019continuous,kubica2020using,zhou2020new}) are recovered, while the exact QEC end provides new lower bounds on various forms of symmetry violation for the commonly studied exact codes. In particular, we use our symmetry bounds to derive fundamental limitations on the set of transversally implementable logical gates for general QEC codes, which represent a new type of improvement of the Eastin--Knill theorem and 
{apply more broadly than previous results along a similar line about stabilizer codes in Refs.~\cite{zeng2011transversality,bravyi2013classification,pastawski2015fault,anderson2016classification,jochym2018disjointness}, advancing our understanding of fault tolerance.} 
Then, to solidify our general theory, we present case studies on two explicit code constructions, which can be seen as examples of approximately covariant codes that exhibit certain key features, as well as upper bounds (achievability results) that help understand how strong our fundamental limits are.  First, we construct a parametrized code family that interpolates between the two ends of exact QEC and exact symmetry and exhibits a full trade-off between QEC and symmetry, by modifying the so-called thermodynamic code~\cite{brandao2019quantum,faist2019continuous}.  In the second case study, we analyze the quantum Reed--Muller codes which exhibit nice structures and features and, in particular, have been widely applied for the transversal implementation of certain non-Clifford gates and magic state distillation~\cite{bravyi2012magic,anderson2014fault,haah2018codes,hastings2018distillation}.  We find that in both cases the codes can almost saturate the bounds on global covariance and charge conservation (up to constant factors) asymptotically, that is, both the code constructions and bounds are nearly optimal.

Here we present the study in a rigorous and comprehensive manner.  In particular, this work contains all technical details of the derivation, thorough discussions of all different approaches, and many additional results. 

This paper is organized as follows. First, in \secref{sec:def}, we 
review the formalism of QEC and the incompatibility between QEC and continuous symmetries, and also 
 formally define the accuracy of approximate QEC codes as well as the different quantitative charaterizations of approximate continuous symmetries associated with QEC codes that will be considered.  In \secref{sec:global-1} and \secref{sec:global-2}, we introduce the two frameworks based on the notions of charge fluctuation and gate implementation error respectively, under which we discuss a series of different approaches to deriving the trade-off relations between the QEC inaccuracy and the group-global covariance violation. Then in \secref{sec:transversal}, we specifically discuss the application to fault-tolerant quantum computing, deriving general restrictions on the transversally implementable logical gates in QEC codes from the results above.  Afterwards, in \secref{sec:local}, we present our results on the trade-off relations between QEC inaccuracy and group-local symmetry measures including the group-local covariance violation and the charge conservation violation.   After the above discussion of fundamental limits, in \secref{sec:case-study} we  study  the modified thermodynamic code and  quantum Reed--Muller codes, which gives concrete examples of nearly optimal approximately covariant approximate codes in certain cases.   Finally, in \secref{sec:discussion} we summarize our study, and discuss important open problems and future directions.

\textbf{Note that this long paper is a companion paper of Ref.~\cite{short} with extended results and technical details which focuses on the most representative results and the physical motivation behind this study. It was published as the Supplementary Information to Ref.~\cite{short}.}

\section{Approximately covariant approximate QEC codes: Quantitative characterizations of QEC and symmetry}
\label{sec:def}

Here we formally define the quantitative measures of QEC and symmetry that will be used in our study. Specifically, we first overview the notions of QEC and covariant codes and discuss how to quantify the deviation of general quantum codes from them. In particular, we will define the QEC inaccuracy which quantifies the QEC capability of a quantum code under specific noise. We will also define a measure of group-global symmetry which quantifies the approximate covariance of a code over the entire $U(1)$ group and two measures of group-local symmetry which are linked to the covariance of a code at an exact point in $U(1)$. The trade-off relations between QEC and these symmetry measures will be thoroughly studied in later sections. 

Note that in this work, we extensively use distance metrics defined based on the purified distance \cite{gilchrist2005distance,tomamichel2015quantum}, which are well-behaved and commonly used in the quantum information literature.  
Notably, the choice of purified distance directly relates the local covariance violation to the well known quantum Fisher information (QFI)~\cite{helstrom1976quantum,holevo1982probabilistic,hubner1992explicit,sommers2003bures}. 
In principle, one may also consider other metrics. We shall also discuss the situations where one uses the diamond distance~\cite{watrous2018theory}, another standard channel distance measure.

\subsection{Approximate quantum error correction}
\label{sec:def-QEC}

QEC functions by encoding the logical quantum system in some quantum code living in a larger physical system with redundancy, so that a limited number of errors can be corrected to recover the original logical information. A quantum code is defined by an encoding quantum channel $\mE_{\LtoS}$ from a logical system $L$ to a  physical system $S$, and it perfectly protects the logical information against a physical noise $\mN_{S}$ if and only if there exists a recovery channel $\mR_{\LtoS}$ such that 
\begin{equation}
\mR_{\StoL} \circ \mN_{S} \circ \mE_{\LtoS} = \id_L. 
\end{equation}
In particular, when $\mE_{\LtoS}$ is isometric, $\mN_{S}(\cdot) = \sum_{i=1}^r K_{S,i}(\cdot)K_{S,i}^\dagger$ and $\Pi$ is the projection onto the code subspace in the physical system, such a recovery channel exists if and only if the Knill--Laflamme (KL) conditions, $\forall i,j,\,\Pi K_{S,i}^\dagger K_{S,j} \Pi \propto \Pi$~\cite{knill1997theory}, hold. 

In many scenarios, a quantum code is still useful in protecting quantum information when it only achieves approximate QEC, namely, $\mR_{\StoL} \circ \mN_{S} \circ \mE_{\LtoS}$ is close to but not exactly equal to $\id_L$. To characterize the inaccuracy of an approximate QEC code, we will use the channel fidelity and the Choi channel fidelity,  defined by 
\begin{gather}
f(\Phi_1,\Phi_2) = \min_{\rho} f\left((\Phi_1\otimes \id)(\rho),(\Phi_2\otimes \id)(\rho)\right),\\
\overline{f}(\Phi_1,\Phi_2) = f\left((\Phi_1\otimes \id)(\Psi),(\Phi_2\otimes \id)(\Psi)\right),
\end{gather}
respectively,
where $f(\rho,\sigma) = \trace(\sqrt{\rho^{1/2}\sigma\rho^{1/2}})$ is the fidelity of quantum states, $\ket{\Psi} = \frac{1}{\sqrt{d}}\sum_{i=1}^{d}\ket{i}\ket{i}$ is the maximally entangled state and $\Psi = \ket{\Psi}\bra{\Psi}$. Here the inputs $\rho$ and $\Psi$ lie in a bipartite system consisting of the original system $\Phi_{1,2}$ acting on and a reference system as large as the original. Correspondingly, one can define the purified distance of states $P(\rho,\sigma) = \sqrt{1 - f(\rho,\sigma)^2}$, the purified distance of channels $P(\Phi_1,\Phi_2) = \sqrt{1 - f(\Phi_1,\Phi_2)^2}$ and the Choi purified distance of channels $\barP(\Phi_1,\Phi_2) = \sqrt{1 - \overline{f}(\Phi_1,\Phi_2)^2}$~\cite{gilchrist2005distance,tomamichel2015quantum,LiuWinter19}. 
The  \emph{(worst-case) QEC inaccuracy} and the \emph{Choi QEC inaccuracy} for approximate QEC codes are then defined as 
\begin{gather}
\epsilon(\mN_S,\mE_{\LtoS}) := \min_{\mR_{\StoL}} P(\mR_{\StoL}\circ\mN_S \circ \mE_{\LtoS},\id_L),\\
\bepsilon(\mN_S,\mE_{\LtoS}) := \min_{\mR_{\StoL}} \barP(\mR_{\StoL}\circ\mN_S \circ \mE_{\LtoS},\id_L), 
\end{gather}
respectively.
The Choi inaccuracy reflects the average-case behavior in the sense that $\bepsilon = \sqrt{\frac{d_L+1}{d_L}}\epsilon_\avg$~\cite{horodecki1999general,nielsen2002simple}, where \sloppy $\epsilon_\avg := \sqrt{{1 - \max_{\mR_{\StoL}}\int d\psi f^2(\psi,\mR_{\StoL}\circ \mN_S \circ \mE_{\LtoS}(\psi))}}$ in which the integral is over the Haar random pure logical states. For simplicity, we will not explicitly write down the arguments of $\epsilon$ or $\bepsilon$ (and of many other measures defined later) when they are unambiguous.

In the above, we used the channel purified distances as channel distance measures. 
As mentioned, we may also consider the  the diamond distance $D_\diamond(\Phi_1,\Phi_2)$ induced by the diamond norm of channels~\cite{diamond,watrous2018theory}:
\begin{align}
D_\diamond(\Phi_1,\Phi_2):=\,& \max_\rho \frac{1}{2} \norm{(\Phi_1\otimes \id)(\rho)-(\Phi_2\otimes \id)(\rho)}_1\\
=\,& \max_{\ket{\psi}} \frac{1}{2} \norm{(\Phi_1\otimes \id)(\ket{\psi}\bra{\psi})-(\Phi_2\otimes \id)(\ket{\psi}\bra{\psi})}_1. 
\end{align}
where $\norm{\cdot}_1$ is the nuclear (trace) norm. 
Naturally,  the diamond distance version of QEC inaccuracy is defined as
\begin{equation}
\epsilon_\diamond := \min_{\mR_{\StoL}} D_\diamond(\mR_{\StoL}\circ\mN_S \circ \mE_{\LtoS},\id_L).     
\end{equation}
It is easy to see that lower bounds on $\epsilon$ (that we derive below) directly indicate lower bounds on  $\epsilon_\diamond$.
According to the Fuchs--van de Graaf inequality $1 - f(\rho,\sigma) \leq \frac{1}{2}\norm{\rho - \sigma}_1$~\cite{fuchs1999cryptographic}, we have
\begin{align}
\label{eq:fuchs}
D_\diamond(\Phi_1,\Phi_2) \geq 1 - f(\Phi_1,\Phi_2)
\geq \frac{1}{2}P(\Phi_1,\Phi_2)^2. 
\end{align}
In the case of our interest where the second channel is the identity, the above inequality can be further improved using $1 - f(\rho,\ket{\psi}\bra{\psi})^2 \leq \frac{1}{2}\norm{\rho - \ket{\psi}\bra{\psi}}_1$:
\begin{align}
D_\diamond(\Phi_1,\id) \geq 1 - f(\Phi_1,\id)^2 = P(\Phi_1,\id)^2. 
\end{align}
Therefore, $\epsilon_\diamond \geq \epsilon^2$.

\subsection{Measuring approximate symmetries of QEC codes}
\label{sec:def-symmetry}

Symmetries of quantum codes manifest themselves in the covariance of the encoder with respect to symmetry transformations.
For the case of current interest, the symmetry transformations on the logical and physical systems are, respectively, $U_{L,\theta} = e^{-iH_L\theta}$ generated by a logical Hamiltonian (charge observable) $H_L$, and $U_{S,\theta} = e^{-iH_S\theta}$ generated by a physical Hamiltonian (charge observable) $H_S$\footnote{Here $H_L$ and $H_S$ are generators of $U(1)$ representations, or ``charge observables'', and should not be confused with the intrinsic system Hamiltonians governing the system dynamics.}, both representations of the $U(1)$ Lie group periodic with a common period $\tau$.   
The transversality property of symmetry transformations (gate actions) corresponds to the 1-local form of $H_S$, namely, $H_S=\sum_{l=1}^n H_{S_l}$ where each term $H_{S_l}$ acts locally on physical subsystem $S_l$.
We say a quantum code is covariant (with respect to such $U(1)$ representations given by $H_L$ and $H_S$), if
\begin{equation}
\label{eq:covariant}
\mU_{S,\theta} \circ \mE_{\LtoS,\theta} = \mE_{\LtoS,\theta} \circ \mU_{L,\theta},\quad \forall \theta \in \bR. 
\end{equation}
The definitions of covariant codes can be easily extended to general compact Lie groups \cite{faist2019continuous,woods2019continuous}. We also assume $H_L$ and $H_S$ to be both non-trivial, i.e., not a constant operator. Note that applying constant shifts on $H_L$ and $H_S$ do not change the definition of \eqref{eq:covariant} and we will often use this property  below.  

As mentioned, the covariance of quantum codes is often incompatible with their error-correcting properties and approximate notions of covariance may play important roles in wide-ranging  scenarios. For example, here the Eastin--Knill theorem indicates that codes that can perfectly correct local noise cannot  simultaneously be covariant with respect to non-trivial 1-local $H_S$ \cite{eastin2009restrictions}. More generally, exact QEC is known to be incompatible with exact covariance as long as 
\begin{equation}
H_S \in {\rm span}\{K_{S,i}^\dagger K_{S,j},\,\forall i,j\},    
\end{equation} 
which we refer to as the \emph{Hamiltonian-in-Kraus-span (HKS) condition}, holds~\cite{zhou2020new,zhou2020theory}. 
The HKS condition   
holds for many typical scenarios, including the one mentioned above where $\mN_S$ represents single-erasure noise {(where one subsystem chosen uniformly at random is erased)} and $H_S$ is 1-local. When the HKS condition does not hold, examples of exactly covariant QEC codes exist, e.g., when $\mN_{S} = \id$ (noiseless dynamics), when $H_S$ is a Pauli-X operator and $\mN_S$ is dephasing noise~\cite{kessler2014quantum,arrad2014increasing}, and when $\mN_S$ is single-erasure noise but $H_S$ is 2-local~\cite{gottesman2016quantum}. We shall assume that the HKS condition holds for the quantum codes considered in our work. We also emphasize that there exist examples of exact QEC codes covariant with respect to discrete symmetry groups~\cite{hayden2017error}, so the assumption of continuous groups is important.

Besides quantum computation, approximately symmetries and covariant codes are potentially useful in quantum gravity and condensed matter physics, as discussed in the main text.  To formally characterize and study approximate covariance, an important first step is to find reasonable ways to quantify it.  We now do so.

\subsubsection{Group-global covariance violation}
\label{sec:group-global}

The first, most important type of measure is based on the global covariance violation over the entire symmetry group.   Codes that are approximately covariant with respect to  $H_L$ and $H_S$ in such a global sense should have small covariance violation for all $\theta$.  We define the \emph{group-global\footnote{We shall refer to ``group-global'' and ``group-local'' as ``global'' and ``local'', respectively,  for simplicity, as is common in e.g.~estimation theory after their definitions. They should not be confused with the geometric notions commonly used in physical contexts.} covariance violation} and the \emph{Choi group-global covariance violation} by 
\begin{gather}
\delta_\group :=  \max_\theta P(\mU_{S,\theta}\circ\mE_{\LtoS},\mE_{\LtoS}\circ\mU_{L,\theta}),\\
\bdelta_\group := \max_\theta \barP(\mU_{S,\theta}\circ\mE_{\LtoS},\mE_{\LtoS}\circ\mU_{L,\theta}), 
\end{gather}
respectively.
Intuitively, they measure the maximum deviation of the encoding channel $\mE_{\LtoS}$ from the exact covariance condition \eqref{eq:covariant} in the entire symmetry group.
It is known that $\delta_\group$ and $\epsilon$ cannot be simultaneously zero in non-trivial situations, and previous works~\cite{faist2019continuous,woods2019continuous,wang2019quasi,kubica2020using,zhou2020new,yang2020covariant,tajima2021symmetry} mostly focus on deriving lower bounds on $\epsilon$ for exactly covariant codes ($\delta_\group = 0$). We will present bounds that involve $\delta_\group$ which reveal the trade-off between QEC and global covariance, derived via two notions we introduce called the charge fluctuation and gate implementation error. This extends the scope of previous consideration to general codes including exact QEC codes.

Similar to the case of QEC inaccuracy, we can also consider the diamond distance and define 
\begin{equation}
\delta_{\group,\diamond} :=  \max_\theta D_\diamond(\mU_{S,\theta}\circ\mE_{\LtoS},\mE_{\LtoS}\circ\mU_{L,\theta}).     
\end{equation}
Again, lower bounds on $\delta_{\group,\diamond}$ that we derive below directly indicate lower bounds on $\delta_{\group}$. 
Using \eqref{eq:fuchs}, we directly see that $\delta_{\group,\diamond} \geq \delta_{\group}^2/2$.  In particular, when $\mE_{\LtoS}$ is isometric, we have $\delta_{\group,\diamond} = \delta_{\group}$, using the fact that $\frac{1}{2}\norm{\ket{\psi_1}\bra{\psi_1} - \ket{\psi_2}\bra{\psi_2}}_1 = P(\ket{\psi_1}\bra{\psi_1},\ket{\psi_2}\bra{\psi_2})$.

\subsubsection{Group-local (point) covariance violation}

One may wonder if the incompatibility between QEC and continuous symmetries can be relieved when we relax the requirement from exact global covariance to exact local covariance, i.e., when we require only the code covariance for $\theta$ inside a small neighborhood of a point $\theta_0$, satisfying $\mU_{S,\theta_0} \circ \mE_{\LtoS} = \mE_{\LtoS} \circ \mU_{L,\theta_0}$. 
Unfortunately, the no-go results also extend to the local case, meaning that a non-trivial QEC code cannot be exactly covariant even in an arbitrarily small neighborhood of $\theta_0$. Without loss of generality, we assume $\theta_0 = 0$ because we can always redefine $\mE_{\LtoS} \circ \mU_{L,\theta_0}$ to be the new encoding channel such that the code is covariant at $\theta = 0$. To characterize the local covariance, we introduce the \emph{group-local (point) covariance violation} defined by 
\begin{align}
\delta_\point
:= \sqrt{{2\partial_\theta^2 P(\mU_{S,\theta}\circ\mE_{\LtoS}\circ\mU_{L,\theta}^\dagger,\mE_{\LtoS})^2\big|_{\theta = 0}}} = \sqrt{F(\mU_{S,\theta}\circ\mE_{\LtoS}\circ\mU_{L,\theta}^\dagger)\big|_{\theta = 0}}.
\end{align}
Here $F(\Phi_\theta)$ is the quantum Fisher information (QFI) defined using the second order derivative of the purified distance $F(\Phi_\theta) = 2\partial_{\theta'}^2 P(\Phi_\theta,\Phi_{\theta'})^2\big|_{\theta' = \theta}$ which characterizes the amount of information of $\theta$ one can extract from $\Phi_\theta$ around point $\theta$~\cite{fujiwara2008fibre}. Correspondingly, the QFI of quantum states is defined by $F(\rho_\theta) = 2\partial_{\theta'}^2 P(\rho_\theta,\rho_{\theta'})^2\big|_{\theta' = \theta}$~\cite{hubner1992explicit,braunstein1994statistical} which characterizes the amount of information of $\theta$ one can extract from $\rho_\theta$ around point $\theta$ and we have $F(\Phi_\theta) = \max_\rho F((\Phi_\theta\otimes\id)(\rho))$.  Note that the QFI defined here using the purified distance is usually called the SLD QFI and there are other types of QFIs, e.g., the RLD QFI~\cite{yuen1973multiple,hayashi2011comparison,katariya2020geometric} which we will encounter later in \secref{sec:global-2-res}. 
When $\delta_\point = 0$, the code is locally covariant up to the lowest order of $d\theta$. We shall see later that for any $\delta_\point < \Delta H_L$ (we will use $\Delta(\cdot)$ to denote the difference between the maximum and minimum eigenvalues of $(\cdot)$), there is a non-trivial lower bound on $\epsilon$, leading to a trade-off relation between QEC and local covariance.

\subsubsection{Charge conservation violation}

The correspondence between symmetries and conservation laws is a landmark result of modern physics. Inspired by this correspondence, we can define another intuitive measure of the symmetry violation by the degree of charge deviation.
It can be shown that for an isometric encoding channel $\mE_{\LtoS}$ the covariance condition \eqref{eq:covariant} is equivalent to 
\begin{equation}
\label{eq:isometric-charge-relation}
    (\mE_{\LtoS})^\dagger(H_S) = H_L-\nu\id, 
\end{equation} for some $\nu\in\mathbb{R}$, 
where $\mE^\dagger$ is the dual channel of $\mE$ satisfying $\trace(H\mE(\rho)) = \trace(\mE^\dagger(H)\rho)$ for any $H$ and $\rho$. Since $H_S$ and $H_L$ represent the charge observables in the physical and logical systems, \eqref{eq:isometric-charge-relation} implies that the eigenstates of $H_L$ are mapped to the corresponding eigenstates of $H_S$ after the encoding operation~\cite{faist2019continuous}, indicating the charge conservation nature of the encoding map.
The charge conservation law can also be understood through the relation $\trace(H_S \mE_{\LtoS}(\rho)) = \trace(H_L\rho)-\nu$ for any $\rho$, where $\nu$ represents a universal constant offset in the charge.
To measure the degree to which the charge conservation law is violated, 
we consider the following quantity which we call the \emph{charge conservation violation} (also defined in Ref.~\cite{faist2019continuous}): 
\begin{equation}
\delta_\charge
:= \Delta\left(H_L - (\mE_{\LtoS})^\dagger(H_S)\right). 
\end{equation}
Note again that $\Delta(\cdot)$ denotes the difference between the maximum and minimum eigenvalues of $(\cdot)$. {It can be easily verified that $\delta_\charge/2$ is equal to the difference between physical and logical charges, formally given by $\min_{\nu\in\bR} \max_\rho| \trace(H_S\mE_{\LtoS}(\rho))-\trace ((H_L-\nu \id)\rho)|$ (a constant offset on the definitions of charges is allowed).} 
For general CPTP encoding maps, $\delta_\charge$ is not always zero for exactly covariant codes~\cite{cirstoiu2020robustness}, unlike $\delta_\group$ and $\delta_\point$. However, for isometric encoding we always have the following relation between $\delta_\point$ and $\delta_\charge$: 
\begin{proposition}
\label{prop:point-charge}
When $\mE_{\LtoS}$ is isometric, $\delta_\point \geq \delta_\charge$. 
\end{proposition}
\begin{proof}
Suppose $\mE_{\LtoS}(\cdot) = W(\cdot)W^\dagger$ where $W$ is isometric. Then 
\begin{align}
    \delta_\charge = \Delta\left({H_L - W^\dagger H_S W}\right),
\end{align}
 and  
\begin{align}
    (\delta_\point)^2 = F(\mU_{S,\theta}\circ\mE_{\LtoS}\circ\mU_{L,\theta}^\dagger)\big|_{\theta = 0} = \max_{\ket{\psi}} F(U_{S,\theta} W U_{L,\theta}^\dagger \ket{\psi})\big|_{\theta = 0}.
\end{align}
Let 
\begin{align}
\ket{\psi_\theta} = U_{S,\theta} W U_{L,\theta}^\dagger \ket{\psi} = e^{-iH_S\theta}We^{iH_L\theta}\ket{\psi}.
\end{align}
Then 
\begin{align}
    \ket{\partial_\theta\psi_\theta} \!=\! e^{-iH_S\theta}We^{iH_L\theta}iH_L\ket{\psi} 
\!-\! iH_S e^{-iH_S\theta}We^{iH_L\theta}\ket{\psi}\!. 
\end{align}
The QFI for pure states is given by
$F(\ket{\psi_\theta}) = 4(\braket{\partial_\theta\psi_\theta|\partial_\theta\psi_\theta} - \abs{\braket{\partial_\theta\psi_\theta|\psi_\theta}}^2 )$~\cite{braunstein1994statistical}.  Since 
\begin{gather}
\begin{split}
\braket{\partial_\theta\psi_\theta|\partial_\theta\psi_\theta}\big|_{\theta = 0} &= \bra{\psi}(W H_L - H_S W)^\dagger  (W H_L - H_S W)\ket{\psi}\\
&\geq \bra{\psi}(W H_L - H_S W)^\dagger WW^\dagger (W H_L - H_S W)\ket{\psi},  
\end{split}\\
\braket{\psi_\theta|\partial_\theta\psi_\theta}\big|_{\theta = 0}
= \bra{\psi}(H_L - W^\dagger H_S W)\ket{\psi}, 
\end{gather}
we have 
\begin{align}
(\delta_\point)^2 =& \max_{\psi} F(\ket{\psi_\theta})\big|_{\theta = 0} \\
\geq& \max_{\psi}  4\bra{\psi}( H_L - W^\dagger H_S W)^\dagger   (H_L - W^\dagger H_S W)\ket{\psi} - 4\abs{\bra{\psi}(H_L - W^\dagger H_S W)\ket{\psi}}^2 = (\delta_\charge)^2,
\end{align}
proving the result. 
\end{proof}
We shall see later that, similar to the situation of local covariance violation, for any $\delta_\charge < \Delta H_L$, there is a non-trivial lower bound on $\epsilon$. We refer to both $\delta_\point$ and $\delta_\charge$ as local symmetry measures because their values only depend on the approximate covariance of a code in the neighborhood of the point $\theta = 0$. Note that both $\delta_\point$ and $\delta_\charge$ have the same unit as the charges while $\delta_\group$ is dimensionless, i.e., after replacing $H_S$ and $H_L$ with $c H_S$ and $c H_L$ for some constant $c$, both $\delta_\point$ and $\delta_\charge$ are changed to $c \delta_\point$ and $c \delta_\charge$ while $\delta_\group$ is unchanged.

\subsubsection{Remarks on non-compact groups and infinite-dimensional codes}

In the above discussion, we assumed compact Lie groups and finite-dimensional quantum codes. Here we remark on possible extensions to non-compact Lie groups and infinite-dimensional codes. 

First, note that our definitions of $\delta_\point$ and $\delta_\charge$ can be naturally extended to the situations of non-compact groups where  $H_S$ and $H_L$ are arbitrary finite-dimensional Hermitian operators but the group transformations are not periodic, because their definitions only depend on the local geometry of the symmetry group.  For the global measure $\delta_\group$, we need to assume compact Lie groups, i.e.,~the physical and logical group transformations are both periodic with a common period.

Moreover, when the physical system $S$ is infinite-dimensional, one may naturally consider some  finite-dimensional truncation $\tilde H_S$ of $H_S$. The trade-off relations we derive below hold for $\tilde H_S$ and $H_L$, so when the truncation is suitably chosen we can still obtain nontrivial results that well indicate the behaviors of $H_S$. For example, when $\|(\mE_{\LtoS})^\dagger(H_S) - (\mE_{\LtoS})^\dagger(\tilde H_S)\|$ ($\norm{\cdot}$ is the spectral norm) is small, $\tilde H_S$ is a good substitute for $H_S$ in terms of the charge conservation violation~\cite{faist2019continuous}.

\section{Trade-off between QEC and global covariance: Charge fluctuation approach}
\label{sec:global-1}

In this section, we derive trade-off relations between the QEC inaccuracy $\epsilon$ and the global covariance violation $\delta_\group$ by connecting them to {a quantity which we call} the charge fluctuation $\chi$ (note that this notion is distinct from the charge conservation violation although they are in some way related  as will be discussed).  Our approach essentially proceed in two steps. First, we connect $\delta_\group$ and $\chi$ by providing a lower bound on $\delta_\group$ which depends on $\chi$.  Then we prove upper bounds on $\abs{\chi}$ in terms of the QEC inaccuracy $\epsilon$ using two different methods. The first one is based on analyzing the deviation of the approximate QEC code from the the KL conditions, which we call the KL-based method, and the second one is based on treating the problem as a channel estimation problem and using quantum metrology techniques.  These methods eventually lead to two types of trade-off bounds between QEC and global covariance. Note that we assume quantum codes are isometric throughout this section unless stated otherwise.  

\subsection{Bounding global covariance violation by charge fluctuation}
\label{sec:global-1-charge}

Consider a code defined by encoding isometry $\mE_{\LtoS}$.  We start by considering the situation where the code achieves exact QEC under the noise channel $\mN_{S}(\cdot) = \sum_i K_{S,i}(\cdot)K_{S,i}^\dagger$. 
According to the KL conditions,
\begin{equation}
\label{eq:KL}
\Pi K_{S,i}^\dagger K_{S,j} \Pi \propto \Pi, 
\end{equation}
where $\Pi$ is the projection onto the code subspace. 
In particular, let $\ket{0_L}$ and $\ket{1_L}$ be eigenstates corresponding to the largest and the smallest eigenvalues of $H_L$. (We do not specify the exact choices of $\ket{0_L}$ and $\ket{1_L}$ even when $H_L$ is degenerate, as long as they correspond to the largest and smallest eigenvalues, respectively.) Using \eqref{eq:KL}, we have 
\begin{equation}
\bra{0_L}(\mE_{\LtoS})^\dagger(K_{S,i}^\dagger K_{S,j})\ket{0_L} =
\bra{1_L}(\mE_{\LtoS})^\dagger(K_{S,i}^\dagger K_{S,j})\ket{1_L}. 
\end{equation}
Using the HKS condition $H_S \in {\rm span}\{K_{S,i}^\dagger K_{S,j},\,\forall i,j\}$, 
we must also have 
\begin{equation}
\label{eq:KL-H}
\bra{0_L}(\mE_{\LtoS})^\dagger(H_S)\ket{0_L} =
\bra{1_L}(\mE_{\LtoS})^\dagger(H_S)\ket{1_L}. 
\end{equation}
The incompatibility between QEC and symmetry could be understood through the incompatibility between \eqref{eq:KL-H} and \eqref{eq:isometric-charge-relation}.  \eqref{eq:KL-H} implies that $\braket{0_L|H_L|0_L} - \braket{1_L|H_L|1_L} = 0$ when the code achieves exact QEC, which contradicts with $\braket{0_L|H_L|0_L} - \braket{1_L|H_L|1_L} = \Delta H_L$ for exactly covariant codes from \eqref{eq:isometric-charge-relation}, implying the non-existence of exact QEC codes with exact covariance.

For general codes, we define the \emph{charge fluctuation}:
\begin{equation}
\chi := \bra{0_L}(\mE_{\LtoS})^\dagger(H_S)\ket{0_L} - 
\bra{1_L}(\mE_{\LtoS})^\dagger(H_S)\ket{1_L}, 
\end{equation}
Based on the discussion above, one can see that $\chi$ embodies the transition between exact QEC and exact symmetry quantitatively---when a code is close to being an exactly covariant code, $\chi$ cannot be too far away from $\Delta H_L$, and when a code is close to being an exact QEC code, $\chi$ cannot be too far away from $0$ (see an illustration in \figref{fig:charge}). Thus the trade-off relation between $\epsilon$ and $\delta_\group$ can be derived by connecting $\chi$ to each of them separately.

\begin{figure}[tb]
	\center
	\includegraphics[width=0.35\textwidth]{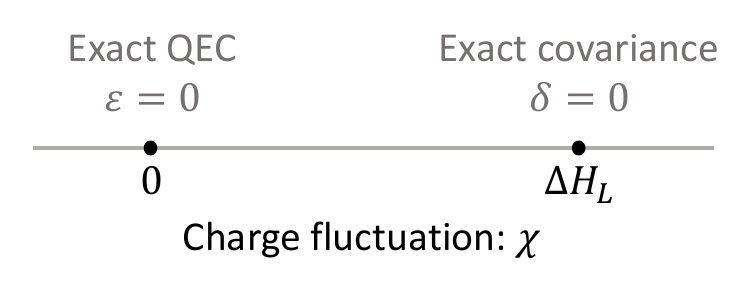}
	\caption{\label{fig:charge} 
	For exact QEC codes ($\epsilon = 0$) which satisfy the HKS condition, the charge fluctuation $\chi = 0$. For exactly covariant codes ($\delta_\group = 0$), the charge fluctuation $\chi = \Delta H_L$. The trade-off can be derived by investigating the relations between the distances of $\chi$ from $\Delta H_L$ and $0$, and the symmetry and QEC measures. 
	}
\end{figure}

We now derive the following lower bound on the global covariance violation $\delta_\group$ in terms of the charge fluctuation $\chi$, which directly connects $\delta_\group$ with $\chi$. 
{Note that, in this paper, ``$\gtrsim$'', ``$\lesssim$'', and ``$\simeq$'' mean ``$\geq$'', ``$\leq$'', and ``$=$'', respectively, up to the leading order.} 
\begin{proposition}
\label{prop:global-charge}
Consider an isometric quantum code defined by $\mE_{\LtoS}$.  Consider  physical Hamiltonian $H_S$,  logical Hamiltonian $H_L$, and noise channel $\mN_S$. Suppose the HKS condition is satisfied. Then when $\abs{\Delta H_L  - \chi} \leq \Delta H_S$, it holds that 
\begin{equation}
\label{eq:delta-charge}
\delta_\group \geq 
\min\left\{\frac{\sqrt{\abs{\Delta H_L - \chi}\left(\Delta H_S - \frac{1}{2}\abs{\Delta H_L - \chi}\right)}}{\Delta H_S},\sqrt{\frac{3}{8}}\right\}, 
\end{equation}
and when $\abs{\Delta H_L  - \chi} > \Delta H_S$, $\delta_\group \geq \sqrt{3/8}$. 
In particular, when $\abs{\Delta H_L - \chi} \ll \Delta H_S$, we have 
\begin{equation}
\delta_\group \gtrsim \sqrt{\frac{\abs{\Delta H_L - \chi}}{\Delta H_S}}.
\end{equation}
\end{proposition}
\begin{proof}
Since $U_{S,\theta}$ and $U_{L,\theta}$ are both periodic with a common period, we assume $H_S$ and $H_L$ both have integer eigenvalues. We also assume the smallest eigenvalue of $H_S$ is zero because constant shifts do not affect the definitions of symmetry measures. 
When $\mE_{\LtoS}(\cdot) = W(\cdot)W^\dagger$ is isometric, let $\ket{\frakc_{0}} = W\ket{0_L}$, $\ket{\frakc_{1}} = W\ket{1_L}$, and write
\begin{align}
\ket{\frakc_0} = \sum_{\eta=0}^{\Delta H_S} c_\eta^0 \ket{\eta^0},\quad 
\ket{\frakc_1} = \sum_{\eta=0}^{\Delta H_S} c_\eta^1 \ket{\eta^1},
\end{align} 
where $\sum_\eta \abs{c_\eta^0}^2 = \sum_\eta \abs{c_\eta^1}^2 = 1$ and $\ket{\eta^0}$ and $\ket{\eta^1}$ are eigenstates of $H_S$ with eigenvalue $\eta$. $\ket{\eta^0}$ and $\ket{\eta^1}$ may not be the same when $H_S$ is degenerate.  Note that when $\eta$ is not an eigenvalue of $H_S$, we simply take $c_\eta^{0} = 0$ (or $c_\eta^{1} = 0$) so that $c_\eta^{0}$ (or $c_\eta^{1}$) is well-defined for any integer $\eta$. 
Let $\ket{\psi} = \frac{1}{\sqrt{2}}(\ket{0_L}\ket{0_R}+\ket{1_L}\ket{1_R})$, where $R$ is a reference system. 
Then the channel fidelity 
\begin{align}
f_\theta :=&\, f(\mU_{S,\theta}\circ \mE_{\LtoS}\circ\mU_{L,\theta}^\dagger,\mE_{\LtoS}) \nonumber\\
=&\, \min_{\rho} f\big((\mU_{S,\theta}\!\circ\! \mE_{\LtoS}\!\circ\!\mU_{L,\theta}^\dagger\!\otimes\!\id_R)(\rho),(\mE_{\LtoS}\!\otimes\!\id_R)(\rho)\big)\nonumber\\
\leq&\, f\big((\mU_{S,\theta}\!\circ\! \mE_{\LtoS}\!\circ\!\mU_{L,\theta}^\dagger\!\otimes\!\id_R)(\ket{\psi}),(\mE_{\LtoS}\!\otimes\!\id_R)(\ket{\psi})\big)\nonumber\\
=&\, \abs{\braket{\psi|W^\dagger  U_{S,\theta} W U_{L,\theta}^{\dagger}|\psi}} \nonumber\\
=&\, \abs{\frac{1}{2}\sum_{\eta=0}^{\Delta H_S} \abs{c_\eta^0}^2 e^{-i\eta\theta+i\Delta H_L\theta} + \frac{1}{2}\sum_{\eta=0}^{\Delta H_S} \abs{c_\eta^1}^2 e^{-i\eta\theta}}\nonumber\\
=&\, \abs{\sum_{\eta=-\Delta H_L}^{\Delta H_S} c_\eta e^{-i\eta\theta}} = \abs{c_{\eta_*} \!+\! \sum_{\eta\neq\eta_*} c_\eta e^{-i(\eta-\eta_*)\theta}},
\end{align}
where we define $c_\eta := \frac{1}{2} \abs{c_{\eta+\Delta H_L}^0}^2+ \frac{1}{2}\abs{c_\eta^1}^2$ for $\eta \in [-\Delta H_L,\Delta H_S]$ and choose $\eta_*$ such that $c_{\eta_*} \geq c_{\eta}$ for all $\eta$. Note that there is always a $\theta$ such that $\sum_{\eta\neq\eta_*} c_\eta \cos((\eta-\eta_*)\theta) = 0$ (because the integration of it from $0$ to $2\pi$ is zero) and that $\sum_{\eta\neq\eta_*} c_\eta e^{-i(\eta-\eta_*)\theta}$ is imaginary. Then we must have 
\begin{equation}
\min_\theta f_\theta \leq \sqrt{c_{\eta_*}^2 + (1 - c_{\eta_*})^2}. 
\end{equation}

To arrive at a non-trivial lower bound on $\delta_\group = \sqrt{1-\min_\theta f_\theta^2}$, we need an upper bound of $\min_\theta f_\theta$ which is smaller than $1$. To this end, we analyze $c_{\eta_*}$  in detail. In particular, we consider two situations: 
\begin{enumerate}[(1),wide, labelindent=0pt]
\item $c_{\eta_*} \leq 1/2$ and a constant upper bound on $\min_\theta f_\theta$ exists. We can always find a subset of $\{\eta\}$ denoted by $\is$ such that $ 1/4 \leq \sum_{\eta\in\is} c_\eta \leq 1/2$. To find such a set, we first include $\eta_*$ in $\is$ and add new elements into $\is$ one by one until their sum is at least $1/4$. Then there is always a $\theta$ such that $(\sum_{\eta\in\is} c_\eta e^{-i\eta\theta})\cdot(\sum_{\eta\notin\is} c_\eta e^{-i\eta\theta})$ is imaginary, in which case $\min_\theta f_\theta \leq \sqrt{(1/4)^2+(3/4)^2} = \sqrt{5/8}$ and we have \begin{equation}\label{eq:delta-1} \delta_\group \geq \sqrt{1-\min_\theta f_\theta^2} \geq \sqrt{3/8}.\end{equation}
\item $c_{\eta_*} > 1/2$. Then $\sqrt{c_{\eta_*}^2 + (1-c_{\eta_*})^2}$ is a monotonically increasing function of $c_{\eta_*}$ and we only need to find an upper bound on $c_{\eta_*}$. To see such an upper bound exists, we first consider the special case where $\epsilon = 0$ and, according to the KL conditions and the HKS condition, $\braket{\frakc_0|H_S|\frakc_0} = \braket{\frakc_1|H_S|\frakc_1}$ holds true. 
On the other hand, if $c_{\eta_*} = \frac{1}{2}\abs{c_{\eta_*+\Delta H_L}^0}^2+ \frac{1}{2}\abs{c_{\eta_*}^1}^2= 1$, we must have $\abs{c_{\eta_*+\Delta H_L}^0}^2 = \abs{c_{\eta_*}^1}^2 = 1$ and $\braket{\frakc_0|H_S|\frakc_0} - \braket{\frakc_1|H_S|\frakc_1} = \Delta H_L > 0$, leading to a contradition. 

In general, to derive a non-trivial upper bound on $c_{\eta_*}$, we first note that $0 \leq \eta^* \leq \Delta H_S - \Delta H_L$ because otherwise either $c_{\eta_*+\Delta H_L}^0 = 0$ or $c_{\eta_*}^1 = 0$ which contradicts with $c_{\eta_*} > 1/2$. We have $\chi = \braket{\frakc_0|H_S|\frakc_0} - \braket{\frakc_1|H_S|\frakc_1}$ and 
\begin{gather*}
\sum_\eta \abs{c_\eta^0}^2 \eta = \sum_\eta \abs{c_\eta^1}^2 \eta + \chi, \\
~\Leftrightarrow~ 
\Delta H_L - \chi = -\left(1-\abs{c_{\eta_*}^1}^2\right)\left(\eta_* \!-\! \frac{\sum_{\eta\neq\eta_*}{\abs{c_\eta^1}^2\eta}}{\sum_{\eta\neq\eta_*}{\abs{c_\eta^1}^2}}\right) 
+\left(1-\abs{c_{\eta_*+\Delta H_L}^0}^2\right)\left(\eta_* + \Delta H_L - \frac{\sum_{\eta\neq\eta_*+\Delta H_L}{\abs{c_\eta^0}^2\eta}}{\sum_{\eta\neq\eta_*+\Delta H_L}{\abs{c_\eta^0}^2}}\right).
\end{gather*} 
Note that both $\bigg|\eta_* - \frac{\sum_{\eta\neq\eta_*}{\abs{c_\eta^1}^2\eta}}{\sum_{\eta\neq\eta_*}{\abs{c_\eta^1}^2}}\bigg|$ and $\bigg|\eta_* + \Delta H_L - \frac{\sum_{\eta\neq\eta_*+\Delta H_L}{\abs{c_\eta^0}^2\eta}}{\sum_{\eta\neq\eta_*+\Delta H_L}{\abs{c_\eta^0}^2}}\bigg|$ are at most $\Delta H_S$. Therefore, $c_{\eta_*} \leq 1 - \abs{\Delta H_L - \chi}/(2\Delta H_S)$ and 
\begin{align}
\delta_\group 
&\geq \sqrt{1 - \min_\theta f_\theta^2} \geq \sqrt{2 c_{\eta_*}(1-c_{\eta_*})} \geq \frac{\sqrt{\abs{\Delta H_L - \chi}\left(\Delta H_S - \frac{1}{2}\abs{\Delta H_L - \chi}\right)}}{\Delta H_S}. \label{eq:KL-last-step}
\end{align}
\end{enumerate}
\propref{prop:global-charge} then follows from combining \eqref{eq:delta-1} and \eqref{eq:KL-last-step}. 
\end{proof}

\subsection{Bounding charge fluctuation by QEC inaccuracy}
\label{sec:global-1-epsilon}

We now need to establish connections between $\chi$ and the QEC inaccuracy in order to link the global covariance violation to the QEC inaccuracy. We discuss two different methods that achieve this. 
 
\subsubsection{Knill--Laflamme-based method}
\label{sec:global-1-KL}

 Intuitively, a non-zero charge fluctuation leads to a violation of the KL conditions (\eqref{eq:KL}), which indicates a non-zero QEC inaccuracy.  Therefore, we may bound the QEC inaccuracy through analyzing the deviation from the KL condition. We call this method the KL-based method. Specifically, we have 
\begin{proposition}
\label{prop:charge-KL} 
\isometric Then it holds that \begin{equation}
\abs{\chi} \leq 2\epsilon\frakJ,
\end{equation}
where $\frakJ$ is a function of $H_S$ and $\mN_{S}$ defined by
\begin{equation}
\label{eq:def-frakJ}
\frakJ(\mN_S,H_S) := \min_{h: H_S = \sum_{ij} h_{ij} K_{S,i}^\dagger K_{S,j}} \Delta h,
\end{equation}
where $h$ is Hermitian. 
\end{proposition}

One can verify that 
$\frakJ(\mN_S,H_S)$ is efficiently computable using the following semidefinite program~\cite{boyd2004convex}: 
\begin{equation}
\begin{split}
    \min_{\substack{x,\nu\in\bR\\h \text{ is Hermitian}}} 2x, \quad\text{s.t. }\begin{pmatrix}x\id & h-\nu\id\\ h-\nu\id & x\id\end{pmatrix} \geq 0, \quad H_S = \sum_{ij}h_{ij} K_{S,i}^\dagger  K_{S,j}. 
\end{split}
\end{equation}

The proof of \propref{prop:charge-KL} is partly based on a useful lemma from Ref.~\cite{beny2010general} which connects the QEC inaccuracy $\epsilon$ to the deviation from the KL conditions:

\begin{lemma}[{\cite{beny2010general}}]
\label{lemma:KL-general}
Let $\Pi$ be the projection onto the code subspace of an isometric quantum code $\mE_{\LtoS}$ and the noise channel is $\mN_{S} = \sum_{i=1}^r K_i(\cdot)K_i^\dagger$. 
Then 
\begin{equation}
\epsilon(\mN_{S},\mE_{\LtoS}) = \min_{{\Lambda}} \sqrt{1 - f^2({\Lambda},{\Lambda}+{\mathcal{B}})}, 
\end{equation}
where ${\Lambda}(\rho) = \sum_{ij} \lambda_{ij} \trace(\rho) \ket{i}\bra{j}$, $({\Lambda}+{\mathcal{B}})(\rho) = {\Lambda}(\rho) + \sum_{ij} \trace({B }_{ij} \mE_{\LtoS}(\rho) ) \ket{i}\bra{j}$, and $\lambda_{ij}$ and ${B }_{ij}$ are constant numbers and operators satisfying $\Pi K_i^\dagger K_j \Pi = \lambda_{ij} \Pi + \Pi {B }_{ij} \Pi$. 
\end{lemma}

\propref{prop:charge-KL} then follows by connecting the deviation from the KL conditions to the charge fluctuation. The proof goes as follows.

\begin{proof}[Proof of {\propref{prop:charge-KL}}]
Let $\Pi$ be the projection onto the code subspace under consideration, 
$\ket{\frakc_{0}} = W\ket{0_L}$, $\ket{\frakc_{1}} = W\ket{1_L}$ where $\mE_{\LtoS}(\cdot) = W(\cdot)W^\dagger$, and the simplified notations $\braket{K_{S,i}^\dagger K_{S,j}}_{0} := \braket{\frakc_0|K_{S,i}^\dagger K_{S,j}|\frakc_0}$ and $\braket{K_{S,i}^\dagger K_{S,j}}_{1} := \braket{\frakc_1|K_{S,i}^\dagger K_{S,j}|\frakc_1}$. Assume $\lambda_{ij}$ and ${B }_{ij}$ satisfies $\Pi K_{S,i}^\dagger K_{S,j} \Pi = \lambda_{ij}\Pi + \Pi {B }_{ij} \Pi$.  Let $\ket{\psi} = \psi_0\ket{0_L}\ket{0_R} + \psi_1\ket{1_L}\ket{1_R}$ where $\abs{\psi_0}^2 + \abs{\psi_1}^2 = 1$, we have 
\begin{align}
 f\left({\Lambda}+{\mathcal{B}},{\Lambda}\right) 
\leq \min_{\psi_{0},\psi_{1}} f\left((({\Lambda}+{\mathcal{B}})\otimes\id_R)(\ket{\psi}),({\Lambda}\otimes\id_R)(\ket{\psi})\right)= \min_{\psi_{0}} f\left(\lambda+\tB,\lambda\right), 
\end{align}
where $\tB_{ij} = \abs{\psi_0}^2 \braket{K_{S,i}^\dagger K_{S,j}}_0 + \abs{\psi_1}^2 \braket{K_{S,i}^\dagger K_{S,j}}_1 - \lambda_{ij}$. 
According to the Fuchs--van de Graaf inequality $f(\rho,\sigma) \leq \sqrt{1 - \frac{1}{4}\norm{\rho-\sigma}_1^2}$~\cite{fuchs1999cryptographic}, 
\begin{align}
\epsilon &= \min_{{\Lambda},{\mathcal{B}}}\sqrt{1-f\left({\Lambda}+{\mathcal{B}},{\Lambda}\right)^2} \geq \min_{\lambda}\max_{\psi_0} \sqrt{1-f\left(\lambda+\tB,\lambda\right)^2}\geq \frac{1}{2}\min_{\lambda}\max_{\psi_0}\norm{\tB}_1. 
\end{align}
According to the HKS condition, $H_S = \sum_{ij} h_{ij}K_{S,i}^\dagger K_{S,j}$ for some Hermitian matrix $h$. Without loss of generality, we assume $h$ is diagonal and $H_S = \sum_{i} h_{ii}K_{S,i}^\dagger K_{S,i}$ because if not, we can always choose another set of Kraus operators that diagonalizes $h$. We can also assume $\max_{i} h_{ii}  = - \min_{i} h_{ii} =\frac{\Delta h}{2}$ because we can replace $H_S$ with $H_S -\nu\id$ for any $\nu\in\bR$. Then we have 
\begin{align}
\epsilon 
&\geq \frac{1}{2} \min_{\lambda}\max_{\psi_0} \norm{\tB}_1 
\geq \frac{1}{2} \min_{\lambda}\max_{\psi_0} \sum_{i}\abs{\tB_{ii}} \nonumber\\
&= \frac{1}{2} \min_{\lambda}\max_{\psi_0} \sum_{i}\Big|\abs{\psi_0}^2 \braket{K_{S,i}^\dagger K_{S,i}}_{0} + \abs{\psi_1}^2 \braket{K_{S,i}^\dagger K_{S,i}}_{1} - \lambda_{ii}\Big|\nonumber\\
&\geq \frac{1}{4} \min_{\lambda}\sum_{i} \left|\braket{K_{S,i}^\dagger K_{S,i}}_{0} - \lambda_{ii}\right| + \left|\braket{K_{S,i}^\dagger K_{S,i}}_{1} - \lambda_{ii}\right|\nonumber\\
&\geq \frac{1}{4} \sum_{i} \left|\braket{K_{S,i}^\dagger K_{S,i}}_{0} - \braket{K_{S,i}^\dagger K_{S,i}}_{1}\right|\nonumber\\
&\geq \frac{1}{4} \frac{1}{\max_i h_{ii}} \sum_{i} \abs{h_{ii} \left( \braket{K_{S,i}^\dagger K_{S,i}}_{0} -   \braket{K_{S,i}^\dagger K_{S,i}}_{1} \right)} \geq \frac{\abs{\chi}}{2 \Delta h}. 
\end{align}
Note that there might be many different choices of $h$ such that $H_S = \sum_{ij} h_{ij}K_{S,i}^\dagger K_{S,j}$ holds true. In order to obtain the tightest lower bound on $\epsilon$, we can minimize $\Delta h$ over all possible $h$ such that $H_S = \sum_{ij} h_{ij}K_{S,i}^\dagger K_{S,j}$, leading to $\epsilon \geq \abs{\chi}/2\frakJ$, where $\frakJ := \min_{h: H_S = \sum_{ij} h_{ij} K_{S,i}^\dagger K_{S,j}} \Delta h$. 
\end{proof}

\subsubsection{Quantum metrology method}
\label{sec:global-1-metrology}

Besides the KL-based method, the relationship between the charge fluctuation and the QEC inaccuracy could be understood through the lens of quantum metrology, which results in another inequality concerning $\chi$ and $\epsilon$, as shown in the following. A detailed comparison between the two bounds obtained from the KL-based method and the quantum metrology method (\propref{prop:charge-KL} and \propref{prop:charge-metrology}) will later be given in \secref{sec:global-1-tradeoff} and \secref{sec:global-1-noise}. 
\begin{proposition}
\label{prop:charge-metrology}
Consider a quantum code defined by $\mE_{\LtoS}$. Consider  physical Hamiltonian $H_S$,  logical Hamiltonian $H_L$, and noise channel $\mN_S$. Suppose the HKS condition is satisfied. 
Then it holds that 
\begin{equation}
\label{eq:charge-frakB}
\abs{\chi} \leq 2\epsilon\left(  \sqrt{(1-\epsilon^2)\frakF} + \frakB\right). 
\end{equation} 
Here 
\begin{equation}
\label{eq:def-frakB}
\frakB := \max_{\ket{\psi}}\sqrt{8 \bV_{H_S}\left(\mE_{\LtoS}(\ket{\psi})\right)} \leq \sqrt{2}\Delta H_S, 
\end{equation}
where the variance $\bV_{H}(\rho) := \trace(H^2\rho) - (\trace(H\rho))^2$, 
and $\frakF$ is a function of $\mN_S$ and $H_S$ defined by 
\begin{equation}
\label{eq:def-frakF}
\frakF(\mN_S,H_S) := 
\\ 4\min_{h: H_S = \sum_{ij} h_{ij} K_{S,i}^\dagger K_{S,j}} \bigg\|\sum_{ij}(h^2)_{ij} K_{S,i}^\dagger K_{S,j} - H_S^2\bigg\|, 
\end{equation}
where $h$ is Hermitian. In particular, when $\epsilon \ll 1$ and $\frakB \ll \sqrt{\frakF}$, we have 
\begin{equation}
\abs{\chi} \lesssim 2\epsilon\sqrt{\frakF}.
\end{equation}
\end{proposition}

Unlike $\frakJ(\mN_S,H_S)$ introduced in \secref{sec:global-1-KL}, $\frakF(\mN_S,H_S)$ appearing in \propref{prop:charge-metrology} has a clear operational meaning:
\begin{equation}
\frakF(\mN_S,H_S) \equiv \barF(\mN_{S,\theta}),\quad \forall \theta \in \bR. 
\end{equation}
Here $\barF(\mN_{S,\theta})$ is the regularized QFI~\cite{kolodynski2013efficient,zhou2020theory} of the quantum channel $\mN_{S,\theta} := \mN_{S}\circ\mU_{S,\theta}$ where $\theta$ is the unknown parameter to be estimated. (Generally, the regularized QFI of quantum channel $\Phi_\theta$ is defined by $\barF(\Phi_\theta) := \lim_{N\rightarrow\infty} F(\Phi_\theta^{\otimes N})/N$.) Note that $\barF(\mN_{S,\theta})$ is independent of $\theta$ and computable using semidefinite programming~\cite{demkowicz2012elusive}. Also, $\barF(\mN_{S,\theta}) \leq \infty$ if and only if the HKS condition is satisfied. The channel QFI inherits many nice properties from the QFI of quantum states. For example, here $\barF$ obeys the monotonicity property, i.e.~$\barF(\mR\circ\Phi_\theta\circ\mE) \leq \barF(\Phi_\theta)$ for arbitrary parameter-independent channels $\mR$ and $\mE$. 

{
The operational meaning of the quantity $\frakB$ is not immediately clear for general encoding channels, but when $\mE_{\LtoS}$ is isometric we have that 
\begin{align}
{\frakB} &= \sqrt{2 F(\mU_{S,\theta}\circ\mE_{\LtoS})}\big|_{\theta = 0},
\end{align}
which satisfies 
\begin{equation}
    \abs{\frakB - \sqrt{2}\delta_\point}\leq \sqrt{2}\Delta H_L,
\end{equation}
due to the chain rule of the square root of the channel QFI~\cite{katariya2020geometric}: 
$\sqrt{F(\Phi_{1,\theta}\circ\Phi_{2,\theta})} \leq \sqrt{F(\Phi_{1,\theta})}+\sqrt{F(\Phi_{2,\theta})}$ for any $\Phi_{1,\theta}$ and $\Phi_{2,\theta}$, and $F(\mU_{L,\theta})= F(\mU_{L,\theta}^\dagger) = (\Delta H_L)^2$. 
In general, $\frakB$ depends on specific encodings and in order to obtain an code-independent bound we should replace $\frakB$ with its upper bound $\Delta H_S$ so that \eqref{eq:charge-frakB} becomes 
\begin{equation}
\label{eq:charge-no-frakB}
\abs{\chi} \leq 2\epsilon\left(  \sqrt{(1-\epsilon^2)\frakF} + \sqrt{2}\Delta H_S\right),
\end{equation}
which still leads to useful bounds, e.g., for single-erasure noise as discussed in \secref{sec:global-1-noise}. 
However, in many cases $\frakB$ is negligible, i.e., $\frakB \ll \sqrt{\frakF}$ (or $\delta_\point \ll \sqrt{\frakF}$ for isometric codes) in the examples we study later in \secref{sec:case-study}. 
}

The monotonicity of the regularized QFI is a key ingredient in the proof of \propref{prop:charge-metrology}. Specifically, we introduce a two-level system $C$ and an ancillary qubit $A$ and consider the channel estimation of the error-corrected noise channel $\mN_{C,\theta} = \mR_{\SAtoC}\circ(\mN_{S,\theta}\otimes \id_A)\circ\mE_{\CtoSA}$ (see \figref{fig:dephasing}). $\mR_{\SAtoC}$ and $\mE_{\CtoSA}$ is carefully chosen such that $\barF(\mN_{C,\theta})$ is roughly $\Theta((\chi/\epsilon)^2)$ around $\theta = 0$. Intuitively, one might interpret $\barF(\mN_{C,\theta})$ handwavily as a quantity proportional to the square of the ``signal-to-noise ratio'' where the QEC inaccuracy $\epsilon$ is roughly the noise rate of $\mN_{C,\theta}$ and the charge fluctuation $\abs{\chi}$ is roughly the signal strength.  \propref{prop:charge-metrology} then follows from the monotonicity of QFI: 
\begin{equation}
\barF(\mN_{C,\theta}) \leq \barF(\mN_{S,\theta}). 
\end{equation}

We now explain the error-corrected metrology protocol in detail. We first introduce an ancilla-assisted two-level encoding. Consider a two-level system $C$ spanned by $\ket{0_C}$ and $\ket{1_C}$ and a Hamiltonian 
\begin{equation}
\label{eq:hamt-c}
H_C = \frac{\Delta H_L}{2} \cdot Z_C, \quad \mU_{C,\theta}(\cdot) = e^{-iH_C\theta} (\cdot) e^{iH_C\theta},
\end{equation}
where $Z_C$ is the Pauli-Z operator. 
We define a repetition encoding from $C$ to $LA$, 
\begin{equation}
\label{eq:rep-enc}
\mE^{\rep}_{\CtoLA}(\ket{0_C}) := \ket{0_L 0_A},\;
\mE^{\rep}_{\CtoLA}(\ket{1_C}) := \ket{1_L 1_A}, 
\end{equation}
where $A$ is the ancillary qubit. The corresponding repetition recovery channel is 
\begin{equation}
\label{eq:rep-rec}
\mR^{\rep}_{\LAtoC}(\cdot) :=
\sum_{i=0}^{d_L-1} R^\rep_i(\cdot)R^{\rep\dagger}_i,
\end{equation} 
where $R^\rep_0 = \ket{0_C}\bra{0_L0_A} + \ket{1_C}\bra{1_L1_A}$ and $R^\rep_1 = \ket{0_C}\bra{1_L0_A} + \ket{1_C}\bra{0_L1_A}$ and $R^\rep_{i>1} = \ket{0_C}\bra{i_L0_A} + \ket{1_C}\bra{i_L1_A}$. Clearly, $\mR^{\rep}_{\LAtoC} \circ \mE^{\rep}_{\CtoLA} = \id_C$. 
The repetition encoding $\mE_{\CtoLA}^{\rep}$ is covariant with respect to $H_C$ and $H_L$, i.e., $\mE_{\CtoLA}^{\rep}\circ \mU_{C,\theta} = (\mU_{L,\theta}\otimes\id_A) \circ \mE_{\CtoLA}^{\rep}$. Moreover, the repetition code corrects all bit-flip noise. When concatenated with $\mE_{\LtoS}$ and $\mR_{\StoL}$, the error-corrected noisy channel $
\mN_{C,\theta} = \mR^{\rep}_{\SAtoC} \circ (\mR_{\StoL} \circ\mN_{S,\theta} \circ \mE_{\LtoS} \otimes \id_A) \circ \mE_{\CtoLA}^{\rep}$ 
becomes a \emph{rotated dephasing channel}, namely, a single-qubit channel which is a composition of dephasing channel $(1-p)(\cdot)+p Z(\cdot)Z$ and a Pauli-$Z$ phase rotation $e^{-i\phi Z}$~\cite{zhou2020new}  (see \figref{fig:dephasing}). The regularized QFI of any rotated dephasing channel $\Phi_\theta$ is~\cite{zhou2020theory} 
\begin{equation}
\barF(\Phi_{\theta}) = \frac{\abs{\partial_\theta x_\theta}^2}{1 - \abs{x_{\theta}}^2},
\end{equation}
where the complex number $x_\theta = \braket{0|\Phi_\theta(\ket{0}\bra{1})|1}$. We consider the estimation around $\theta = 0$ for $\mN_{C,\theta}$. The monotonicity of the regularized QFI guarantees that 
\begin{equation}
\label{eq:local-QFI}
\barF(\mN_{C,\theta})\big|_{\theta = 0} = \frac{\abs{\partial_\theta \xi_\theta}^2\big|_{\theta=0}}{1 - \abs{\xi_{\theta = 0}}^2} \leq \barF(\mN_{S,\theta}),
\end{equation}
where 
\begin{equation}
\xi_\theta := \bra{0_C}\mN_{C,\theta}(\ket{0_C}\bra{1_C})\ket{1_C}. 
\end{equation}
\propref{prop:charge-metrology} can then be proven, connecting $\xi_\theta$ with $\epsilon$ and $\chi$.

\begin{figure}[tb]
	\center
	\includegraphics[width=0.48\textwidth]{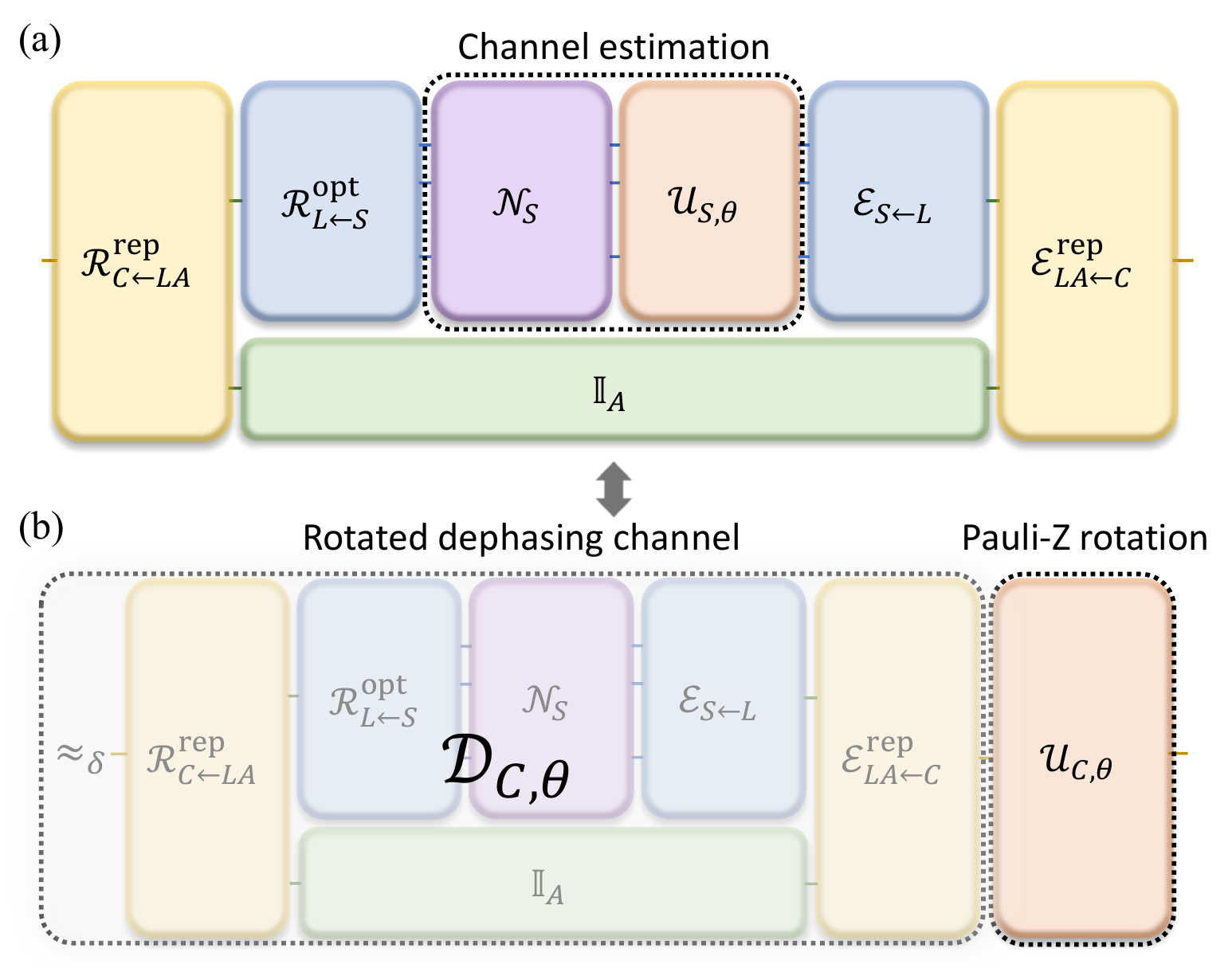}
	\caption{\label{fig:dephasing} A two-level encoding scheme for the estimation of $\theta$. 
	(a) Definition of the encoded channel $\mN_{C,\theta} $ in the system $C$. $\mN_{C,\theta} = \mR_{\SAtoC} \circ (\mN_{S,\theta} \otimes \id_A) \circ \mE_{\CtoSA}$ where $\mR_{\SAtoC} = \mR_{\CtoLA}^{\rep} \circ (\mR^{\optL}_{\StoL} \otimes \id_A)$ and $\mE_{\CtoSA} = (\mE_{\LtoS} \otimes \id_A) \circ \mE_{\CtoLA}^{\rep}$. It includes a concatenation of the repetition encoding from $C$ to $LA$ and the encoding under consideration $\mE_{\LtoS}$ from $L$ to $S$ with the optimal recovery channels chosen accordingly. 
	 (b) $\mN_{C,\theta}$ is the composition of $\mD_{C,\theta}$ and $\mU_{C,\theta}$ where $\mU_{C,\theta}$ is the unitary rotation with respect to a Pauli-Z Hamiltonian $H_C$ and $\mD_{C,\theta}$ is a rotated dephasing channel which is at most $\epsilon$-far from identity at $\theta = 0$. 
	 When $\delta_\group \approx 0$, $\mD_{C,\theta} \approx \mR_{\SAtoC} \circ (\mN_{S} \otimes \id_A) \circ \mE_{\CtoSA}$ is almost $\theta$-independent. Note that the gates are applied from right to left.}
\end{figure}

Now we are ready to present the formal proof of \propref{prop:charge-metrology}. 
\begin{proof}[Proof of {\propref{prop:charge-metrology}}]
Let $\mR^{\optL}_{\StoL}$ be the recovery channel such that 
$\epsilon =  P(\mR^{\optL}_{\StoL}\circ \mN_{S}\circ\mE_{\LtoS}, \id_{L})$. Let $\mR_{\SAtoC} = \mR_{\SAtoC}^\rep \circ (\mR_{\StoL}^{\optL} \otimes \id_A)$ and $\mE_{\CtoSA} = (\mE_{\LtoS}^{\optL} \otimes \id_A) \circ \mE_{\CtoLA}^\rep $, we have two rotated dephasing channels (see \figref{fig:dephasing}): 
\begin{align}
\mN_{C,\theta} 
&= \mR_{\SAtoC} \circ (\mN_{S,\theta} \otimes \id_A) \circ \mE_{\CtoSA} =: \mD_{C,\theta} \circ \mU_{C,\theta},
\end{align} 
where $\mD_{C,\theta}$ and $\mN_{C,\theta}$ are rotated dephasing channels of the following forms:
\begin{gather*}
\begin{split}
\mD_{C,\theta}(\cdot) =: (1 - p_\theta) e^{-i\frac{\phi_\theta}{2} Z_C} (\cdot) e^{i\frac{\phi_\theta}{2}Z_C} 
& + p_\theta Z_C   e^{-i\frac{\phi_\theta}{2}Z_C} (\cdot) e^{i\frac{\phi_\theta}{2}Z_C} Z_C, 
\end{split}\\
\begin{split}
\mN_{C,\theta}(\cdot) =: (1 - p_\theta) e^{-i\frac{\phi_\theta+\Delta H_L\theta}{2} Z_C} (\cdot) e^{i\frac{\phi_\theta+\Delta H_L\theta}{2} Z_C}
& + p_\theta Z_C e^{-i\frac{\phi_\theta+\Delta H_L\theta}{2} Z_C} (\cdot) e^{i\frac{\phi_\theta+\Delta H_L\theta}{2} Z_C} Z_C, 
\end{split}
\end{gather*}
and $\xi_\theta = \bra{0_C}\mD_{C,\theta}\circ\mU_{C,\theta}(\ket{0_C}\bra{1_C})\ket{1_C} = (1 - 2p_\theta)e^{-i(\phi_\theta+\Delta H_L\theta)}$. $\mD_{C,\theta}$ is identity when the code is both exactly covariant and exactly error-correcting; it is $\theta$-independent when the code is exactly covariant (see also Ref.~\cite{zhou2020new}). {Note that we will not use the channel $\mD_{C,\theta}$ in this proof (it will be used later in \secref{sec:global-2}), but we introduce the notation here 
to clarify its physical meaning. }

Consider the parameter estimation of $\theta$ in the neighborhood of $\theta = 0$. On one hand, for rotated dephasing channels (see \appref{app:dephasing} for the purified distance of rotated dephasing channels from identity), we have 
\begin{equation}
\label{eq:noise-rate-epsilon-1}
\sqrt{p_{\theta=0}} \leq P(\mD_{C,\theta = 0},\id_C) = P(\mN_{C,\theta = 0},\id_C).
\end{equation}
On the other hand, 
we have
\begin{align}
 P(\mN_{C,\theta=0},\id_C) = P(\mR_{\SAtoC} \circ (\mN_{S} \otimes \id_A) \circ \mE_{\CtoSA} ,\id_C)
\leq P(\mR^{\optL}_{\StoL} \circ \mN_{S} \circ \mE_{\LtoS}, \id_L) \leq \epsilon, 
\label{eq:noise-rate-epsilon-2}
\end{align}
where we use the monotonicity of the purified distance~\cite{tomamichel2015quantum} and the definition of $\epsilon$. Combining \eqref{eq:noise-rate-epsilon-1} and \eqref{eq:noise-rate-epsilon-2}, we have 
\begin{equation}
\label{eq:xi-epsilon}
\abs{\xi_{\theta = 0}} \geq 1 - 2\epsilon^2. 
\end{equation}
As shown in \appref{app:local-proof}, we also have 
\begin{equation}
\abs{\partial_\theta\xi_\theta}^2\big|_{\theta = 0}  \geq (\abs{\chi}-2\epsilon\frakB)^2,
\end{equation}
when $\abs{\chi} \geq 2\epsilon\frakB$. 
Hence, when $\abs{\chi} \geq 2\epsilon\frakB$, we must have 
\begin{align}
\barF(\mN_{S,\theta}) = \barF(\mN_{S,\theta}) \big|_{\theta = 0} \geq \barF(\mN_{C,\theta})\big|_{\theta = 0} \geq \frac{(\abs{\chi}  -2\epsilon\frakB)^2}{4\epsilon^2(1-\epsilon^2)},
\end{align} 
completing the proof. 
\end{proof}

\subsection{Consequent bounds on the trade-off between QEC and global covariance}
\label{sec:global-1-tradeoff}

In \secref{sec:global-1-charge} and \secref{sec:global-1-epsilon}, we derived bounds on $\delta_\group$ and $\epsilon$ separately, using the notion of charge fluctuation $\chi$. Combining these results, we immediately obtain the trade-off relations between $\delta_\group$ and $\epsilon$. 
\begin{theorem}
\label{thm:global-1}
\isometric It holds that, when $0 \leq G(\epsilon) \leq \Delta H_S$, 
\begin{equation}
\label{eq:tradeoff-1}
\delta_\group \geq 
\min\left\{\frac{\sqrt{G(\epsilon)\left(\Delta H_S - \frac{1}{2}G(\epsilon)\right)}}{\Delta H_S},\sqrt{\frac{3}{8}}\right\}, 
\end{equation}
and when $G(\epsilon) > \Delta H_S$, $\delta_\group \geq \sqrt{3/8}$, 
where we could take either 
\begin{equation}
G(\epsilon) = \Delta H_L - 2\epsilon \frakJ, 
\end{equation} 
or 
\begin{equation}
G(\epsilon) = \Delta H_L - 2 \epsilon (\sqrt{(1-\epsilon^2)\frakF}+\frakB),
\end{equation}
where $\frakJ$, $\frakF$, and $\frakB$ are given by \eqref{eq:def-frakJ}, \eqref{eq:def-frakF}, and \eqref{eq:def-frakB}, respectively. 
\end{theorem}
For the extreme cases of  exactly covariant codes and exactly error-correcting codes, we have the following corollaries: 
\begin{corollary}
\label{col:global-1-covariant}
\isometric  When $\epsilon = 0$, i.e., the code is exactly error-correcting, it holds that when $\Delta H_L \leq \Delta H_S$, 
\begin{equation}
\delta_\group 
\geq \frac{\sqrt{\Delta H_L\left(\Delta H_S -\frac{1}{2}\Delta H_L\right)}}{\Delta H_S}, 
\end{equation}
and when $\Delta H_L > \Delta H_S$, $\delta_\group \geq \sqrt{3/8}$. 
\end{corollary}
\begin{corollary}
\label{col:global-1-qec}
\isometric
When $\delta_\group = 0$, i.e., the code is exactly covariant, it holds that 
\begin{equation}
\label{eq:epsilon-covariant-1}
\epsilon \geq \frac{\Delta H_L}{2 \frakJ}, 
\end{equation}
where $\frakJ$ is given by \eqref{eq:def-frakJ}. 
\end{corollary}
\begin{corollary}
\label{col:global-1-qec-2}
\nonisometric
When $\delta_\group = 0$, i.e., the code is exactly covariant, we must have either $2\epsilon\frakB \geq \Delta H_L$ or 
\begin{equation}
\label{eq:epsilon-covariant-2}
\epsilon\frac{\sqrt{1-\epsilon^2}}{1-2\epsilon\frakB/\Delta H_L} \geq \frac{\Delta H_L}{2\sqrt{\frakF}}, 
\end{equation}
where $\frakF$ and $\frakB$ are given by \eqref{eq:def-frakF} and \eqref{eq:def-frakB}, respectively. 
\end{corollary}

We make a few remarks on the scope of application of these results. Although \propref{prop:global-charge} and \propref{prop:charge-KL} need the isometric encoding assumption, \propref{prop:charge-metrology} (and thus \corollaryref{col:global-1-qec-2}) holds for arbitrary codes. Also, \propref{prop:global-charge} only holds when $H_L$ and $H_S$ share a common period, but \propref{prop:charge-KL} and  \propref{prop:charge-metrology} hold true for arbitrary Hamiltonians without the $U(1)$ assumption. {Finally, a keen reader might have already noticed that the choice of the pair of orthonormal states $\{\ket{0_L},\ket{1_L}\}$ in the proofs of \propref{prop:global-charge}, \propref{prop:charge-KL} and \propref{prop:charge-metrology} is quite arbitrary (chosen only for the purpose of proving \thmref{thm:global-1}) and we can in principle replace it with any other pair and the proofs will still hold, leading to refinements of these propositions. We present these refinements in detail in \appref{app:refine-2}. In particular, \propref{prop:global-charge} leads to an inequality between $\delta_\group$ and $\delta_\charge$. }

{
To compare the results from the KL-based method and the quantum metrology method, we first consider the limiting situation where $\delta_\group \ll 1$ and $\epsilon \ll 1$.
Then we have
\begin{gather}
\label{eq:frakJ-global-asym}
\delta_\group \gtrsim \sqrt{\frac{\Delta H_L - 2\epsilon \frakJ}{\Delta H_S}} ~\Leftrightarrow~ \epsilon + \delta_\group^2 \frac{\Delta H_S}{2\frakJ} \gtrsim \frac{\Delta H_L}{2\frakJ},\\ 
\delta_\group \gtrsim \sqrt{\frac{\Delta H_L - 2\epsilon (\sqrt{\frakF}+\frakB)}{\Delta H_S}} ~\Leftrightarrow~ 
\label{eq:frakF-global-asym}
\epsilon + \delta_\group^2 \frac{\Delta H_S}{2(\sqrt{\frakF}+\frakB)} \gtrsim \frac{\Delta H_L}{2(\sqrt{\frakF}+\frakB)}.
\end{gather}
When $\frakB \ll \sqrt{\frakF}$, the metrology bound performs no worse than the KL-based bound because we always have $\frakJ^2 \geq \frakF$ (proof in \appref{app:inequality}). 
For the examples we study later in \secref{sec:case-study}, we find that $\frakB$ is  negligible, but in practice one may need to bound the parameter $\frakB$ a priori using properties of specific codes to obtain desired trade-off relations. It still open in general under which conditions $\frakB \ll \sqrt{\frakF}$ holds, and whether \propref{prop:charge-metrology} might be further improved with $\frakB$ removed. }

A byproduct of our results are lower bounds on $\epsilon$ (\eqref{eq:epsilon-covariant-1} and \eqref{eq:epsilon-covariant-2}) for exactly covariant codes, a special case which has been extensively studied in previous works~\cite{faist2019continuous,woods2019continuous,kubica2020using,zhou2020new,yang2020covariant}. As discussed below, the bound \eqref{eq:epsilon-covariant-1} for random local erasure noise behaves almost the same as the one in Ref.~\cite{faist2019continuous} and our  \propref{prop:charge-KL} provides an alternative method to obtain this result. However, compared to \propref{prop:charge-metrology}, the bound in Ref.~\cite{zhou2020new}
\begin{equation}
\label{eq:epsilon-covariant-3}
\epsilon\frac{\sqrt{1-\epsilon^2}}{1 - 2\epsilon^2} \geq \frac{\Delta H_L}{2\sqrt{\frakF}}
\end{equation}
does not involve the parameter $\frakB$, implying that the proof of our \propref{prop:charge-metrology} might be further improved.

\subsection{Noise models and explicit behaviors of the bounds}
\label{sec:global-1-noise}

Now we explicitly discuss how the bounds in \thmref{thm:global-1} behave under difference types of noise in an $n$-partite system. 
We consider 1-local Hamiltonians $H_S = \sum_{l=1}^n H_{S_l}$, so $\Delta H_S = O(n)$. In this case we have $\delta_\group = \Omega(1/\sqrt{\Delta H_S}) = \Omega(1/\sqrt{n})$ as long as $G(\epsilon) = \Omega(1)$. On the other hand, when $G(\epsilon) = o(1)$, the scaling of $\epsilon$ must be lower bounded by $\Omega(1/\frakJ)$ {(or $\Omega(1/(\sqrt{\frakF}+\frakB))$)} so it is important to understand the scalings of $\frakJ$, $\frakF$ and $\frakB$. When $n$ is large, the values of $\frakJ$ and $\frakF$ may be not efficiently computable. However, in \propref{prop:charge-KL},  \propref{prop:charge-metrology} and \thmref{thm:global-1}, we could always replace them with their efficiently computable upper bounds  and the trade-off relations still hold then. 
We discuss the following two general noise models~\cite{faist2019continuous,woods2019continuous,kubica2020using,zhou2020new,yang2020covariant} (there are still other types of noises that we will not cover, e.g., random long-range phase errors~\cite{woods2019continuous}): 
\begin{enumerate}[(1),wide, labelwidth=0pt, labelindent=0pt]
\item Random local noise, where different local noise channels acting on a constant number of subsystems randomly. Specifically, $\mN_S = \sum_{i} q_i \mN_{S}^{(i)}$ and $H_S = \sum_i H_{S}^{(i)}$, where $\mN_{S}^{(i)}$ represent the local noise channels acting on a constant number of subsystems, $q_i$ represent their probabilities ($q_i>0$ and $\sum_i q_i = 1$), and the HKS condition is satisfied for each pair of $(\mN_{S}^{(i)},H_{S}^{(i)})$. Then we have 
\begin{align}
\label{eq:frakJ-local}
\frakJ(\mN_S,H_S) &\leq \max_i \frac{1}{q_i}\frakJ(\mN_S^{(i)},H_S^{(i)}),\\
\label{eq:frakF-local}
\frakF(\mN_S,H_S) &\leq \sum_i \frac{1}{q_i}\widetilde{\frakF}(\mN_S^{(i)},H_S^{(i)}),
\end{align}
where 
\begin{equation}
\widetilde{\frakF}(\mN_S,H_S) := 
4\min_{h: H_S = \sum_{ij} h_{ij} K_{S,i}^\dagger K_{S,j}} \bigg\| \sum_{ij}(h^2)_{ij} K_{S,i}^\dagger K_{S,j} \bigg\|.
\end{equation} 
We prove \eqref{eq:frakJ-local} in \appref{app:frakJ} and \eqref{eq:frakF-local} was previously known in Ref.~\cite{zhou2020new}. Note that $\widetilde{\frakF}$ might be different when we replace $H_{S_l}$ with $H_{S_l} - \nu \id$ for some $\nu \in \bR$, one need to minimize over $\nu$ to find the optimal $\widetilde{\frakF}$\cite{zhou2020new}. 
For example, consider single-erasure noise in an $n$-partite system and let the erasure channel of the $\ell$-th subsystem be $\mN_{S_l}(\cdot) = \ket{\vac}\bra{\vac}_{S_l} \otimes \trace_{S_l}(\cdot)$ (we use  $\ket{\vac}$ to represent the vacuum state after erasure). When $\mN_S = \sum_{l=1}^n \frac{1}{n} \mN_{S_l}$ and the Hamiltonian takes the 1-local form $H_S = \sum_{l=1}^n H_{S_l}$, we have $\frakJ(\mN_{S_l},H_{S_l}) = \Delta H_{S_l}$ and $\widetilde{\frakF}(\mN_{S_l},H_{S_l})  = (\Delta H_{S_l})^2$. Then we have  
\begin{align}
\label{eq:erasure-J}
    \frakJ(\mN_S,H_S) &\leq n\max_l \Delta H_{S_l},\\
    \frakF(\mN_S,H_S) &\leq n\sum_{l=1}^n (\Delta H_{S_l})^2, \\
    \label{eq:erasure-frakB}
    \sqrt{\frakF(\mN_S,H_S)}+\frakB(\mN_S,H_S) &\leq n\!\left(\sqrt{\frac{\sum_{l=1}^n (\Delta H_{S_l})^2}{n}} + \sqrt{2}{\frac{\Delta H_S}{n}} \right). 
\end{align}
Note that using \eqref{eq:erasure-J} and \eqref{eq:epsilon-covariant-1}, we obtain $\epsilon \geq \Delta H_L/(2n\max_l \Delta H_{S_l})$ which is identical to Theorem~1 in Ref.~\cite{faist2019continuous}. 
{In \eqref{eq:erasure-frakB}, we use 
\begin{equation}
\frakB \leq \sqrt{2}{\Delta H_S}    
\end{equation}
Comparing \eqref{eq:erasure-frakB} with \eqref{eq:erasure-J}, we find that when $\frakB$ is not negligible, the quantum metrology method can still outperform the KL-based method in some cases (e.g., when one of $\Delta H_{S_l}$ is extremely large). 
For other types of random local noise acting on each subsystem uniformly randomly, we also have $\frakJ=O(n)$ and $\sqrt{\frakF}+\frakB=O(n)$ and the behaviors of the trade-off relations from the KL-based method and the quantum metrology method are similar.} {From~\eqref{eq:frakJ-global-asym} and \eqref{eq:frakF-global-asym}, we have $\epsilon + \Theta(1)\cdot \delta_\group^2  = \Omega(1/n)$, meaning that when both $\epsilon$ and $\delta_\group$ are sufficiently small, their optimal scalings are $\Theta(1/n)$ and $\Theta(1/\sqrt{n})$, respectively. }

\item Independent noise, where noise channels act on each subsystem independently. Note that  independent noise is considered a ``stronger'' noise model than local noise because the noise actions are no longer guaranteed to be local. Specifically, $\mN_S = \bigotimes_{l=1}^n \mN_{S_l}$ and $H_S = \sum_{l=1}^n H_{S_l}$ where $\mN_{S_l}$ represent independent noise channels acting on each subsystem and $H_{S_l}$ acts only non-trivially on the subsystem $S_l$, and the HKS condition is  satisfied for each pair of $(\mN_{S_l},H_{S_l})$. Here we have 
\begin{align}
\label{eq:frakJ-independent}
\frakJ(\mN_S,H_S) &\leq \sum_{l=1}^n \frakJ(\mN_{S_l},H_{S_l}),\\
\label{eq:frakF-independent}
\frakF(\mN_S,H_S) &= \sum_{l=1}^n \frakF(\mN_{S_l},H_{S_l}).
\end{align}
The proof of \eqref{eq:frakJ-independent} is provided in \appref{app:frakJ}, and \eqref{eq:frakF-independent} follows directly from the additivity of the regularized QFI $\frakF = \barF(\bigotimes_{l=1}^n\mN_{S_l,\theta}) = \sum_{l=1}^n \barF(\mN_{S_l,\theta})$~\cite{zhou2020new}. 
Therefore we have $\frakJ = O(n)$ and $\sqrt{\frakF} = O(\sqrt{n})$ and there is now a quadratic gap between them.
If $\frakB$ can be upper bounded by $O(\sqrt{n})$ (e.g., in \secref{sec:case-study}), from \eqref{eq:frakF-global-asym}, we have $\epsilon + \Theta(\sqrt{n})\cdot \delta_\group^2  = \Omega(1/\sqrt{n})$, meaning that when both $\epsilon$ and $\delta_\group$ are sufficiently small, their optimal scalings are both $\Theta(1/\sqrt{n})$. From \eqref{eq:frakJ-global-asym}, we only have $\epsilon + \Theta(1)\cdot \delta_\group^2  = \Omega(1/n)$ and a worse lower bound $\epsilon = \Omega(1/n)$ for small $\delta_\group$. In general, for independent noise, the trade-off bound from the KL-based method is asymptotically weaker than the one from the quantum metrology method as long as $\frakB = o(n)$.   
\end{enumerate}

Finally, we remark here that the exact values of $\frakJ(\mN_{S_l},H_{S_l})$, $\frakF(\mN_{S_l},H_{S_l})$ and $\widetilde{\frakF}(\mN_{S_l},H_{S_l})$ can also be analytically calculated (or upper bounded) for not only erasure noise, but also other types of practically relevant noise, e.g., depolarizing noise. In principle, to derive an upper bound on $\frakJ(\mN_{S_l},H_{S_l})$,  $\frakF(\mN_{S_l},H_{S_l})$ or $\widetilde{\frakF}(\mN_{S_l},H_{S_l})$, one only need to find a Hermitian matrix $h$ that satisfies $H_{S_l} = \sum_{ij} h_{ij} K_{S_l,i}^\dagger K_{S_l,j}$ and use the target functions $\Delta h$,  $\bigg\| \sum_{ij}(h^2)_{ij} K_{S,i}^\dagger K_{S,j} - H_S^2\bigg\|$ or $\bigg\| \sum_{ij}(h^2)_{ij} K_{S,i}^\dagger K_{S,j}\bigg\|$ as the upper bound. One can further tighten the bound by minimizing the target functions over all possible $h$. We give a few examples below. 

First, we note that for an erasure noise channel $\mN(\cdot) = (1-p)(\cdot) + p \ket{\vac}\bra{\vac}\trace(\cdot)$, we have 
\begin{align}
  \frakJ(\mN,H)  = \frac{\Delta H}{p},\quad 
  \frakF(\mN,H)  = (\Delta H)^2 \frac{1-p}{p},\quad 
  \widetilde{\frakF}(\mN,H) =  \frac{(\Delta H)^2}{p}. 
\end{align}
To derive these, we assume the system is $d$-dimensional and let $K_1 = \sqrt{1-p}\id$, $K_{i+1} = \sqrt{p}\ket{\vac}\bra{i}$, for $i = 1,\ldots, d$. Then the $d$-dimensional Hermitian matrix $h$ such that $H = \sum_{ij} h_{ij} K_i^\dagger K_j$ must be $h = \begin{pmatrix} \frac{h_{11}}{1-p} & 0 \\  0 & \frac{H-h_{00}\id}{p} \end{pmatrix}$ for some $h_{11} \in \bR$. The above equations follow straightforwardly by minimizing the target functions over $h_{11}$ (see also Ref.~\cite{zhou2020new}). 

Similarly, for single-qubit depolarizing noise $\mN(\cdot) = (1-p)(\cdot) + p \frac{\id}{2}$, we have 
\begin{align}
  \frakJ(\mN,H)  \leq \frac{2\Delta H}{p(1-p/2)},\quad 
  \frakF(\mN,H)  = (\Delta H)^2 \frac{2(1-p)^2}{p(3-2p)},\quad 
  \widetilde{\frakF}(\mN,H) =  (\Delta H)^2 \frac{2-p}{p(3-2p)}. 
\end{align}
To derive these, we let $K_1 = \sqrt{1-\frac{3}{4}p}\id$, $K_2 = \sqrt{\frac{p}{4}}X$, $K_3 = \sqrt{\frac{p}{4}}Y$ and $K_4 = \sqrt{\frac{p}{4}}Z$. Then the $4$-dimensional Hermitian matrix such that $H = \sum_{ij} h_{ij} K_i^\dagger K_j$ must be $h = \begin{pmatrix} 0 & 0 & 0 & h_{14} \\ 0 & 0 & i h_{23} & 0 \\ 0 & -i h_{23} & 0 & 0 \\ h_{14} & 0 & 0 & 0 \end{pmatrix}$ for some $h_{14}, h_{23} \in \bR$ when $\bigg\| \sum_{ij}(h^2)_{ij} K_{S,i}^\dagger K_{S,j} - H_S^2\bigg\|$ in $\frakF$ and $\bigg\| \sum_{ij}(h^2)_{ij} K_{S,i}^\dagger K_{S,j}\bigg\|$ in $\frakJ$ are minimized (see Ref.~\cite{zhou2020theory}). The above equations follow straightforwardly by minimizing the target functions over $h_{14}$ and $h_{23}$. Note that for $\frakJ$, there is no guarantee that the anti-diagonal form of $h$ is optimal, so it only provides an upper bound on $\frakJ$. 

Finally, for depolarizing noise on qudits: $\mN(\cdot) = (1-p)(\cdot) + p \frac{\id}{d}$, we have from Ref.~\cite{zhou2020new} that 
\begin{align}
  \frakF(\mN,H)  \leq (\Delta H)^2 \frac{1-p}{p},\quad 
  \widetilde{\frakF}(\mN,H) \leq (\Delta H)^2 \frac{1}{p}, 
\end{align}
and how to find a simple upper bound on $\frakJ$ is still open.

\section{Trade-off between QEC and global covariance: Gate implementation error approach}
\label{sec:global-2}

In this section, we introduce another framework that also enables us to derive the trade-off between the QEC inaccuracy $\epsilon$ and the global covariance violation $\delta_\group$ and could be interesting in its own right. Here the idea is to analyze a key notion which we call the gate implementation error  $\gamma$ that allow us to treat $\epsilon$ and $\delta_\group$ on the same footing. More specifially, we first formally define $\gamma$ in \secref{sec:global-2-gate} and show that $\epsilon + \delta_\group \geq \gamma$. Then we derive two lower bounds on $\gamma$ using two different methods from quantum metrology and quantum resource theory, which automatically induce two trade-off relations between the QEC inaccuracy and the global covariance violation. We will compare the gate implementation error approach to the charge fluctuation approach at the end of this section.  

\subsection{Gate implementation error as a unification of QEC inaccuracy and global covariance violation} 
\label{sec:global-2-gate}

Consider a practical quantum computing scenario where we want to implement a set of logical gates $U_{L,\theta}= e^{-iH_L\theta}$ for $\theta \in \bR$ using physical gates $U_{S,\theta} = e^{-iH_S\theta}$ under noise $\mN_S$. We would like to design an encoding and a recovery channel such that $\mR_{\StoL}\circ\mN_{S,\theta}\circ\mE_{\LtoS}$ simulate $\mU_{L,\theta}$. We call the error in such simulations the \emph{gate implementation error} and the \emph{Choi gate implementation error}, defined by 
\begin{align}
\gamma
&:= \min_{\mR_{\StoL}} \max_\theta P(\mR_{\StoL}\circ \mN_{S,\theta}\circ\mE_{\LtoS}, \mU_{L,\theta}), 
\\
\bgamma
&:= \min_{\mR_{\StoL}} \max_\theta \barP(\mR_{\StoL}\circ \mN_{S,\theta}\circ\mE_{\LtoS}, \mU_{L,\theta}). 
\end{align}
Both the QEC inaccuracy and the covariance violation contribute to the gate implementation error (see \figref{fig:gate}). Clearly, $\gamma = 0$ when the quantum code is exactly error-correcting and covariant. In general, $\gamma$ is upper bounded by the sum of $\epsilon$ and $\delta_\group$, as shown in the following 
proposition. 
\begin{proposition}
\label{prop:gate}
Consider a quantum code defined by $\mE_{\LtoS}$. Consider  physical Hamiltonian $H_S$,  logical Hamiltonian $H_L$, and noise channel $\mN_S$.  It holds that
\begin{gather}
\epsilon + \delta_\group \geq \gamma,\\
\bepsilon + \bdelta_\group \geq \bgamma.
\end{gather}
\end{proposition}
\begin{proof}
Using the triangular inequality of the purified distance~\cite{tomamichel2015quantum}, we have 
\begin{multline}
P(\mR_{\StoL}\circ \mN_{S,\theta}\circ\mE_{\LtoS}, \mU_{L,\theta})   \leq \\ P(\mR_{\StoL}\circ \mN_{S,\theta}\circ\mE_{\LtoS}, \mR_{\StoL}\circ \mN_{S}\circ\mE_{\LtoS}\circ\mU_{L,\theta}) + P(\mR_{\StoL}\circ \mN_{S}\circ\mE_{\LtoS}\circ\mU_{L,\theta}, \mU_{L,\theta}).
\end{multline}
The first term is upper bounded by $P(\mU_{S,\theta}\circ\mE_{\LtoS}, \mE_{\LtoS}\circ\mU_{L,\theta})$ using the monotonicity of the purified distance and the second term is equal to $P(\mR_{\StoL}\circ \mN_{S}\circ\mE_{\LtoS}, \id_L)$ by definition. Then $\gamma \leq \epsilon + \delta_\group$ follows by taking the maximization over $\theta$ and the minimization over $\mR_{\StoL}$ on both sides. The above discussion also holds when we replace the purified distance $P(\cdot,\cdot)$ with the Choi purified distance $\barP(\cdot,\cdot )$, implying that $\bgamma \leq \bepsilon + \bdelta_\group$. 
\end{proof}

\begin{figure}[tb]
	\center
	\includegraphics[width=0.4\textwidth]{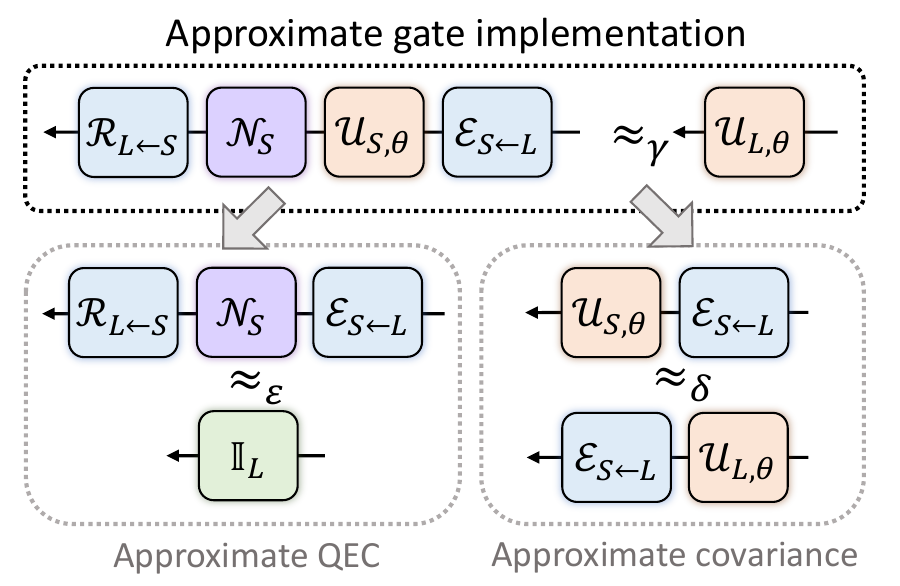}
	\caption{\label{fig:gate} 
	Both the QEC inaccuracy $\epsilon$ and the covariance violation $\delta_\group$ contribute to the error in approximate gate implementation. Specifically, the (Choi) gate implementation error $\gamma$ ($\bgamma$) is upper bounded by $\delta_\group + \epsilon$ ($\bdelta_\group + \bepsilon$). Note that the gates are applied from right to left.}
\end{figure}

\subsection{Bounding gate implementation error}
\subsubsection{Quantum metrology method}
\label{sec:global-2-metrology}

Now we derive a lower bound on the gate implementation error $\gamma$, where we consider the approximate gate implementation of $\mU_{L,\theta}$ using noisy gates $\mN_{S,\theta}$ as an error-corrected metrology protocol where $\theta$ is an unknown parameter to be estimated. 

Again, we use the ancilla-assisted two-level encoding, as introduced in \secref{sec:global-1-metrology}. We choose the repetition code concatenated with the quantum code under study, so that the error-corrected noise channel $
\mN_{C,\theta} = \mR^{\rep}_{\SAtoC} \circ (\mR_{\StoL} \circ\mN_{S,\theta} \circ \mE_{\LtoS} \otimes \id_A) \circ \mE_{\CtoLA}^{\rep}$ becomes a rotated dephasing channel. 
The main difference between the error-corrected metrology protocol we use here and the one in \secref{sec:global-1-metrology} is that now we choose the recovery channel $\mR_{\StoL}$ to be the optimal recovery channel which minimizes the gate implementation error (instead of the QEC inaccuracy) and guarantees a lower noise rate over the entire group of $\theta$ (instead of just around $\theta = 0$). In this case, we show that there always exists some $\theta_*$ such that $\barF(\mN_{C,\theta_*}) = \Theta((\Delta H_L/\gamma)^2)$, which then provide us a lower bound on $\gamma$ using the monotonicity of the regularized QFI. Now we state and prove \thmref{thm:global-2-metrology} which provides a lower bound on $\gamma$ (and thus on $\epsilon + \delta_\group$). 

\begin{theorem}
\label{thm:global-2-metrology}
Consider a quantum code defined by $\mE_{\LtoS}$. Consider  physical Hamiltonian $H_S$,  logical Hamiltonian $H_L$, and noise channel $\mN_S$. Suppose the HKS condition is satisfied. Then it holds that
\begin{equation}
\epsilon + \delta_\group \geq \gamma \geq \ell_1\left(\frac{\Delta H_L}{2\sqrt{\frakF}
}\right), 
\end{equation}
where $\frakF$ is given by \eqref{eq:def-frakF},  $\ell_1(x) = x + O(x^2)$ is the inverse function of the monotonic function $x=y\frac{\sqrt{1-y^2}}{1-2y^2}$ on $[0,1/\sqrt{2})$.
\end{theorem}

In particular, for exact QEC codes, we have the following corollary: 
\begin{corollary}
\label{col:global-2-covariant-metrology}
\nonisometric When $\epsilon = 0$, i.e., when the code is exactly error-correcting, it holds that $\delta_\group \geq \ell_1(\Delta H_L/2\sqrt{\frakF}
)$, where $\frakF$ is given by \eqref{eq:def-frakF}. 
\end{corollary}

\begin{proof}[Proof of {\thmref{thm:global-2-metrology}}]
By definition, there exists a $\mR^{\optG}_{\StoL}$ such that 
$\gamma =  \max_\theta P(\mR^{\optG}_{\StoL}\circ \mN_{S,\theta}\circ\mE_{\LtoS}, \mU_{L,\theta})$. 
Let $\mR_{\SAtoC} = \mR_{\SAtoC}^\rep \circ (\mR_{\StoL}^{\optG} \otimes \id_A)$ and $\mE_{\CtoSA} = (\mE_{\LtoS}^{\optG} \otimes \id_A) \circ \mE_{\CtoLA}^\rep $, we have $\mN_{C,\theta} 
= \mR_{\SAtoC} \circ (\mN_{S,\theta} \otimes \id_A) \circ \mE_{\CtoSA} 
= \mD_{C,\theta} \circ \mU_{C,\theta}$, 
and $\mD_{C,\theta}$ and $\mN_{C,\theta}$ are rotated dephasing channels of the following forms:
\begin{gather}
\begin{split}
\mD_{C,\theta}(\cdot) =\; (1 - p_\theta) e^{-i\frac{\phi_\theta}{2} Z_C} (\cdot) e^{i\frac{\phi_\theta}{2}Z_C}  + p_\theta Z_C   e^{-i\frac{\phi_\theta}{2}Z_C} (\cdot) e^{i\frac{\phi_\theta}{2}Z_C} Z_C, 
\end{split}\\
\begin{split}
\mN_{C,\theta}(\cdot) =\; (1 - p_\theta) e^{-i\frac{\phi_\theta+\Delta H_L\theta}{2} Z_C} (\cdot) e^{i\frac{\phi_\theta+\Delta H_L\theta}{2} Z_C}  + p_\theta Z_C e^{-i\frac{\phi_\theta+\Delta H_L\theta}{2} Z_C} (\cdot) e^{i\frac{\phi_\theta+\Delta H_L\theta}{2} Z_C} Z_C. 
\end{split}
\end{gather}
where $p_\theta \in (0,1)$ and $\phi_\theta \in [0,2\pi)$. 
Let $
\xi_\theta 
= \bra{0_C}\mN_{C,\theta}(\ket{0_C}\bra{1_C})\ket{1_C} = (1 - 2p_\theta)e^{-i(\phi_\theta+\Delta H_L\theta)}$. 
The regularized channel QFI of rotated dephasing channels is 
\begin{align}
\label{eq:dephasing-SLD}
&\barF(\mN_{C,\theta}) = \frac{\abs{\partial_\theta \xi_\theta}^2}{1 - \abs{\xi_\theta}^2}  = \frac{(1-2p_\theta)^2 (\partial_\theta\phi_\theta + \Delta H_L)^2 }{4p_\theta(1-p_\theta)} + \frac{(\partial_\theta p_\theta)^2 }{4p_\theta(1-p_\theta)}. 
\end{align}
In order to get a lower bound on $\barF(\mN_{C,\theta})$ as a function of $\gamma$. We note that the purified distance between $\mD_{C,\theta}$ and $\id_C$ (see \appref{app:dephasing}) is upper bounded by $\gamma$: 
\begin{equation}
\label{eq:purified}
P(\mD_{C,\theta},\id_C) =  \sqrt{\frac{1-(1-2p_\theta)\cos\phi_\theta}{2}} \leq \gamma,  
\end{equation}
because 
\begin{align*}
P(\mD_{C,\theta},\id_C)
&= P(\mR_{\SAtoC} \circ (\mN_{S,\theta} \otimes \id_A) \circ \mE_{\CtoSA} \circ \mU_{C,\theta}^\dagger,\id_C)\nonumber
\\
&= 
P(  \mR_{\SAtoC} \circ (\mN_{S,\theta}\circ \mE_{\LtoS} \circ \mU_{L,\theta}^\dagger \otimes \id_A) \circ \mE_{\CtoLA}^{\rep} , \id_C )\nonumber\\
&\leq P(\mR^{\optG}_{\StoL} \circ \mN_{S} \circ \mU_{S,\theta}\circ \mE_{\LtoS} \circ \mU_{L,\theta}^\dagger, \id_L) \leq \gamma,  
\end{align*}
where we use the monotonicity of the purified distance and the definition of $\gamma$.  \eqref{eq:purified} implies $\sqrt{p_\theta} \leq P(\mD_{C,\theta},\id_C) \leq \gamma$ for all $\theta \in \bR$.  
Since $U_{L,\theta}$ and $U_{S,\theta}$ are periodic with a common period $\tau$, we must have $\phi_\theta = \phi_{\theta + \tau}$. Therefore, there must exists a $\theta_*$ such that $\partial_\theta \phi_\theta\big|_{\theta = \theta_*} = 0$. Then using \eqref{eq:dephasing-SLD} and the monotonicity of the regularized QFI, we see that 
\begin{align}
&\barF(\mN_{S,\theta}) = \barF(\mN_{S,\theta_*}) \geq \barF(\mN_{C,\theta_*}) \geq 
\frac{(1-2p_{\theta_*})^2 (\Delta H_L)^2 }{4p_{\theta_*}(1-p_{\theta_*})} \geq 
\frac{(1-2\gamma^2)^2 (\Delta H_L)^2 }{4\gamma^2(1-\gamma^2)}.
\end{align}
\thmref{thm:global-2-metrology} then follows from \propref{prop:gate}.

\end{proof}

Note that \thmref{thm:global-2-metrology} coincides with Theorem 1 in Ref.~\cite{zhou2020new} in the special case where $\delta_\group = 0$. 

\subsubsection{Quantum resource theory method}
\label{sec:global-2-res}

Now we present another derivation based on quantum resource theory, which allows us to derive not only a lower bound on the worst-case gate implementation error, but also on the Choi gate implementation error. 

We work with a resource theory of coherence~\cite{marvian2020coherence} where the free (incoherent) states are those whose density operators commute with the Hamiltonian and the free (covariant) operations are those that commute with the Hamiltonian evolution, e.g., a covariant operation $\mC_\StoL$ from $S$ to $L$ satisfies $\mC_\StoL \circ \mU_{S,\theta} = \mU_{L,\theta} \circ \mC_\StoL$ for all $\theta \in \bR$. Assuming that the recovery operations $\mR_{\StoL}$ and the noise channel $\mN_S$ are covariant, we can formulate the covariant QEC as a resource conversion task from noisy physical states to error-corrected logical states and the noise rate of the latter is upper bounded by $\gamma$, illustrated by the following lemma:

\begin{proposition}
\label{prop:cov}
Consider a quantum code defined by $\mE_{\LtoS}$. Consider  physical Hamiltonian $H_S$,  logical Hamiltonian $H_L$, and noise channel $\mN_S$. 
Suppose $\mN_S$ commutes with $\mU_{S,\theta}$. Then the QEC inaccuracy measures under the restriction that the recovery channel is covariant satisfy 
\begin{align}
\label{eq:cov-ineq-1}
\epsilon_\cov = \min_{\mR^\cov_\StoL}
P(\mR^\cov_{\StoL}\circ \mN_{S}\circ\mE_{\LtoS}, \id_L)\leq \gamma ,\\ 
\label{eq:cov-ineq-2}
\bepsilon_\cov = \min_{\mR^\cov_\StoL}
\barP(\mR^\cov_{\StoL}\circ \mN_{S}\circ\mE_{\LtoS},\id_L) \leq \bgamma, 
\end{align} 
where $\mR^\cov_{\StoL}$ is a recovery channel satisfying $\mU_{L,\theta}\circ\mR^\cov_{\StoL} = \mR_{\StoL}^\cov\circ\mU_{S,\theta}$. 
\end{proposition}
\begin{proof}
Let $\mR^{\optG}_{\StoL}$ be a recovery channel such that $\gamma =  \max_\theta P(\mR^{\optG}_{\StoL}\circ \mN_{S,\theta}\circ\mE_{\LtoS}, \mU_{L,\theta})$. Suppose $U_{S,\theta}$ and $U_{L,\theta}$ share a common period $\tau$. 
Consider the following recovery channel: 
\begin{equation}
\label{eq:def-cov-rec}
\mR_{\StoL}^\covG = \frac{1}{\tau} \int_0^\tau  d\theta\;\mU^\dagger_{L,\theta} \circ \mR^{\optG}_{\StoL} \circ \mU_{S,\theta}. 
\end{equation}
It can be verified that $\mR_{\StoL}^\covG$ must be covariant and 
\begin{equation*}
\begin{split}
 P(\mR^\covG_{\StoL}\circ \mN_{S}\circ\mE_{\LtoS}, \id_L)
&= P\left(\frac{1}{\tau}\int_0^\tau d\theta\; \mU_{L,\theta}^\dagger \circ \mR^{\optG}_{\StoL} \circ  \mN_{S} \circ \mU_{S,\theta} \circ \mE_{\LtoS} ,\id_L\right)\\
&\leq \max_\theta P\left(\mU_{L,\theta}^\dagger \circ \mR^{\optG}_{\StoL} \circ  \mN_{S} \circ \mU_{S,\theta} \circ \mE_{\LtoS} ,\id_L\right) = \gamma,
\end{split}
\end{equation*}
where we used the concavity of $f^2(\Phi,\id)$ with respect to $\Phi$~\cite{schumacher1996sending}, leading to \eqref{eq:cov-ineq-1}. 

Similarly, let $\mR^{\optCG}_{\StoL}$ be a recovery channel such that $\bgamma =  \max_\theta \overline{P}(\mR^{\optCG}_{\StoL}\circ \mN_{S,\theta}\circ\mE_{\LtoS}, \mU_{L,\theta})$. We can define $\mR_{\StoL}^{\covCG} = \frac{1}{\tau} \int_0^\tau  d\theta\;\mU^\dagger_{L,\theta} \circ \mR^{\optCG}_{\StoL} \circ \mU_{S,\theta} $  and verify that $\overline{P}(\mR^{\covCG}_{\StoL}\circ \mN_{S}\circ\mE_{\LtoS}, \id_L) \leq \bgamma$, leading to \eqref{eq:cov-ineq-2}.  
\end{proof}

In order to derive a concrete lower bound on $\gamma$ and $\bgamma$ using \propref{prop:cov}, we choose a resource monotone based on another type of QFI of quantum states called the RLD QFI~\cite{yuen1973multiple} defined by $F^\txr(\rho_\theta) = \trace((\partial_\theta\rho_\theta)^2\rho_\theta^{-1})$ when the support of $\partial_\theta \rho_\theta$ is contained in $\rho_\theta$ and $= +\infty$ otherwise. The resource monotone satisfies  
\begin{equation}
\label{eq:monotone-RLD}
F^\txr(\mC_{\StoL}(\rho_S),H_L) \leq F^\txr(\rho_S,H_S),
\end{equation}
for all $\rho_S$ and covariant operations $\mC_{\StoL}$, where 
\begin{align}
&F^\txr(\rho,H) := F^\txr(e^{-iH\theta}\rho e^{iH\theta}) 
=\begin{cases}
\trace(H\rho^2 H \rho^{-1}) \!-\! \trace(\rho H^2) & \supp(H\rho H)\!\subseteq\! \supp(\rho),\\
+\infty & \text{otherwise.}
\end{cases}
\end{align}

Consider an error-corrected logical state $\mR^\cov_{\StoL}\circ\mN_{S}\circ\mE_{\LtoS}(\ket{\psi_L})$ using covariant recovery operations. On one hand, its RLD QFI is lower bounded by $\Theta(1/\gamma^2)$ when $\rho$ is $\gamma$-close to a coherent pure state in terms of purified distance~\cite{marvian2020coherence,zhou2020new}. On the other hand, its RLD QFI is upper bounded by the RLD QFI of the noisy physical state $F^\txr(\mN_{S}\circ\mE_{\LtoS}(\ket{\psi_L}),H_S)$ is no less than the channel RLD QFI $F^\txr(\mN_{S,\theta}) = \max_\rho F^\txr((\mN_{S,\theta}\otimes \id)(\rho))$~\cite{hayashi2011comparison,katariya2020geometric}.
Specifically, 
\begin{equation}
    F_\txr(\mN_{S,\theta}) = \big\|\trace_S\big(\Psi_{SR}^{\mN_S,H_S}(\Psi_{SR}^{\mN_S})^{-1}\Psi_{SR}^{\mN_S,H_S}\big)\big\|, 
\end{equation} 
where $\Gamma_{SR}^{\mN_S} = (\mN_S \otimes \id_R)(\Gamma_{SR})$, $\Gamma_{SR}^{\mN_S,H_S} = (\mN_S \otimes \id_R)\big((H_S\otimes \id_R)\Gamma_{SR}\big) - (\mN_S \otimes \id_R)\big(\Gamma_{SR}(H_S\otimes \id_R)\big)$, $\Gamma_{SR} = \ket{\Gamma}_{SR}\bra{\Gamma}_{SR}$ and $\ket{\Gamma}_{SR} = \sum_i \ket{i}_S\ket{i}_R$. 

We now state and prove \thmref{thm:global-2-res} which provides lower bounds on $\gamma$ and $\bgamma$ by considering different input logical states $\ket{\psi_L}$. 
\begin{theorem}
\label{thm:global-2-res}
Consider a quantum code defined by $\mE_{\LtoS}$. Consider  physical Hamiltonian $H_S$,  logical Hamiltonian $H_L$, and noise channel $\mN_S$. Suppose $\mN_S$ commutes with $\mU_{S,\theta}$. Then it holds that
\begin{gather}
\label{eq:coh-1}
\epsilon + \delta_\group \geq \gamma \geq \ell_2\left(\frac{\Delta H_L}{\sqrt{4 F^{\txr}(\mN_{S,\theta})}}\right), \\
\label{eq:coh-2}
\bepsilon + \bdelta_\group \geq \bgamma \geq \ell_3\left(\sqrt{\frac{\frac{1}{d_L}\trace(H_L^2)-\frac{1}{d_L^2}\trace(H_L)^2}{{ F^{\txr}(\mN_{S,\theta})}}}\right), 
\end{gather}
where $\ell_2(x)= x + O(x^2)$ is the inverse function of the monotonic increasing function $x = \frac{y}{\sqrt{(1 - 3y^2 + y^4)(1 - 6\sqrt{2}y^2)}}$ for $y \in [0,1/(6\sqrt{2}))$ and $\ell_3(x)= x + O(x^2)$ is the inverse function of the monotonic increasing function $x = \frac{y}{\sqrt{(1 - 3y^2 + y^4)\left(1 - ({3(\Delta H_L)^2y})/{\sqrt{2}(\frac{1}{d_L}\trace(H_L^2)-\frac{1}{d_L^2}\trace(H_L)^2)}\right)}}$ for $y \in \left[0,\frac{\sqrt{2}(\frac{1}{d_L}\trace(H_L^2)-\frac{1}{d_L^2}\trace(H_L)^2)}{3(\Delta H_L)^2}\right)$. 
\end{theorem} 

In particular, when $\epsilon = 0$, i.e., when the code is exactly error-correcting, we have the following corollary: 
\begin{corollary}
\label{col:global-2-covariant-res}
Consider a quantum code defined by $\mE_{\LtoS}$. Consider  physical Hamiltonian $H_S$,  logical Hamiltonian $H_L$, and noise channel $\mN_S$. 
Suppose $\mN_S$ commutes with $\mU_{S,\theta}$.
When $\epsilon = 0$, i.e., when the code is exactly error-correcting, it holds that $\delta_\group \geq \ell_2(\Delta H_L/\sqrt{4 F^\txr(\mN_{S,\theta})})$ and $\bdelta_\group \geq \ell_3(\Delta H_L/\sqrt{4 F^\txr(\mN_{S,\theta})})$. 
\end{corollary}

\begin{proof}[Proof of {\thmref{thm:global-2-res}}]
Let $\ket{+_L} = \frac{\ket{0_L} + \ket{1_L}}{\sqrt{2}}$. Then according to \propref{prop:cov}, there exists a covariant recovery channel $\mR_{\StoL}^\cov$ such that 
\begin{equation}
1 - \braket{+_L|\rho_L|+_L} \leq \epsilon^2_\cov \leq \gamma^2,
\end{equation}
where $\rho_L = ( \mR_{\StoL}^\cov \circ \mN_S \circ \mE_{\LtoS}) (\ket{+_L}\bra{+_L})$. 
According to Lemma 3 in Ref.~\cite{zhou2020theory}, 
\begin{equation}
\label{eq:variance}
F^\txr(\rho_L,H_L) \geq  \frac{1 - 3\gamma^2 + \gamma^4}{\gamma^2} \cdot \left(\bV_{H_L}(\ket{+_L} - \frac{3\sqrt{2}\gamma(\Delta H_L)^2}{2}\right) ,
\end{equation}
where the variance $\bV_{H_L}(\ket{+_L}) = \braket{+_L|H_L^2|+_L} - \braket{+_L|H_L|+_L}^2 = \frac{(\Delta H_L)^2}{4}$. $\gamma < 1/(6\sqrt{2})$ guarantees the right-hand side is positive. On the other hand, using \eqref{eq:monotone-RLD},  
\begin{equation}
\label{eq:channel-RLD}
F^\txr(\rho_L,H_L) 
\leq F^\txr(\rho_S,H_S)= 
F^\txr(\mN_{S,\theta}(\mE_{\LtoS}(\ket{+_L}\bra{+_L})) ) \leq  F^\txr(\mN_{S,\theta}),
\end{equation}
where $\rho_S = (\mN_S \circ \mE_{\LtoS}) (\ket{+_L}\bra{+_L})$.
Using \eqref{eq:variance} and \eqref{eq:channel-RLD}, we have
\begin{equation}
\gamma^2 \cdot \frac{1}{(1 - 3\gamma^2 + \gamma^4)(1 - 6\sqrt{2}\gamma)} \geq \frac{(\Delta H_L)^2}{4 F^\txr(\mN_{S,\theta})}, 
\end{equation}
proving \eqref{eq:coh-1}. 

Similarly, let $\ket{\Psi_{LR}} = \frac{1}{\sqrt{d_L}}\sum_{i=1}^{d_L} \ket{i}_L \ket{i}_R$. Then according to \propref{prop:cov}, there exists a covariant recovery channel $\mR_{\StoL}^{\optCG}$ such that 
\begin{equation}
1 - \braket{\Psi_{LR}|\rho_{LR}|\Psi_{LR}} \leq \bepsilon_\cov \leq \bgamma,
\end{equation}
where $\rho_{LR} = ( \mR_{\StoL}^\cov \circ \mN_S \circ \mE_{\LtoS}) (\ket{\Psi_{LR}}\bra{\Psi_{LR}})$. 
According to Lemma 3 in Ref.~\cite{zhou2020theory}, 
\begin{align}
F^\txr(\rho_{LR},H_L\otimes\id_R) 
\geq\;& \frac{1 - 3\bgamma^2 + \bgamma^4}{\bgamma^2} \left( \bV_{H_L\otimes\id_R}(\ket{\Psi_{LR}}) - \frac{3\sqrt{2}\bgamma(\Delta H_L)^2}{2}\right) \nonumber\\
\geq\;& \frac{1 - 3\bgamma^2 + \bgamma^4}{\bgamma^2} \left( \frac{\trace(H_L^2)}{d_L}-\frac{\trace(H_L)^2}{d_L^2} - \frac{3\sqrt{2}\bgamma(\Delta H_L)^2}{2} \right). 
\end{align}
The rest of the proof is exactly the same as in the proof of the lower bound on the worst-case gate implementation error. 
\end{proof}

In fact, the proof of \thmref{thm:global-2-res} follows almost exactly from the proof of Theorem 2 in Ref.~\cite{zhou2020new} and our new contribution here is \propref{prop:cov}. 

To compare \thmref{thm:global-2-res}  to \thmref{thm:global-2-metrology} , we first note that $F^\txr(\Phi_\theta) \geq \barF(\Phi_\theta)$ for any $\Phi_\theta$ because $F^\txr(\rho_\theta) \geq F(\rho_\theta)$~\cite{petz2011introduction} for any $\rho_\theta$ and $F^\txr(\Phi_\theta) = \lim_{N\rightarrow \infty} F^\txr(\Phi_\theta^{\otimes N})/N$~\cite{hayashi2011comparison,katariya2020geometric}. 
Moreover, \thmref{thm:global-2-res} requires the commutativity between the noise and the Hamiltonian and $\ell_1(x) \geq \ell_2(x)$, so \thmref{thm:global-2-res} provides a weaker bound on the (worst-case) gate implementation error than \thmref{thm:global-2-metrology}. Note that $F^\txr(\mN_{S,\theta}) < +\infty$ only when ${\rm span}\{K_{S,i} H_S,\forall i\} \subseteq {\rm span}\{K_{S,i},\forall i\}$ which is also a stronger condition than the HKS condition. The resource theory method leads to a bound on the Choi gate implementation error, however, which is not available using the quantum metrology method.

Also note recent works Refs.~\cite{tajima2018uncertainty,tajima2020coherence} (results implied by Ref.~\cite{tajima2021symmetry}) which considered the coherence cost of implementing unitary gates based on relevant insights.

\subsection{Explicit behaviors of the bounds and comparison with the charge fluctuation approach}
\label{sec:comparison}

We first make a general comparison between the trade-off relations derived using the gate implementation error approach (\thmref{thm:global-2-metrology} and \thmref{thm:global-2-res}) and the charge fluctuation approach (\thmref{thm:global-1}) in \secref{sec:global-1}. Two clear advantages of the gate implementation error approach are that 1) it applies to general quantum codes (e.g., the non-isometric encodings in \cite{woods2019continuous,yang2020covariant}) while the charge fluctuation approach only holds for isometric codes; 2) it leads to a trade-off relation for the Choi measures. Additionally, for the special case of $\delta_\group = 0$, the results based on the gate implementation error approach  directly reduce to the previous known result for exactly covariant codes in Ref.~\cite{zhou2020new}, while there is still some discrepancy with previous results using the charge fluctuation approach (see more discussion in \secref{sec:global-1-tradeoff}). For the special case of $\epsilon = 0$ which was not previously studied, we have two lower bounds on $\delta$ from \corollaryref{col:global-2-covariant-metrology} and \corollaryref{col:global-1-covariant} which behave as follows:
\begin{align}
\delta_\group \gtrsim \frac{\Delta H_L}{2\sqrt{\frakF}}, \qquad 
\delta_\group \gtrsim \sqrt{\frac{\Delta H_L}{\Delta H_S}}. 
\end{align}
It is interesting to observe that the first bound depends on the noise channel while the second one does not (as long as the HKS condition is satisfied).

We now remark on the explicit scalings of our bounds for different noise models as in \secref{sec:global-1-noise}. Again, consider a $n$-partite system and a local physical Hamiltonian with $\Delta H_S = O(n)$. For random local noise which acts uniformly randomly on each subsystem, $\frakF = O(n^2)$ and the two bounds give $\delta_\group = \Omega(1/n)$ and $\delta_\group = \Omega(1/\sqrt{n})$, respectively. 
That is, the charge fluctuation approach outperforms the gate implementation error approach in this case. For noise acting independently on each subsystems, we have $\frakF = O(n)$ which gives $\delta_\group = \Omega(1/\sqrt{n})$ using the gate implementation error approach. In this situation, the bounds based on the two approaches are comparable. Note that in situations where $\frakF = o(n)$, i.e., the noise is even stronger than independent noise so that the regularized QFI is sublinear, the bound based the gate implementation error approach should outperform the bound based on the charge fluctuation approach. In general situations where both the QEC inaccuracy and the global covariance violation are non-vanishing, we expect a similar behavior, i.e., the gate implementation error approach performs better in the extremely strong noise regime, while the charge fluctuation approach performs better in weaker noise regimes.

\section{Limitations on transversal logical gates}
\label{sec:transversal}

Note that a key implication of our results is symmetry constraints on QEC codes that achieve a given accuracy, which extends the scope of previous knowledge on the incompatibility between symmetries and QEC  to general codes, especially exact QEC codes which are most commonly studied.
As we shall discuss in this section, such constraints actually allow us to derive restrictions on the transversally implementable gates for general QEC codes, advancing our understanding of fault tolerance. A key intuition is that the precision of gate implementation is associated with the degree of symmetry.
Recall that there are no QEC codes which admit a continuous symmetry acting transversally on physical qubits and thus there are no transversal universal gate sets, according to the Eastin--Knill theorem. For stabilizer codes, the incompatibility between QEC and symmetry are reflected in the classification of transversally logical gates in finite levels of the Clifford hierarchy~\cite{bravyi2013classification,pastawski2015fault,anderson2016classification,jochym2018disjointness}. Here we present new restrictions of transversal gates for arbitrary QEC codes from the perspective of global covariance violations.

The following corollary puts a restriction on the logical transversal gates using \corollaryref{col:global-1-covariant}. Namely, the non-trivial logical gates cannot be too close to the identity operators when implemented by transversal physical gates in the vicinity of identity operators because $\delta_\group$ has a lower bound of $\Theta(1/\sqrt{n})$. Note that here we implicitly consider exact QEC codes under single-erasure noise (so that the HKS condition is satisfied for 1-local Hamiltonians) and in this case \corollaryref{col:global-1-covariant} outperforms \corollaryref{col:global-2-covariant-metrology}, so we will only use \corollaryref{col:global-1-covariant} in this section. 

\begin{corollary}
\label{col:gate}
Suppose an $n$-qudit QEC code with distance at least $2$ admits a transversal implementation $V_S = \bigotimes_{l=1}^n e^{-i2\pi T_{S_l}/D}$ of the logical gate $V_L = e^{-i2\pi T_L/D}$ where $D$ is a positive integer and $T_{L,S}$ have integer eigenvalues. 
Then it holds that
\begin{equation}
D \leq \max\Bigg\{ 
4\pi\sqrt{\frac{2}{3}}\bigg(\Delta T_L + \sum_{l=1}^n \Delta T_{S_l}\bigg),
 {2\sqrt{2}\pi}\sqrt{ \frac{ \sum_{l=1}^n \Delta T_{S_l}}{\Delta T_L}} \left(\Delta T_L + \sum_{l=1}^n \Delta T_{S_l}\right) \Bigg\}.
\end{equation} 
In particular, $D = O(\poly(n))$ when $\Delta T_S = O(\poly(n))$, $\Delta T_L = O(\poly(n))$ and $\Delta T_S/\Delta T_L = O(\poly(n))$, where $\Delta T_S = \sum_{\ell=1}^n \Delta T_{S_\ell}$\footnote{ 
The conditions are satisfied in common settings; see, e.g., the proof of \corollaryref{col:gate-stab}. 
}.
\end{corollary}

\begin{proof}
Any codes with distance at least $2$ can correct single-erasure noise. Let $H_{L} = T_L$ and $H_S = \sum_l T_{S_l}$. They have integer eigenvalues implies that $\mU_{S,\theta}$ and $\mU_{L,\theta}$ share a common period $2\pi$. According to \corollaryref{col:global-1-covariant}, the code must satisfy
\begin{equation}
\label{eq:app-gate}
\delta_\group 
\geq \min\left\{ \frac{\sqrt{\Delta H_L\left(\Delta H_S -\frac{1}{2}\Delta H_L\right)}}{\Delta H_S}, \sqrt{\frac{3}{8}}\right\}.
\end{equation}
We can always write $\theta = \frac{2j\pi}{D} + \theta_1$ for some $j\in\bN$ and $\theta_1 \in [0,2\pi/D)$. Then we have 
\begin{align}
 P(\mU_{S,\theta}\circ \mE_{\LtoS} \circ \mU_{L,\theta}^\dagger,\mE_{\LtoS}) = \;& P(\mU_{S,\theta_1}\circ \mE_{\LtoS} \circ \mU_{L,\theta_1}^\dagger,\mE_{\LtoS}) \nonumber\\
\leq\;&  P(\mU_{S,\theta_1}\circ \mE_{\LtoS} \circ \mU_{L,\theta_1}^\dagger,\mU_{S,\theta_1}\circ\mE_{\LtoS})  + P(\mU_{S,\theta_1}\circ\mE_{\LtoS},\mE_{\LtoS})\nonumber\\
\leq\;& P(\mU_{L,\theta_1}^\dagger,\id_L) + P(\mU_{S,\theta_1},\id_S), 
\end{align}
where we use the monotonicity and the triangular inequality of the purified distance. Without loss of generality, assume $H_L = \sum_{i=1}^{d_L} \lambda_i \ket{i}\bra{i}$ and let $\ket{\psi} = \sum_{i=1}^{d_L} \psi_i \ket{i}\ket{i}$. Consider first the situation where $D > 2\max\{\Delta H_L,\Delta H_S\}$, then we have $\theta_1 \max\{\Delta H_L,\Delta H_S\} < \pi$ and 
\begin{equation}
\begin{split}
P(\mU_{L,\theta_1}^\dagger,\id_L) =&\; \max_{\{\psi_i\}} P((\mU_{L,\theta_1}^\dagger\otimes \id)(\ket{\psi}\bra{\psi}),\ket{\psi}\bra{\psi}) \\
=&\; \max_{\{\psi_i\}} \!\sqrt{1 \!- \!\bigg|\!\sum_{i} \abs{\psi_i}^2 e^{i\theta_1 \lambda_i}\!\bigg|^2} = \max_{i,i'} \abs{\sin\left(\frac{1}{2}\theta_1(\lambda_i-\lambda_{i'})\right)} \leq \theta_1 \Delta H_L. 
\end{split}
\end{equation}
Similarly, $P(\mU_{S,\theta_1},\id_S) \leq \theta_1 \Delta H_S = \theta_1 \sum_{l=1}^n \Delta H_{S_l}$. Since $\delta_\group = \max_\theta P(\mU_{S,\theta}\circ \mE_{\LtoS} \circ \mU_{L,\theta}^\dagger,\mE_{\LtoS}) = \max_{\theta_1} P(\mU_{S,\theta_1}\circ \mE_{\LtoS} \circ \mU_{L,\theta_1}^\dagger,\mE_{\LtoS})$, we obtain 
\begin{equation}
\label{eq:app-gate-1}
\delta_\group \leq \frac{2\pi}{D} \left(\Delta H_L + \sum_{l=1}^n \Delta H_{S_l}\right), 
\end{equation}
Otherwise, 
\begin{equation}
\label{eq:app-gate-2}
D \leq 2\max\{\Delta H_L,\Delta H_S\}.     
\end{equation}
The result  then follows by combining \eqref{eq:app-gate}, \eqref{eq:app-gate-1} and \eqref{eq:app-gate-2}. 
\end{proof}

\corollaryref{col:gate} shows that the precision of transversal logical  gates under certain restrictions only increases polynomially in the number of qubits. For the important case of stabilizer codes, this implies that the levels of the Clifford hierarchy that can be reached only increase polynomially in the number of qubits. 
Specifically, 
consider an $n$-qubit stabilizer code with distance at least $2$. 
The following corollary describes the limitation on the transversally implementable logical gates for stabilizer codes:

\begin{corollary}
\label{col:gate-stab}
Let $\tilde{V}_S =  V_1 (\bigotimes_{l=1}^n e^{-i 2\pi a_l Z_l/D}) V_2$ be a transversal logical gate for an $n$-qubit stabilizer code with distance at least $2$, where $D$ is a power of two and $a_l$ is an integer and $V_{1,2}$ are transversal Clifford operators. (This describes the most general form of transversal logical gates for stabilizer codes~\cite{zeng2011transversality,anderson2016classification}). When $a_l = O(\poly(n))$, we must have $D = O(\poly(n))$ and $\tilde{V}_S$ implements a logical gate $\tilde V_L$ in the \mbox{$O(\log n)$-th} level of the Clifford hierarchy. 
\end{corollary}

\begin{proof}
Let $\Pi$ be the projection onto the stabilizer code under consideration and $Q_S = \bigotimes_{l=1}^n e^{-i 2\pi a_l Z_l/D}$. Then
\begin{equation}
\tilde{V}_S \Pi = \Pi\tilde{V}_S\Pi,~~\Rightarrow~~
Q_S \Pi_2  =
\Pi_1 Q_S \Pi_2 ,
\end{equation}
where $\Pi_1 = V_1^\dagger\Pi V_1$ and $\Pi_2 = V_2\Pi V_2^\dagger$. Both $\Pi_1$ and $\Pi_2$ are stabilizer codes with the same code distance as $\Pi$. Without loss of generality, we assume $\Pi$, $\Pi_1$ and $\Pi_2$ are two-dimensional stabilizer codes (by considering subcodes of the original codes). 

As proven in Proposition 4 in Ref.~\cite{anderson2016classification}, $V_S = (Q_S)^4$ must be a logical gate on $\Pi_2$, satisfying 
\begin{equation}
(Q_S)^4 \Pi_2 = \Pi_2 (Q_S)^4 \Pi_2,
\end{equation} 
and the logical gate $V_L$ has the form $V_L = e^{-i2\pi a Z_L/D}$ where $a$ is an integer. First consider the situation where $a = 0$ (for any choice of two-dimensional codes), i.e., $(Q_S)^4 \Pi_2 = \Pi_2$. By writing down the stabilizer code $\Pi_2$ in its computational basis, it is easy to observe that either $Q_S \Pi_2 = \Pi_2$, then $\tilde{V}_S$ implements a Clifford logical gate and the Corollary holds, or $Q_S \Pi_2 \neq \Pi_2$, then $\sum_{l=1}^n a_l/D$ must be a positive constant and $D = O(\sum_{l=1}^n a_l) = O(\poly(n))$.

Now we consider the situation where $a \neq 0$. By writing down the stabilizer code $\Pi_2$ in its computational basis, we observe $1 \leq a \leq \sum_{l=1}^n 4a_l = O(\poly(n))$. Let $H_L = a Z_L$ and $H_{S_l} = 4a_l Z_l$. Then we must have 
\begin{equation}
\frac{\sqrt{\Delta H_L\left(\Delta H_S -\frac{1}{2}\Delta H_L\right)}}{\Delta H_S} \geq \frac{1}{O(\poly(n))}.
\end{equation}
Using \corollaryref{col:gate}, we have $D = O(\poly(n))$. 
Since $D$ is a power of 2, for all $l$, $e^{-i 2\pi a_l Z_l/D}$  (see Proposition 1 in Ref.~\cite{anderson2016classification}) and thus $Q_S = \bigotimes_{l=1}^n e^{-i 2\pi a_l Z_l/D}$ must be in the \mbox{$(\log D)$-th} level of the Clifford hierarchy. \corollaryref{col:gate-stab} then follows from the fact that  Clifford operators $V_1$, $V_2$ preserve the level of the Clifford hierarchy and any physical gate in the \mbox{$j$-th} level of the Clifford hierarchy implements a logical gate in the \mbox{$j$-th} level of the Clifford hierarchy (because logical Pauli operators can be implemented by physical Pauli operators for stabilizer codes). 

\end{proof}

\corollaryref{col:gate-stab} provides a simple proof on the limitations of transversal logical gates for stabilizer codes from the perspective of continuous symmetries. Note that the relevant results previously known for stabilizer codes \cite{bravyi2013classification,pastawski2015fault,jochym2018disjointness} were obtained using very different techniques.

\section{Trade-off between QEC and local symmetry measures}
\label{sec:local}

In this section, we study relations between QEC and local symmetry measures, that is,  the local covariance violation and the charge conservation violation. We will first prove a lemma which links the charge conservation to the charge fluctuation and then derive trade-off relations using \propref{prop:charge-KL} and \propref{prop:charge-metrology}. We will also derive a lower bound on the local covariance violation using the quantum metrology method.

Note that the results in this section (\thmref{thm:local-KL}, \thmref{thm:local-metrology} and \thmref{thm:local-point}) hold true for arbitrary Hermitian operators $H_L$ and $H_S$, which do not necessarily share a common period as generators of $U(1)$ representations. 

\subsection{Bounds via charge fluctuation}
We first observe a simple connection between the charge fluctuation $\chi$ and the charge conservation violation $\delta_\charge$: 
\begin{lemma}
\label{lemma:chi-local} Consider a quantum code $\mE_{\LtoS}$, a physical Hamiltonian $H_S$ and a logical Hamiltonian $H_L$. Then 
\begin{equation}
\abs{\chi} \geq \Delta H_L - \delta_\charge.
\end{equation} 
\end{lemma}

\begin{proof}
By definition, $\abs{\chi} = \big|\!\braket{0_L|(\mE_{\LtoS})^\dagger(H_S)|0_L} - \braket{1_L|(\mE_{\LtoS})^\dagger|1_L}\!\big|$. Then we must have $\abs{\chi} \geq \Delta H_L - \delta_\charge$, because $\delta_\charge \geq \big|\!\braket{0_L|((\mE_{\LtoS})^\dagger(H_S)-H_L)|0_L} - \braket{1_L|((\mE_{\LtoS})^\dagger(H_S)-H_L)|1_L}\!\big|$, and $\Delta H_L = \big|\!\braket{0_L|H_L|0_L} - \braket{1_L|H_L|1_L}\!\big|$. 
\end{proof}

Using the KL-based method (\propref{prop:charge-KL}) and \propref{prop:point-charge}, we immediately have the following trade-off relations: 
\begin{theorem}
\label{thm:local-KL}
\isometric It holds that 
\begin{align}
\delta_{\point} + 2\epsilon\frakJ &\geq \Delta H_L,\\
\delta_{\charge} + 2\epsilon\frakJ &\geq \Delta H_L,
\label{eq:tradeoff-3}
\end{align} 
where $\frakJ$ is given by \eqref{eq:def-frakJ}. 
\end{theorem}
Note that \eqref{eq:tradeoff-3} reduces to Corollary 3 in Ref.~\cite{faist2019continuous} for random local erasure noise.

In particular, for exact QEC codes, we have the following corollary: 
\begin{corollary}
\label{col:local-KL}
\isometric  When $\epsilon = 0$, i.e., when the code is exactly error-correcting, we must have $\delta_\point \geq \Delta H_L$ and $\delta_\charge \geq \Delta H_L$. 
\end{corollary}

Similarly, using the quantum metrology method (\propref{prop:charge-metrology}) and \propref{prop:point-charge}, we have the following trade-off relations: 
\begin{theorem}
\label{thm:local-metrology}
\nonisometric It holds that 
\begin{align}
\delta_{\charge} + 2\epsilon\left(\sqrt{(1-\epsilon^2)\frakF
} + \frakB \right) &\geq \Delta H_L. 
\end{align} 
In particular, when $\epsilon \ll 1$ and $\frakB \ll \sqrt{\frakF}$, we have 
\begin{align}
\delta_{\charge} + 2 \epsilon \sqrt{\frakF} 
&\gtrsim \Delta H_L,
\end{align}
where $\frakF$ and $\frakB$ are given by \eqref{eq:def-frakF} and \eqref{eq:def-frakB}, respectively. 
Furthermore, when the code is isometric, it holds that 
\begin{equation}
\delta_{\point} + 2\epsilon\left(\sqrt{(1-\epsilon^2)\frakF} + \frakB \right) \geq \Delta H_L. 
\end{equation}
\end{theorem}
\begin{corollary}
\label{col:local-metrology}
\nonisometric When $\epsilon = 0$, i.e., when the code is exactly error-correcting, we must have $\delta_\charge \geq \Delta H_L$. 
\end{corollary}

Note that \corollaryref{col:local-metrology} is slightly more general than \corollaryref{col:local-KL} as the former covers the situation where the encoding is non-isometric. 

\subsection{Bounding local covariance violation using quantum metrology}

The trade-off relation between $\epsilon$ and $\delta_\point$ in \thmref{thm:local-metrology} requires the code to be isometric. 
In fact, we can show a cleaner version of the trade-off between the QEC inaccuracy and the local covariance violation using the quantum metrology method which does not contain $\frakB$ and also covers the non-isometric scenario, as shown below.  
\begin{theorem}
\label{thm:local-point}
\nonisometric It holds that 
\begin{equation}
\label{eq:tradeoff-2}
\delta_\point + 2\epsilon\left(\sqrt{(1-\epsilon^2)\frakF} + \epsilon \Delta H_L\right) \geq \Delta H_L. 
\end{equation}
When $\epsilon \ll 1$, we have
\begin{equation}
\delta_\point + 2\epsilon\sqrt{\frakF}
\gtrsim \Delta H_L, 
\end{equation} 
where $\frakF$ and $\frakB$ are given by \eqref{eq:def-frakF} and \eqref{eq:def-frakB}, respectively. 
\end{theorem}
\begin{corollary}
\label{col:local-metrology-2}
\nonisometric When $\epsilon = 0$, i.e., when the code is exactly error-correcting, it holds that $\delta_\point \geq \Delta H_L$. 
\end{corollary}

\begin{proof}
Let $\mR^{\optL}_{\StoL}$ be the recovery channel such that 
$\epsilon =  P(\mR^{\optL}_{\StoL}\circ \mN_{S}\circ\mE_{\LtoS}, \id_{L})$. Let (see \figref{fig:dephasing}) 
\begin{align}
\mN_{C,\theta} =\;&\mR^{\rep}_{\SAtoC} \circ ( \mR_{\StoL}^{\optL} \circ\mN_{S} \otimes \id_A) \circ \nonumber\\&   (\mU_{S,\theta}\circ \mE_{\LtoS} \otimes \id_A) \circ \mE_{\CtoLA}^{\rep}. 
\end{align}
Consider the parameter estimation of $\theta$ in the neighborhood of $\theta = 0$ and let $\xi_\theta = \bra{0_C}\mN_{C,\theta}(\ket{0_C}\bra{1_C})\ket{1_C}$. Following the proof of \propref{prop:charge-metrology}, we have from \eqref{eq:xi-epsilon} that 
\begin{equation}
\abs{\xi_{\theta = 0}} \geq 1 - 2\epsilon^2. 
\end{equation}
As shown in \appref{app:local-proof}, we also have 
\begin{equation}
\abs{\partial_\theta\xi_\theta}^2\big|_{\theta = 0}  \geq ((1-2\epsilon^2)\Delta H_L - \delta_\point)^2, 
\end{equation}
when $(1-2\epsilon^2)\Delta H_L \geq \delta_\point$. 
Hence, when $(1-2\epsilon^2)\Delta H_L \geq \delta_\point$, we must have 
\begin{equation}
\barF(\mN_{S,\theta}) 
= \barF(\mN_{S,\theta}) \big|_{\theta = 0} 
\geq \barF(\mN_{C,\theta})\big|_{\theta = 0} \geq \frac{((1-2\epsilon^2)\Delta H_L - \delta_\point)^2}{4\epsilon^2(1-\epsilon^2)},
\end{equation} 
completing the proof. 
\end{proof}

\subsection{Remarks on the behaviors of the bounds}

From \thmref{thm:local-KL}, \thmref{thm:local-metrology} and \thmref{thm:local-point}, we observe that in $n$-partite systems, the local covariance violation $\delta_\point$ and the charge conservation violation $\delta_\charge$ are usually lower bounded by constants for small $\epsilon$ which does not vanish
as $n\rightarrow \infty$ like the global covariance violation $\delta_\group$. However, also note that $\delta_\point$ and $\delta_\charge$ may naturally be superconstant (for example, for the trivial encoding $\mE_{\LtoS} = \id$ we usually have $\delta_\point=\delta_\charge=\Delta(H_S-H_L) = \Theta(n)$), indicating that the constant or even sublinear scaling of $\delta_\point$ and $\delta_\charge$ requires non-trivial code structures.

Also note that the bounds on $\delta_\point$ in both \thmref{thm:local-KL} and \thmref{thm:local-metrology} rely on the fact that $\delta_\point \geq \delta_\charge$, indicating that these bounds may not be tight when there is a gap between $\delta_\point$ and a function of $\delta_\charge$. Such a gap does exist as we shown later in examples (see \secref{sec:case-study}) and we provide a possible explanation of the existence of the gap in \appref{app:refine}.

\section{Case studies of explicit codes}
\label{sec:case-study}

In the above, we have derived several forms of fundamental limits on the QEC accuracy and degree of symmetry or charge conservation that a quantum code can possibly admit.  Then a natural question is to what extent these limits can be attained by certain codes.  Furthermore, explicit constructions of approximately covariant codes would be important for our understanding of the QEC-symmetry trade-off and may find broad applications.   In this section, we introduce and analyze two code examples with interesting approximate covariance features to address these needs.  In the first example, we generalize a covariant code called the thermodynamic code~\cite{brandao2019quantum,faist2019continuous} to a class of general quantum codes which exhibits a full trade-off between symmetry and QEC via a smooth transition from exact covariance to exact QEC. The second one involves a well-known QEC code called the quantum Reed--Muller  codes~\cite{steane1999quantum,macwilliams1977theory}, which can be seen as a prominent example of approximately covariant exact QEC codes.  In particular, we explicitly compute their QEC and symmetry measures, and compare them to the fundamental limits. Remarkably, the scalings of the global covariance violation and the charge conservation violation for both examples match well with the optimal scalings from our bounds.  

\subsection{Modified thermodynamic codes}
\label{sec:thermo}

Thermodynamic codes \cite{brandao2019quantum,faist2019continuous} are $n$-qubit quantum codes given by certain Dicke states with different magnetic charges which become approximately quantum error-correcting for large $n$. Specifically, a two-dimensional thermodynamic code have codewords
\begin{equation}
\ket{\frakc_0} = \ket{m_n} , \quad 
\ket{\frakc_1} = \ket{(-m)_n} ,  
\end{equation}
where $\ket{m_n}$ for $m \in [-n,n]$ is the Dicke state defined by
\begin{equation}
\ket{m_n} = \frac{1}{\sqrt{\binom{n}{\frac{n+m}{2}}}} \sum_{\substack{\vj\in\{0,1\}^n:\\\sum_l j_l = \frac{n+m}{2}}}\ket{\vj}, 
\end{equation} 
satisfying $\sum_{l=1}^n Z_l\ket{m_n} = -m$. Note that $m+n$ must be an even number and we assume $2\leq m \ll n$. It is easy to verify that the thermodynamic code is exactly covariant with respect to $H_L = \frac{m}{2}Z_L$ and $H_S = - \frac{1}{2} \sum_{l=1}^n Z_l$ and it was proven that for single-erasure noise~\cite{zhou2020theory} $\epsilon = m/2n + O(m^2/n^2)$ which is infinitely small when $m/n\rightarrow 0$. 
Here we extend the thermodynamic code in such a way that it  smoothly transitions from an exactly covariant code to an exact QEC code as tuned by a continuous parameter $0\leq q \leq 1$. Specifically, our modified thermodynamic code is defined by 
\begin{align}
\ket{\frakc^q_0} &= \sqrt{\frac{n}{n+qm}} \ket{m_n} + \sqrt{\frac{qm}{n+qm}} \ket{(-n)_n}, \\
\ket{\frakc^q_1} &= \sqrt{\frac{n}{n+qm}} \ket{(-m)_n} + \sqrt{\frac{qm}{n+qm}} \ket{n_n}.  
\end{align}
In particular, when $q = 0$, we have the original thermodynamic code, and when $q = 1$, we obtain an modified code which is exactly error-correcting under single-erasure noise.
We shall compute the QEC inaccuracy and the different covariance violation measures, and compare them with our trade-off bounds. 

\subsubsection{QEC inaccuracy}

Here we compute the QEC inaccuracy of modified thermodynamic codes $\epsilon(\mN_{S},\mE_{\LtoS})$ where $\mE_{\LtoS}(\cdot) = W(\cdot)W^\dagger$ with $W=\ket{\frakc^q_0}\bra{0_L} + \ket{\frakc^q_1}\bra{1_L}$ and $\mN_S$ is the single-erasure noise channel $\mN_S = \sum_{l=1}^n \frac{1}{n} \mN_{S_l}$, where $\mN_{S_l}(\cdot) = \ket{\vac}\bra{\vac}_{S_l} \otimes \trace_{S_l}(\cdot)$.

We need to use the following lemma which compute the purified distance between error-corrected channels, employing the formalism of complementary channels~\cite{beny2010general}.  Let $\widehat \Phi_{B\leftarrow A}(\cdot) = \trace_A(V_{AB\leftarrow A}(\cdot)V_{AB\leftarrow A}^\dagger)$ 
be the complementary channel of channel $\Phi_A(\cdot) = \trace_B(V_{AB\leftarrow A}(\cdot)V_{AB\leftarrow A}^\dagger)$, 
where $V_{AB\leftarrow A}$ is a Stinespring dilation of $\Phi_A$. Then we have
\begin{lemma}[\cite{beny2010general}]
\label{lemma:complementary}

\begin{equation}\min_\mR P(\mR\circ\Phi_1,\Phi_2) = \min_{\mR'} P(\widehat{\Phi}_1,\mR' \circ \widehat{\Phi}_2),
\end{equation}
for arbitrary $\Phi_{1,2}$, where the minimizations are taken over all channels with the appropriate input and output spaces.
\end{lemma}
Choosing $\Phi_1 = \mN_{S} \circ \mE_{\LtoS}$ and $\Phi_2 = \id_L$, 
we have 
\begin{equation}
\epsilon = 
\min_{\mR'_{\LtoB}} P(\widehat{\mN}_{\StoB}\circ\mE_{\LtoS},\mR'_{\LtoB}\circ \widehat{\id}_{L}), 
\end{equation}
As detailed in \appref{app:thermo}, we have that 
\begin{equation}
\label{eq:thermo-epsilon}
\epsilon =  \frac{(1-q)m}{2n} + O\left(\frac{m^2}{n^2}\right),
\end{equation}
and furthermore, explicitly construct a recovery channel $\mR^{\optL}_{\StoL}$ which achieves the optimal QEC inaccuracy up to the lowest order of $m/n$:
\begin{align}
\tilde \epsilon &= P(\mR^{\optL}_{\StoL}\circ\mN_S\circ\mE_{\LtoS},\id_L) = \sqrt{\frac{1}{2}-\sqrt{\frac{(n+m)(n+(2q-1)m)}{4(n+qm)^2}}} \approx \frac{(1-q)m}{2n}.
\end{align}

\subsubsection{Symmetry violation measures}

We now compute all the approximate symmetry measures associated with our modified thermodynamic codes.  Note that we let $H_L = \frac{m}{2}Z_L$ and $H_S = -\frac{1}{2}\sum_{l=1}^n Z_l$, which guarantees that the code tends to be covariant as $n\rightarrow\infty$.

We first compute $\delta_\group$ and $\delta_\point$. Let $\ket{\psi} = \ket{0_L}\ket{\psi_R^0} + \ket{1_L}\ket{\psi_R^1}$ be an arbitrary pure state on $L \otimes R$. Then 
\begin{align}
\ket{\psi_\theta} &:= U_{S,\theta} W U_{L,\theta}^\dagger \ket{\psi} \nonumber\\&=
\left(\!\sqrt{\frac{n}{n\!+\!qm}} \ket{m_n} \!+\! e^{\frac{i(m+n) \theta}{2}} \! \sqrt{\frac{qm}{n\!+\!qm}} \ket{(\!-n\!)_n}\!\right)\!\ket{\psi_R^0}  + \left(\!\sqrt{\frac{n}{n\!+\!qm}} \ket{(\!-m\!)_n}  
\!+\! e^{\frac{-i(m+n)\theta}{2}}\! \sqrt{\frac{qm}{n\!+\!qm}} \ket{n_n}\!\right)\!\ket{\psi_R^1},
\end{align}
and 
\begin{equation}
P(\mU_{S,\theta}\circ\mE_{\LtoS},\mE_{\LtoS}\circ\mU_{L,\theta}) = \max_{\psi} P(\ket{\psi_\theta},\ket{\psi})
= \sqrt{1 - \left|\frac{n + qm \cos((m+n)\theta/2)}{n + qm}\right|^2},
\end{equation}
where the maximum of $P(\ket{\psi_\theta},\ket{\psi})$ is attained at $\ket{\psi} = (\ket{0_L}+\ket{1_L})/\sqrt{2}$ (here the reference system can be one-dimensional, namely $\ket{\psi_R^0} = \ket{\psi_R^1} = 1/\sqrt{2}$, because $U_{S,\theta} W U_{L,\theta}^\dagger \ket{0_L}$ does not overlap with $U_{S,\theta} W U_{L,\theta}^\dagger \ket{1_L}$ for all $\theta$). 
Therefore, the global covariance violation is given by 
\begin{align} 
\delta_\group 
=\;& \max_\theta P(\mU_{S,\theta}\circ\mE_{\LtoS},\mE_{\LtoS}\circ\mU_{L,\theta}) = \sqrt{1 - \left(\frac{n-qm}{n+qm}\right)^2} = \sqrt{\frac{4qm}{n}} + O\left(\left(\frac{m}{n}\right)^{3/2}\right). 
\label{eq:thermo-delta}
\end{align}
The code is exactly covariant when $\theta = \frac{4k\pi}{m+n}$ and $k$ is an integer, and the corresponding local covariance violation is given by
\begin{align}
\delta_\point
=\;& \sqrt{2 \partial^2_\theta P(\mU_{S,\theta}\circ\mE_{\LtoS},\mE_{\LtoS}\circ\mU_{L,\theta})^2\big|_{\theta = \frac{4k\pi}{m+n},\forall k \in \bZ}} = \sqrt{\frac{qm(m+n)^2}{n+qm}} = \sqrt{qmn} + O\left(m\sqrt{\frac{m}{n}}\right), 
\end{align}

To compute the charge conservation violation $\delta_\charge$, note that $\mE_\StoL^\dagger(H_S) = \frac{mn(1-q)}{2(n+qm)} Z_L$ and $\mE_\StoL^\dagger(H_S^2) = \frac{mn(m+qn)}{4(n+qm)}\id_L$ so we have 
\begin{align}
\delta_\charge =\;& \Delta\left({H_L - (\mE_{\LtoS})^\dagger(H_S)}\right) 
= \frac{qm(n+m)}{(n+qm)} = qm + O\left(m\cdot\frac{m}{n}\right). 
\end{align}
Also note that the parameter $\frakB$ which shows up in \thmref{thm:global-1} and \thmref{thm:local-metrology} is given by
\begin{align}
\frakB = 2\sqrt{2}\max_{\ket{\psi}}\sqrt{\braket{\psi|\mE^\dagger(H_S^2)|\psi}-|\!\braket{\psi|\mE^\dagger(H_S)|\psi}\!|^2} 
= \sqrt{\frac{2mn(m+qn)}{n+qm}}
 = \sqrt{2qnm} + O\left(m\sqrt{\frac{m}{n}}\right),
\end{align}
when $qn \gg m$ and $\frakB = O(m)$ otherwise.

\subsubsection{Trade-off between QEC and symmetry, and explicit comparisons with lower bounds}

\begin{figure}[tb]
	\center
	\includegraphics[width=0.4\textwidth]{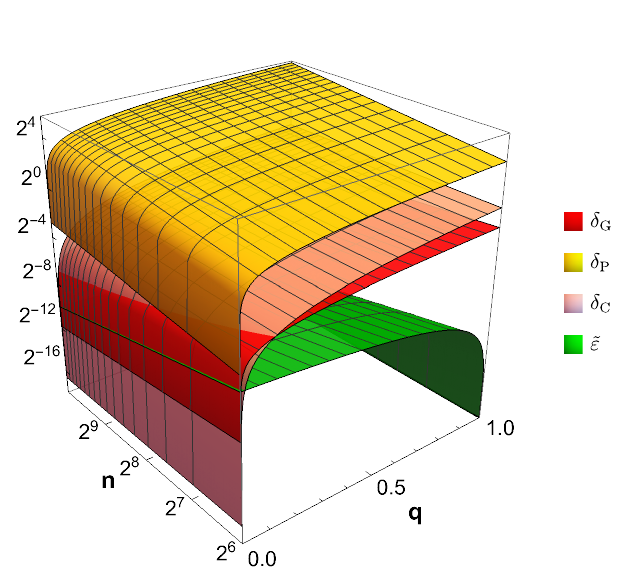}
	\caption{\label{fig:thermo}  Various types of symmetry measures and QEC accuracy of thermodynamic codes, where $m = 2$, $n \in [2^6,2^{10}]$ and $q \in [10^{-5},1-10^{-5}]$. Here we plot $\delta_\group$, $\delta_\point$, $\delta_C$ and $\tilde{\epsilon}$ (which is approximately equal to $\epsilon$ for large $n$). As $q$ increases from $0$ to $1$, $\tilde{\epsilon}$ decreases while the symmetry measures increase. We can also see that fixing $q$, the slopes of $\delta_\group$, $\delta_\point$, $\delta_C$, and $\tilde{\epsilon}$ with respect to $n$ are $-1/2$, $1/2$, $0$, and $-1$, respectively, matching our calculations. 
	}
\end{figure}

Let us first overview the behavior of modified thermodynamic codes.
Our calculations above indicate that up to the leading order, $\epsilon \simeq (1-q)m/2n$, while $\delta_\group \simeq \sqrt{4qm/n}$, $\delta_\point \simeq \sqrt{qmn}$, and $\delta_\charge \simeq qm$ (see \figref{fig:thermo}).  That is, as $q$ varies from $0$ to $1$, the symmetry violation (in terms of different measures) and the QEC inaccuracy exhibit trade-off behaviors---the former increases from $0$ while the latter decreases to $0$.

We now discuss the comparison with our lower bounds, focusing on the large $n$ asymptotics. Note that $H_L = \frac{m}{2}Z_L$ and $H_S = -\frac{1}{2}\sum_{l=1}^n Z_l$, so we have $\Delta H_L = m$ and $\Delta H_{S_l} = 1$ for each $l$. For single-erasure noise channels (as shown in \secref{sec:global-1-noise}), we have $\frakJ \leq n \max_l \Delta H_{S_l} = n$ and $\frakF \leq n\sum_{l=1}^n (\Delta H_{S_l})^2 = n^2$, and \thmref{thm:global-1} then gives: 
\begin{equation}
\label{eq:thermo-tradeoff}
\delta_\group \geq \sqrt{\frac{m-2n\epsilon}{n}} + O\left(\left(\frac{m}{n}\right)^{3/2}\right). 
\end{equation}
Plugging in the QEC inaccuracy $\epsilon \simeq \sqrt{(1-q)m/n}$, we have 
\begin{equation}
\delta_\group \geq \sqrt{\frac{qm}{n}} +  O\left(\left(\frac{m}{n}\right)^{3/2}\right). 
\end{equation}
Recall that for the modified thermodynamic code we have $\delta_\group \simeq \sqrt{4qm/n}$, which saturates this lower bound on $\delta_\group$ up to a constant factor of 2 in the leading order of $m/n$. 
Similarly, we could also plug the actual value $\delta_\group \simeq \sqrt{4qm/n}$ into \eqref{eq:thermo-tradeoff} and obtain the lower bound
\begin{equation}
\epsilon \geq \frac{(1-4q)m}{2n} + O\left(\frac{m^2}{n^2}\right),
\end{equation}
which shows that the  actual value $\epsilon \simeq (1-q)m/2n$  of the modified thermodynamic code saturates this lower bound up to a constant factor in the leading order of $m/n$ for $q < 1/4$.

For the local symmetry measures, we first note that for  the modified thermodynamic code with $q>0$ we have $\delta_\point/m \simeq \sqrt{qn/m}$ which becomes larger than $1$ as $m/n\rightarrow 0$, thus \thmref{thm:local-point} is not saturated. We provide one possible explanation of this gap between $\delta_\point$ and its lower bound in \appref{app:refine}, where we show a refinement of \thmref{thm:local-point} by replacing $\delta_\point$ in \thmref{thm:local-point} with $\delta_\point^\star$ ($\leq \delta_\point$) which is defined using the QFI of the error-corrected channel $\mR^\optL_{\StoL}\circ\mN_{S,\theta}\circ\mE_{\LtoS}\circ\mU_{L,\theta}^\dagger$ instead of the QFI of $\mU_{S,\theta}\circ\mE_{\LtoS}\circ\mU_{L,\theta}^\dagger$ at $\theta = 0$. We show that the gap between $\delta_\point$ and its lower bound could be explained by its gap with $\delta_\point^*$, explaining the looseness of \thmref{thm:local-point}.

Recall that for the modified thermodynamic code the charge conservation violation is $\delta_\charge \simeq qm$.  From \thmref{thm:local-KL}, we have 
\begin{equation}
\delta_\charge \geq \Delta H_L - 2\epsilon\frakJ = qm +  O\left(\frac{m^2}{n}\right).
\end{equation}
Namely, $\delta_\charge$ exactly saturates this lower bound in the leading order of $m/n$.

Note that $\frakB = O(\sqrt{n})$ and $\frakF = O(n^2)$ satisfies $\frakB \ll \sqrt{\frakF}$ in \thmref{thm:global-1} and \thmref{thm:local-metrology} so that $\frakB$ is negligible in the trade-off relations from \thmref{thm:global-1} and \thmref{thm:local-metrology} for modified thermodynamic codes in the large $n$ asymptotics.

Finally, note that the trade-off relation given by the diamond distance, 
\begin{equation} 
\delta_{\group,\diamond}^2 + 2\sqrt{\epsilon_\diamond} \geq \delta_\group^2 + 2\epsilon \geq \frac{m}{n}, 
\end{equation}
which follows from \thmref{thm:global-1} and the discussion in \secref{sec:def}, is also saturated up to a constant factor because $\delta_{\group,\diamond} = \delta_\group$ and $\epsilon_\diamond = \epsilon^2 + O(m^3/n^3)$ (see \appref{app:thermo} for details).

\subsection{Quantum Reed--Muller codes}
\label{sec:RM}

Reed--Muller codes constitute a family of error-correcting codes of great theoretical and technological interest.
The classical Reed--Muller code $R(s,t)$~\cite{macwilliams1977theory} is a $[2^t,\sum_{i=0}^s\binom{t}{i},2^{t-s}]$ code whose codewords correspond to Boolean functions of $t$ variables of degree at most $s$. Then the shortened Reed--Muller codes $\overline{R(s,t)} = [2^t-1,\sum_{i=1}^s\binom{t}{i},2^{t-s}]$ are obtained by selecting the codewords of $R(s,t)$ whose first digits are $0$ and deleting their first digits.  

The generalization to the quantum regime based on the stabilizer formalism and CSS construction, which leads to the quantum Reed--Muller codes, are also an important type of QEC codes~\cite{steane1999quantum}. Given the nice structures and features of quantum Reed--Muller codes, they provide a natural platform for understanding code properties. {For example, quantum Reed--Muller codes were widely applied in magic state distillation and implementing transversal non-Clifford operations~\cite{bravyi2012magic,anderson2014fault,haah2018codes,hastings2018distillation}. Quantum Reed--Muller codes were also known to reach the highest level of the Clifford hierarchy possible under the disjointness restriction~\cite{jochym2018disjointness}.}
Here we consider the $[[n = 2^t - 1,1,3]]$ quantum Reed--Muller code, which is a  CSS code~\cite{nielsen2002quantum} whose $X$ stabilizers correspond to $\overline{R(1,t)}$ and $Z$ stabilizers correspond to $\overline{R(t-2,t)}$. It is exactly error-correcting under single-erasure noise and admits a transversal implementation $\bigotimes_l \big(e^{i\pi Z_l/2^{t-1}}\big)$ of the logical operator $e^{-i\pi Z_L/2^{t-1}}$. 
 We now compute its symmetry violation measures. 
This code has the following form in the computation basis: 
\begin{align}
\ket{\frakc_0} & 
= \frac{1}{\sqrt{2^t}} \bigg(\ket{\v{0}} + \sum_{\v{x} \in \overline{R(1,t)}\backslash\{\v{0}\}} \ket{\v{x}}\bigg), \\
\ket{\frakc_1} &
= \frac{1}{\sqrt{2^t}} \bigg(\ket{\v{1}} + \sum_{\v{x} \in \overline{R(1,t)}\backslash\{\v{0}\}} \ket{\v{1} + \v{x}}\bigg),
\end{align}
where we use $\v{x}$ to denote $n$-bit strings {(${\v{0}}$ and ${\v{1}}$ are all-zero and all-one strings, respectively)}. All strings in $\overline{R(1,t)}\backslash\{\v{0}\}$ have weight $2^{t-1}$. 
Let $W$ be the encoding isometry $W = \ket{\frakc_0}\bra{0_L} + \ket{\frakc_1}\bra{1_L}$. {Consider $H_L = \frac{1}{2}Z_L$ and $H_S = -\frac{1}{2}\sum_{l=1}^n Z_l$, which again guarantees that the code tends to be covariant as $n\rightarrow\infty$.}
Let $\ket{\psi} = \ket{0_L}\ket{\psi_R^0} + \ket{1_L}\ket{\psi_R^1}$ be an arbitrary pure state on $L \otimes R$. Then 
\begin{align}
\ket{\psi_\theta}\! &= U_{S,\theta} W U_{L,\theta}^\dagger \ket{\psi} \nonumber\\
&= \bigg(\!\frac{1}{\sqrt{2^t}} \bigg(e^{i2^{t-1}{\theta}} \ket{\v{0}} + \sum_{\v{x} \in \overline{R(1,t)}\backslash\{\v{0}\}} \ket{\v{x}}\bigg)\!\bigg)\!\ket{\psi_R^0}  + \bigg(\!\frac{1}{\sqrt{2^t}} \bigg(e^{-i2^{t-1}{\theta}}\ket{\v{1}} + \sum_{\v{x} \in \overline{R(1,t)}\backslash\{\v{0}\}} \ket{\v{1} + \v{x}}\bigg)\!\bigg)\!\ket{\psi_R^1}, 
\end{align}
and thus
\begin{align}\label{eq:rm_p}
P(\mU_{S,\theta}\circ\mE_{\LtoS},\mE_{\LtoS}\circ\mU_{L,\theta}) = \max_{\psi} P(\ket{\psi_\theta},\ket{\psi})
= \sqrt{1 - \left|\frac{n + \cos((n+1)\theta/2)}{n + 1}\right|^2}. 
\end{align}
Therefore, we have
\begin{align}
\delta_\group  = \max_\theta P(\mU_{S,\theta}\circ\mE_{\LtoS}\circ \mU_{L,\theta}^\dagger, \mE_{\LtoS})  = \sqrt{1 - \left(\frac{n - 1}{n + 1}\right)^2} = \sqrt{\frac{4}{n}} + O\bigg(\bigg(\frac{1}{n}\bigg)^{3/2}\bigg). 
\end{align}
The lower bound from \thmref{thm:global-1} gives 
\begin{align}
\delta_\group \geq \sqrt{\frac{\Delta H_L(\Delta H_S - \frac{1}{2}\Delta H_L)}{\Delta H_S}} = \sqrt{\frac{n - \frac{1}{2}}{n^2}} = \sqrt{\frac{1}{n}} + O\bigg(\bigg(\frac{1}{n}\bigg)^{3/2}\bigg). 
\end{align}
Similar to the modified thermodynamic code, $\delta_\group$ saturates its lower bound up to a constant factor of 2 in the leading order of $1/n$. Also note that $\delta_{\group,\diamond} = \delta_\group$  in this case according to the discussion in \secref{sec:def}, indicating the saturation of the lower bound when we consider the diamond distance. The code is exactly covariant when $\theta = \frac{4k\pi}{n+1}$ and $k$ is an integer, and the corresponding local covariance violation can also be easily computed from \eqref{eq:rm_p}:
\begin{equation} 
\delta_\point = \sqrt{2 \partial^2_\theta P(\mU_{S,\theta}\circ\mE_{\LtoS},\mE_{\LtoS}\circ\mU_{L,\theta})^2\big|_{\theta = \frac{4k\pi}{n+1},\forall k \in \bZ}} = \sqrt{n+1}, 
\end{equation}
which has a quadratic gap with its lower bound $\Delta H_L = 1$. 
To compute the charge conservation violation $\delta_\charge$, we note that $(\mE_{\LtoS})^\dagger(H_S) = 0$ and  $(\mE_{\LtoS})^\dagger(H_S^2) = n/4$, thus 
\begin{align}
\delta_\charge = \Delta\left( H_L - (\mE_{\LtoS})^\dagger(H_S)\right) = 1,
\end{align}
matching our lower bound $\Delta H_L$. 
Also note that $\frakB = \sqrt{2n}$, so it is negligible in the trade-off relations from \thmref{thm:global-1} and \thmref{thm:local-metrology}. 

\section{Discussion}
\label{sec:discussion}

In this work, we devised and explored various approaches that enable us to quantitatively understand the fundamental trade-off between the QEC capability and several different characterizations of the degree of continuous symmetries associated with general QEC codes, including the violation of covariance conditions in both global and local senses as well as the violation of charge conservation (see \appref{app:comparison} for a summary). In particular, we introduced two intuitive and powerful frameworks based on the notions of charge fluctuation and gate implementation error respectively, and
employed several different methods from approximate QEC, quantum metrology, and quantum resource theory, to derive various forms of the trade-off relations in terms of distance metrics that address both worst-case and average-case inputs. 
Our results and techniques are expected to have numerous interesting applications to quantum computation as well as physics (see the main text).  We specifically discussed  the consequent restrictions on the transversal logical gates for general QEC codes, which could be of interest for fault tolerance.  We also provided detailed analysis of two interesting examples of approximately covariant codes---a parametrized extension of the (covariant) thermodynamic code, which gives a code family that continuously interpolates between exactly covariant and error-correcting, and the quantum Reed--Muller codes.  We showed that both codes can saturate the lower bounds asymptotically up to constant factors, indicating that the bounds are quite tight.

We would like to point out a few issues arising from our technical analysis that are not yet satisfactorily understood and could be worth further investigation: 
\begin{itemize}
    \item For both of the code examples we studied, the global covariance violation and charge conservation violation attain the optimal asymptotic scaling as indicated by the bounds based on the charge fluctuation approach, but the local covariance violation does not (there is a $\Theta(\sqrt{n})$ vs.\ $O(1)$ gap).   Note the observation (discussed above and in \appref{app:refine}) that if we additionally consider a recovery step in the definition of local covariance violation then it exhibits a tight scaling. We would hope to close this gap by further understanding both sides of it.  This is potentially key to a complete understanding of the behavior and practical meanings of the local symmetry measures. 
    \item The gate implementation error approach provides bounds that behave worse than the corresponding bounds from the charge fluctuation approach under uniformly random local noise.  It would be interesting to further understand whether this gap stems from the looseness of \propref{prop:gate}. On the other hand, the discussion in \secref{sec:comparison} also indicates that for extremely strong noise (so that the regularized QFI is sublinear), the gate implementation error approach outperforms the charge fluctuation approach.
    The question remains whether there is a universal bound which exhibits optimal scalings under any noise models.      
\end{itemize}

There are also several important directions for future study: 
\begin{itemize}
    \item Gate implementation error. We introduced the gate implementation error as a notion that nicely unifies QEC inaccuracy and global covariance violation, and in turn serves as a tool for deriving the trade-off between them.  We believe that this quantity is interesting in its own right and expect it to find broader applications in the analysis of QEC, distillation etc.
    \item General continuous symmetry groups. Here we mainly carried out the discussion in terms of $U(1)$ which corresponds to a single conserved quantity, but obviously the symmetry groups are often more complicated in quantum computation and physical scenarios.  It would be useful to extend our study to other important continuous symmetry groups such as $SU(d)$, for which we expect that our analysis for $U(1)$ provides a basis and serves as a sub-theory but it is useful to invoke corresponding representation theory machinery  (like in Refs.~\cite{faist2019continuous,KongLiu21:random}). 
    \item Discrete symmetries. Given the incompatibility results for continuous symmetries, it is natural to ask whether discrete symmetries, which are also broadly important, place restrictions on QEC.  It is known that for discrete symmetry groups one can in principle construct a covariant code which also achieves exact QEC  \cite{hayden2017error}, indicating that the incompatibility is not as fundamental as continuous symmetries.  However, we do know interesting cases  where exact QEC is forbidden even in the presence of discrete symmetries under simple additional constraints  (e.g.,~AdS/CFT codes---see Refs.~\cite{harlow2018constraints,harlow2018symmetries,faist2019continuous}). It would be interesting to further explore both the possible limitations as well as good code constructions for QEC with discrete symmetries in more general terms. 
\end{itemize}

Furthermore, we expect the study of how QEC interacts symmetries to be relevant in wide-ranging physical scenarios.  In the main text, we discussed potential applications of our theory and techniques to several topics of great interest in physics, including AdS/CFT, black hole radiation, and many-body physics.  It would be interesting to further consolidate these ideas.  To this end, an important task is to bridge the language of quantum information used here and those commonly used in high energy and condensed matter physics.  To conclude, our study enriches the ``physical'' understanding of QEC using a wide variety of approaches in quantum information. We hope it will stimulate further interest into exploring the interaction between QEC, quantum information, and physics.

\appendix

\newpage 

\setcounter{theorem}{0}
\setcounter{lemma}{0}
\setcounter{figure}{0}
\renewcommand{\thefigure}{S\arabic{figure}}
\renewcommand{\thelemma}{S\arabic{lemma}}
\renewcommand{\thetheorem}{S\arabic{theorem}}
\renewcommand{\theproposition}{S\arabic{proposition}}
\renewcommand{\thecorollary}{S\arabic{corollary}}
\renewcommand{\theHfigure}{Supplement.\arabic{figure}}
\renewcommand{\theHlemma}{Supplement.\arabic{lemma}}
\renewcommand{\theHtheorem}{Supplement.\arabic{theorem}}
\renewcommand{\theHcorollary}{Supplement.\arabic{corollary}}


\section{Summary and comparison of the different trade-off results}
\label{app:comparison}

\begin{table}[h!]
\begin{center}
 \begin{tabular}{|| c | c ||}  
 \hline
  Trade-off Relations  &  Summary of Strengths and Weaknesses \\
 \hline\hline
  \makecell{\thmref{thm:global-1} \\($\delta_\group$ and $\epsilon$)} & 
  \makecell{For single-erasure errors, the bound is near-optimal (as shown in \secref{sec:case-study}). \\The bound only applies to isometric encodings. } \\
 \hline
  \makecell{\thmref{thm:global-2-metrology} \\($\delta_\group$ and $\epsilon$)} & \makecell{The bound is suboptimal for single-erasure errors,\\ but it outperforms \thmref{thm:global-1} for strong noise when $\frakF = o(n)$. \\The bound applies to general encoding channels.} \\ 
 \hline 
  \makecell{\thmref{thm:global-2-res}\\ ($\delta_\group$ and $\epsilon$; $\overline{\delta}_\group$ and $\overline{\epsilon}$)} & \makecell{The bounds apply to general encodings channels, \\ and can characterize the average behavior based on the Choi purified distance.\\
  The worst-case bound is weaker than \thmref{thm:global-2-metrology}.} \\
 \hline
  \makecell{\thmref{thm:local-KL} \\($\delta_\charge$  and $\epsilon$; $\delta_\point$  and $\epsilon$)} & 
  \makecell{The trade-off between $\delta_\charge$  and $\epsilon$ is near-optimal for single-erasure errors (as shown in \secref{sec:case-study}). \\ The bounds only apply to isometric encodings. } \\
 \hline
  \makecell{\thmref{thm:local-metrology}\\ ($\delta_\charge$ and $\epsilon$; $\delta_\point$  and $\epsilon$)} & \makecell{The trade-off between $\delta_\charge$  and $\epsilon$ is near-optimal for single-erasure errors (as shown in \secref{sec:case-study}). \\ It contains the parameter $\frakB$, the subtlety related to which is discussed in \secref{sec:global-1-metrology}.\\The bound on $\delta_\charge$  and $\epsilon$ applies to general encoding channels.\\The bound on $\delta_\point$  and $\epsilon$ only applies to isometric encodings. }  \\
 \hline
  \makecell{\thmref{thm:local-point}\\ ($\delta_\point$  and $\epsilon$)} & \makecell{The bound does not contain parameter $\frakB$ and also applies to general encoding channels. \\
  The bound is suboptimal in some cases, see further discussions in \appref{app:refine}.} \\
 \hline 
\end{tabular}
\end{center}
\caption{Summary of the strengths and weaknesses of the various trade-off relations we derived using different approaches.}
\label{table:compare}
\end{table}

\section{Purified distance and diamond distance between rotated dephasing channels and the identity}
\label{app:dephasing}

Here, we calculate the purified distance between rotated dephasing channels and the identity for completeness. The same calculations could be found, e.g., in Ref.~\cite{zhou2020new}. 
\begin{lemma}
\label{lemma:dephasing}
Consider rotated dephasing channels $\mD(\cdot) = (1-p) e^{-i\frac{\phi}{2}Z} \rho e^{i\frac{\phi}{2}Z} + p e^{-i\frac{\phi}{2}Z} Z \rho Z e^{i\frac{\phi}{2}Z}$, where $Z$ is the Pauli-Z operator. It holds that $P(\mD,\id) = \sqrt{\frac{1}{2}(1 - (1-2p)\cos\phi)}$ and $D_\diamond(\mD,\id) = \frac{1}{2}\sqrt{1-2(1-2p)\cos\phi + (1-2p)^2}$. 
\end{lemma}
\begin{proof}
Let the input state be $\ket{\psi} = \psi_{00}\ket{00} + \psi_{01}\ket{01} + \psi_{10}\ket{10} + \psi_{11}\ket{11}$, then 
\begin{equation}
(\mD\otimes \id)(\ket{\psi}\bra{\psi})
= 
\begin{pmatrix}
\psi_{00}\psi_{00}^* & \psi_{00}\psi_{01}^* & (1-2p)e^{-i\phi}\psi_{00}\psi_{10}^* & (1-2p)e^{-i\phi} \psi_{00}\psi_{11}^* \\
\psi_{00}\psi_{01}^* & \psi_{01}\psi_{01}^* & (1-2p)e^{-i\phi}\psi_{01}\psi_{10}^* & (1-2p)e^{-i\phi} \psi_{01}\psi_{11}^* \\
(1-2p)e^{i\phi}\psi_{10}\psi_{00}^* & (1-2p)e^{i\phi}\psi_{10}\psi_{01}^* & \psi_{10}\psi_{10}^* & \psi_{10}\psi_{11}^* \\
(1-2p)e^{i\phi}\psi_{11}\psi_{00}^* & (1-2p)e^{i\phi}\psi_{11}\psi_{01}^* & \psi_{11}\psi_{10}^* & \psi_{11}\psi_{11}^* \\
\end{pmatrix}. 
\end{equation}
Then 
\begin{align}
P(\mD,\id) &= \sqrt{1 - f^2(\mD,\id)} \\
&= \max_{\psi_{00,01,10,11}} \left( 2 \Re[(1-2p)e^{-i\phi}](|\psi_{00}|^2+|\psi_{01}|^2)(|\psi_{10}|^2+|\psi_{11}|^2) \right)^{1/2} = \sqrt{\frac{1}{2}\left(1 - (1-2p)\cos\phi\right)}, 
\end{align}
and 
\begin{align}
D_\diamond(\mD,\id) &= \frac{1}{2}\norm{\mD-\id}_\diamond \\
&= 2 \abs{1 - (1-2p)e^{-i\phi}} \max_{\psi_{00,01,10,11}} \norm{\begin{pmatrix}\psi_{00}\\ \psi_{01}\end{pmatrix} \begin{pmatrix}\psi_{10}^* & \psi_{11}^* \end{pmatrix} }_1
= \frac{1}{2}\sqrt{1-2(1-2p)\cos\phi + (1-2p)^2}. 
\end{align}
\end{proof}

\section{Two lower bounds on \texorpdfstring{$\abs{\partial_\theta\xi_\theta}^2\big|_{\theta = 0}$}{the absolute value of the derivative of xi at theta = 0}. }
\label{app:local-proof}

Here, we present the proofs of the two lower bounds on $\abs{\partial_\theta\xi_\theta}^2\big|_{\theta = 0}$ used in \propref{prop:charge-metrology} and \thmref{thm:local-point}.

\begin{lemma}
\label{lemma:local}
\nonisometric 
Let $\mR^{\optL}_{\StoL}$ be the recovery channel such that 
$\epsilon =  P(\mR^{\optL}_{\StoL}\circ \mN_{S}\circ\mE_{\LtoS}, \id_{L})$ and $\mN_{C,\theta} = \mR^{\rep}_{\SAtoC} \circ ( \mR_{\StoL}^{\optL} \circ\mN_{S} \otimes \id_A) \circ  (\mU_{S,\theta}\circ \mE_{\LtoS} \otimes \id_A) \circ \mE_{\CtoLA}^{\rep}$. Then 
$\xi_\theta = \bra{0_C}\mN_{C,\theta}(\ket{0_C}\bra{1_C})\ket{1_C}$ satisfies 
\begin{equation}
\label{eq:lower-QFI-1}
\abs{\partial_\theta\xi_\theta}^2\big|_{\theta = 0}  \geq  ((1-2\epsilon^2)\Delta H_L - \delta_\point)^2,
\end{equation}
when $(1-2\epsilon^2)\Delta H_L \geq \delta_\point$. 
\begin{equation}
\label{eq:lower-QFI-3}
\abs{\partial_\theta\xi_\theta}^2\big|_{\theta = 0}  \geq (\abs{\chi}-2\epsilon\frakB)^2,
\end{equation}
when $\abs{\chi} \geq 2\epsilon\frakB$. Here $\frakB = \max_{\ket{\psi}}\sqrt{8 \bV_{H_S}\left(\mE_{\LtoS}(\ket{\psi})\right)}$. 
\end{lemma}

\begin{proof}
We have two rotated dephasing channels: 
\begin{align}
\mD_{C,\theta} &= \mR^{\rep}_{\SAtoC} \circ ( \mR_{\StoL}^{\optL} \circ\mN_{S} \otimes \id_A) \circ  (\mU_{S,\theta}\circ \mE_{\LtoS} \circ \mU_{L,\theta}^\dagger \otimes \id_A) \circ \mE_{\CtoLA}^{\rep},\\
\mN_{C,\theta} &= \mR^{\rep}_{\SAtoC} \circ ( \mR_{\StoL}^{\optL} \circ\mN_{S} \otimes \id_A) \circ  (\mU_{S,\theta}\circ \mE_{\LtoS} \otimes \id_A) \circ \mE_{\CtoLA}^{\rep}, 
\end{align}
of the following forms:
\begin{gather}
\mD_{C,\theta}(\cdot) = (1 - p_\theta) e^{-i\frac{\phi_\theta}{2} Z_C} (\cdot) e^{i\frac{\phi_\theta}{2}Z_C} 
 + p_\theta Z_C   e^{-i\frac{\phi_\theta}{2}Z_C} (\cdot) e^{i\frac{\phi_\theta}{2}Z_C} Z_C, \\
\mN_{C,\theta}(\cdot) = (1 - p_\theta) e^{-i\frac{\phi_\theta+\Delta H_L\theta}{2} Z_C} (\cdot) e^{i\frac{\phi_\theta+\Delta H_L\theta}{2} Z_C} 
 + p_\theta Z_C e^{-i\frac{\phi_\theta+\Delta H_L\theta}{2} Z_C} (\cdot) e^{i\frac{\phi_\theta+\Delta H_L\theta}{2} Z_C} Z_C. 
\end{gather}
Consider the parameter estimation of $\theta$ in the neighborhood of $\theta = 0$. Then let $\mR_{\SAtoC} = \mR_{\SAtoC}^\rep \circ (\mR_{\StoL}^{\optL} \circ \id_A)$ and $\mE_{\CtoSA} = (\mE_{\LtoS} \circ \id_A) \circ \mE_{\CtoLA}^\rep  $, we have 
\begin{align}
\sqrt{p_{\theta=0}} &\leq P(\mD_{C,\theta=0},\id_C) 
= P(\mR_{\SAtoC} \circ (\mN_{S} \otimes \id_A) \circ \mE_{\CtoSA} ,\id_C)
\nonumber\\
&\leq P(\mR^{\optL}_{\StoL} \circ \mN_{S} \circ \mE_{\LtoS}, \id_L) \leq \epsilon, 
\end{align}
where we use \lemmaref{lemma:dephasing}, the monotonicity of the purified distance  and the definition of $\epsilon$. 

\begin{enumerate}[(1),wide, labelwidth=0pt, labelindent=0pt]
\item We first prove \eqref{eq:lower-QFI-1}. The channel QFI~\cite{fujiwara2008fibre,demkowicz2012elusive,zhou2020theory} of rotated dephasing channel $\mD_{C,\theta}$ is 
\begin{align}
F(\mD_{\theta,C}) 
&= 4\min_{h} \norm{(\partial_\theta\vK_C - i h \vK_C)^\dagger(\partial_\theta\vK - i h \vK_C)} \nonumber\\
&= (1-2p_\theta)^2(\partial_\theta \phi_\theta)^2 + \frac{(\partial_\theta p_\theta)^2}{(1-p_\theta)p_\theta},
\end{align}
where $h$ is an arbitrary Hermitian matrix and $\vK_C = \begin{pmatrix}\sqrt{1-p_\theta}e^{-i\phi_\theta Z_C} \\ \sqrt{p_\theta}e^{-i\phi_\theta Z_C} Z_C\end{pmatrix}$. 
Moreover, using the monotonicity of the channel QFI, 
\begin{equation}
\label{eq:refine-1}
F(\mD_{\theta,C})\big|_{\theta = 0} \leq F(\mU_{S,\theta}\circ \mE_{\LtoS} \circ \mU_{L,\theta}^\dagger)\big|_{\theta = 0} = (\delta_\point)^2. 
\end{equation}
Then  
\begin{align}
\abs{\partial_\theta\xi_\theta}^2\big|_{\theta = 0}
&= \abs{\bra{0_C}\partial_\theta \mD_{C,\theta}\circ\mU_{C,\theta}(\ket{0_C}\bra{1_C})\ket{1_C} + \bra{0_C}\mD_{C,\theta}\circ\partial_\theta\mU_{C,\theta}(\ket{0_C}\bra{1_C})\ket{1_C}}^2 \big|_{\theta = 0}\\
&= \abs{\bra{0_C}\partial_\theta \mD_{C,\theta}\circ\mU_{C,\theta}(\ket{0_C}\bra{1_C})\ket{1_C} + (-i\Delta H_L)\bra{0_C}\mD_{C,\theta}\circ\mU_{C,\theta}(\ket{0_C}\bra{1_C})\ket{1_C}}^2 \big|_{\theta = 0}\\
&\geq \abs{(1-2p_\theta)i(-\partial_\theta \phi_\theta)e^{-i\phi_\theta} + (-i\Delta H_L)(1-2p_\theta)e^{-i\phi_\theta}}^2 \big|_{\theta = 0}\geq ((1 - 2 \epsilon^2) \Delta H_L - \delta_\point)^2. \label{eq:refine-2}
\end{align}

\item Next we prove \eqref{eq:lower-QFI-3}. First, note that
\begin{equation}
\partial_\theta (\mU_{S,\theta}\circ \mE_{\LtoS})(\cdot)  = -i[H_S,\mU_{S,\theta}\circ \mE_{\LtoS} (\cdot)]. 
\end{equation}
Then 
\begin{equation}
\partial_\theta\xi_\theta\big|_{\theta = 0} = \bra{0_C}\partial_\theta \mN_{C,\theta}(\ket{0_C}\bra{1_C})\ket{1_C}\big|_{\theta = 0} =
\bra{0_C} \mR_{\SAtoC}\circ \mN_{SA} \left( -i[H_{SA},\mE_{\CtoSA}(\ket{0_C}\bra{1_C})]\right)\ket{1_C}, 
\end{equation}
where $\mN_{SA} := \mN_{S}\otimes\id_A$ and $H_{SA} := H_S \otimes \id_A$.

We claim that 
\begin{equation}
\label{eq:approx-channel}
\abs{ \bra{0_C}\partial_\theta \mN_{C,\theta}(\ket{0_C}\bra{1_C})\ket{1_C}\big|_{\theta = 0}
 - \bra{0_L}\partial_\theta \widetilde{\mU}_{L,\theta}(\ket{0_L}\bra{1_L})\ket{1_L}\big|_{\theta = 0} } \leq 2\epsilon\frakB,
\end{equation}
where $\widetilde\mU_{L,\theta} = e^{-i (\mE_{\LtoS})^\dagger(H_S)\theta} (\cdot) e^{i (\mE_{\LtoS})^\dagger(H_S)\theta}$. Clearly, 
\begin{equation}
\label{eq:approx-channel-1}
\bra{0_L}\partial_\theta \widetilde{\mU}_{L,\theta}(\ket{0_L}\bra{1_L})\ket{1_L}\big|_{\theta = 0} = -i \bra{0_L}\left([(\mE_{\LtoS})^\dagger(H_S),\ket{0_L}\bra{1_L}]\right)\ket{1_L} = -i\chi.
\end{equation}
\eqref{eq:lower-QFI-3} is then proven combining \eqref{eq:approx-channel} and \eqref{eq:approx-channel-1}. 

Now we prove \eqref{eq:approx-channel}. 
First, let $\mE_{\CtoSA}(\cdot) = \sum_j E_j (\cdot) E_j^\dagger$, $\mR_{\SAtoC}\circ(\mN_{S}\otimes\id_A)(\cdot) = \sum_i R_i(\cdot)R_i^\dagger$. We have 
\begin{equation}
f^2\left(\sum_{ij} R_iE_j \rho E_j^\dagger R_i^\dagger,\rho\right) = f^2\left(\mR_{\SAtoC}\circ(\mN_{S}\otimes\id_A)\circ\mE_{\CtoSA}(\rho),\rho\right) \geq f^2\left(\mR^{\optL}_{\StoL}\circ \mN_{S} \circ\mE_{\LtoS},\id_L\right) = 1 - \epsilon^2,
\end{equation}
for any $\rho$ in system $C\otimes R$ where $R$ is a two-dimensional reference system. In particular, choose $\rho$ to be the maximally entangled state in $C\otimes R$, we have 
\begin{equation}
2\left(1 - f^2\left(\sum_{ij} R_iE_j \rho E_j^\dagger R_i^\dagger,\rho\right)\right) = \sum_{ij} \left(\trace( E_j^\dagger R_i^\dagger R_iE_j) - \frac{\abs{\trace(R_iE_j)}^2}{2}\right)  = \sum_{ij} \norm{\check{A}_{ij}}_{HS}^2\leq 2 \epsilon^2, 
\end{equation}
where $A_{ij} = R_i E_j$ and $\check{(\cdot)} = (\cdot) - \frac{\trace(\cdot) \id }{2}$,
and $\norm{\cdot}_{HS}$ denotes the Hilbert-Schmidt norm. Note that 
\begin{align}
&\bra{0_C}\partial_\theta \mN_{C,\theta}(\ket{0_C}\bra{1_C})\ket{1_C}\big|_{\theta = 0}
 - \bra{0_L}\partial_\theta \widetilde{\mU}_{L,\theta}(\ket{0_L}\bra{1_L})\ket{1_L}\big|_{\theta = 0} 
 \nonumber\\
=\;&\bra{0_C} \mR_{\SAtoC}\circ \mN_{SA} \left( -i[H_{SA},\mE_{\CtoSA}(\ket{0_C}\bra{1_C})]\right)\ket{1_C} + i \bra{0_C} [\mE_{\SAtoC}^\dagger(H_{SA}),\ket{0_C}\bra{1_C}] \ket{1_C} \nonumber\\
=\;& -i \bra{0_C}\left(\sum_{ij} R_i H_{SA} E_j \ket{0_C}\bra{1_C} E_j^\dagger R_i^\dagger  - \sum_j E_j^\dagger H_{SA} E_j (\ket{0_C}\bra{1_C})\right)\ket{1_C} \nonumber \\
&+ i \bra{0_C}\left(  \sum_{ij} R_i  E_j \ket{0_C}\bra{1_C} E_j^\dagger H_{SA} R_i^\dagger  - \sum_j  (\ket{0_C}\bra{1_C}) E_j^\dagger H_{SA} E_j\right)\ket{1_C}.
\end{align}
Then 
\begin{align}
&\Big\| \sum_{ij} R_i H_{SA} E_j \ket{0_C}\bra{1_C} E_j^\dagger R_i^\dagger  - \sum_j E_j^\dagger H_{SA} E_j (\ket{0_C}\bra{1_C}) \Big\|\nonumber\\
\leq\;& \Big\| \sum_{ij} R_i H_{SA} E_j \ket{0_C}\bra{1_C} E_j^\dagger R_i^\dagger  - \sum_j E_j^\dagger H_{SA} E_j (\ket{0_C}\bra{1_C}) \Big\|_{HS}\\
=\;& \Big\| \sum_{ij} R_i H_{SA} E_j \ket{0_C}\bra{1_C} A_{ij}^\dagger  - A_{ij}^\dagger R_i H_{SA} E_j \ket{0_C}\bra{1_C} \Big\|_{HS}\\
=\;& \Big\| \sum_{ij} R_i H_{SA} E_j \ket{0_C}\bra{1_C} \check A_{ij}^\dagger  - \check A_{ij}^\dagger R_i H_{SA} E_j \ket{0_C}\bra{1_C} \Big\|_{HS}\\
\leq\;& 2 \sum_{ij} \norm{\check A_{ij}}_{HS} \norm{R_i H_{SA} E_j \ket{0_C}\bra{1_C}}_{HS} \\ \leq\;& 2 \sqrt{\sum_{ij}\norm{\check A_{ij}}_{HS}^2 \sum_{ij}\norm{R_i H_{SA} E_j \ket{0_C}\bra{1_C}}_{HS}^2}\\
\leq\;& 2 \sqrt{2 \epsilon^2  \bra{0_C}\mE_{\SAtoC}^\dagger(H_{SA}^2)\ket{0_C}} \\ \leq\;& 2\sqrt{2} \epsilon \sqrt{\norm{(\mE_{\LtoS})^\dagger(H_S^2)}},
\end{align}
where we used the triangular inequality, the submultiplicity of $\norm{\cdot}_{HS}$ and the Cauchy--Schwartz inequalities. 
Similarly, 
\begin{equation}
\Big\|\sum_{ij} R_i  E_j \ket{0_C}\bra{1_C} E_j^\dagger H_{SA} R_i^\dagger  - \sum_j  (\ket{0_C}\bra{1_C}) E_j^\dagger H_{SA} E_j \Big\|\leq 2\sqrt{2}\epsilon \sqrt{\norm{(\mE_{\LtoS})^\dagger(H_S^2)}}.
\end{equation}
We have
\begin{equation}
\abs{ \bra{0_C}\partial_\theta \mN_{C,\theta}(\ket{0_C}\bra{1_C})\ket{1_C}\big|_{\theta = 0}
 - \bra{0_L}\partial_\theta \widetilde{\mU}_{L,\theta}(\ket{0_L}\bra{1_L})\ket{1_L}\big|_{\theta = 0} } \leq 4\sqrt{2}\epsilon \sqrt{\norm{(\mE_{\LtoS})^\dagger(H_S^2)}}. 
\end{equation}
Note that the equation above still holds when replacing $H_S$ with $H_S-\nu\id$, and that 
\begin{align}
\min_{\nu \in \bR} 2\sqrt{2}\sqrt{\norm{(\mE_{\LtoS})^\dagger((H_S - \nu\id)^2)}} &= 2\sqrt{2}\sqrt{\min_\nu \max_{\ket{\psi}}    \bra{\psi}(\mE_{\LtoS})^\dagger(H_S^2)\ket{\psi}   - 2\nu  \bra{\psi}(\mE_{\LtoS})^\dagger(H_S)\ket{\psi}  +  \nu^2 } \\
& 
= 2\sqrt{2}\sqrt{\max_{\ket{\psi}}  \min_\nu   \bra{\psi}(\mE_{\LtoS})^\dagger(H_S^2)\ket{\psi}   - 2\nu  \bra{\psi}(\mE_{\LtoS})^\dagger(H_S)\ket{\psi}  +  \nu^2 }
\\
& 
= \max_{\ket{\psi}}   2\sqrt{2}\sqrt{ \bra{\psi}(\mE_{\LtoS})^\dagger(H_S^2)\ket{\psi}   -  (\bra{\psi}(\mE_{\LtoS})^\dagger(H_S)\ket{\psi})^2  } = \frakB,
\end{align}
where in the second step we used the Sion's minimax theorem~\cite{komiya1988elementary}. 
\eqref{eq:approx-channel} is then proven.

\end{enumerate}

\end{proof}

\section{Refinements of 
{Proposition~2}, 
{Proposition~3} and 
{Proposition~5}}
\label{app:refine-2}

 Here, we present the refinements of \propref{prop:global-charge}, \propref{prop:charge-KL} and \propref{prop:charge-metrology}. First, we can slightly modify the proof of \propref{prop:global-charge} to obtain the following: 
\begin{proposition}
\label{prop:global-charge-app}
\isometric 
Then when $\delta_\charge \leq \Delta H_S$, the following trade-off relation holds: 
\begin{equation}
\label{eq:app-global-charge}
\delta_\group \geq 
\min\left\{\frac{\sqrt{\delta_\charge\left(\Delta H_S - \frac{1}{2}\delta_\charge\right)}}{\Delta H_S},\sqrt{\frac{3}{8}}\right\}, 
\end{equation}
and when $\delta_\charge > \Delta H_S$, $\delta_\group \geq \sqrt{3/8}$. 
In particular, when $\delta_\charge \ll \Delta H_S$, we have 
\begin{equation}
\delta_\group \gtrsim \sqrt{\frac{\delta_\charge}{\Delta H_S}}.
\end{equation}
\end{proposition}

\begin{proof}
Since $U_{S,\theta}$ and $U_{L,\theta}$ are both periodic with a common period, we assume $H_S$ and $H_L$ both have  integer eigenvalues. We also assume the smallest eigenvalue of $H_S$ is zero because constant shifts do not affect the definitions of symmetry measures. Choose orthonormal $\ket{\tilde{0}_L}$ and $\ket{\tilde{1}_L}$ such that
\begin{equation}
    \tilde\delta_\charge = \abs{\bra{\tilde{0}_L}(H_L - (\mE_{\LtoS})^\dagger(H_S))\ket{\tilde{0}_L}
    - \bra{\tilde{1}_L}(H_L - (\mE_{\LtoS})^\dagger(H_S))\ket{\tilde{1}_L}}
\end{equation} 
is arbitrarily close to $\delta_\charge$ and that $\widetilde{\Delta H_L} =  \bra{\tilde{0}_L}H_L\ket{\tilde{0}_L}   - \bra{\tilde{1}_L} H_L\ket{\tilde{1}_L} \geq 0$ is a rational number. We can always multiply both $H_{S}$ and $H_{L}$ by a large integer (e.g., the denominator of $\widetilde{\Delta H_L}$) such that $\widetilde{\Delta H_L}$ becomes an integer. This rescaling does not change the value of $\delta_\group$ and the right-hand side of \eqref{eq:app-global-charge}. Therefore, without loss of generality, we assume $\widetilde{\Delta H_L}$ is an integer. 

When $\mE_{\LtoS}(\cdot) = W(\cdot)W^\dagger$ is isometric, let $\ket{\tilde\frakc_{0}} = W\ket{\tilde{0}_L}$, $\ket{\tilde\frakc_{1}} = W\ket{\tilde{1}_L}$, and write
\begin{align}
\ket{\tilde\frakc_0} = \sum_{\eta=0}^{\Delta H_S} c_\eta^0 \ket{\eta^0},\quad 
\ket{\tilde\frakc_1} = \sum_{\eta=0}^{\Delta H_S} c_\eta^1 \ket{\eta^1},
\end{align} 
where $\sum_\eta \abs{c_\eta^0}^2 = \sum_\eta \abs{c_\eta^1}^2 = 1$ and $\ket{\eta^0}$ and $\ket{\eta^1}$ are eigenstates of $H_S$ with eigenvalue $\eta$. $\ket{\eta^0}$ and $\ket{\eta^1}$ may not be the same when $H_S$ is degenerate.  Note that when $\eta$ is not an eigenvalue of $H_S$, we simply take $c_\eta^{0} = 0$ (or $c_\eta^{1} = 0$) so that $c_\eta^{0}$ (or $c_\eta^{1}$) and is well-defined for any integer $\eta$. 
Let $\ket{\psi} = \frac{1}{\sqrt{2}}(\ket{\tilde 0_L}\ket{0_R}+\ket{\tilde 1_L}\ket{1_R})$. 
Then the channel fidelity 
\begin{align}
f_\theta &:=f(\mU_{S,\theta}\circ \mE_{\LtoS}\circ\mU_{L,\theta}^\dagger,\mE_{\LtoS}) \nonumber\\
&= \min_{\rho} f\big((\mU_{S,\theta}\!\circ\! \mE_{\LtoS}\!\circ\!\mU_{L,\theta}^\dagger\!\otimes\!\id_R)(\rho),(\mE_{\LtoS}\!\otimes\!\id_R)(\rho)\big)\nonumber\\
&\leq f\big((\mU_{S,\theta}\!\circ\! \mE_{\LtoS}\!\circ\!\mU_{L,\theta}^\dagger\!\otimes\!\id_R)(\ket{\psi}),(\mE_{\LtoS}\!\otimes\!\id_R)(\ket{\psi})\big)\nonumber\\
&= \abs{\braket{\psi|W^\dagger  U_{S,\theta} W U_{L,\theta}^{\dagger}|\psi}} \nonumber\\
&= \abs{\frac{1}{2}\sum_{\eta=0}^{\Delta H_S} \abs{c_\eta^0}^2 e^{-i\eta\theta+i\widetilde{\Delta H_L}\theta} + \frac{1}{2}\sum_{\eta=0}^{\Delta H_S} \abs{c_\eta^1}^2 e^{-i\eta\theta}}\nonumber\\
&= \abs{\sum_{\eta=-\widetilde{\Delta H_L}}^{\Delta H_S} c_\eta e^{-i\eta\theta}} = \abs{c_{\eta_*} \!+\! \sum_{\eta\neq\eta_*} c_\eta e^{-i(\eta-\eta_*)\theta}},
\end{align}
where we define $c_\eta := \frac{1}{2} \abs{c_{\eta+\widetilde{\Delta H_L}}^0}^2+ \frac{1}{2}\abs{c_\eta^1}^2$ for $\eta \in [-\widetilde{\Delta H_L},\Delta H_S]$ and choose $\eta_*$ such that $c_{\eta_*} \geq c_{\eta}$ for all $\eta$. Note that there is always a $\theta$ such that $\sum_{\eta\neq\eta_*} c_\eta \cos((\eta-\eta_*)\theta) = 0$ (because the integration of it from $0$ to $2\pi$ is zero) and that $\sum_{\eta\neq\eta_*} c_\eta e^{-i(\eta-\eta_*)\theta}$ is imaginary. Then we must have 
\begin{equation}
\min_\theta f_\theta \leq \sqrt{c_{\eta_*}^2 + (1 - c_{\eta_*})^2}. 
\end{equation}

To arrive at a non-trivial lower bound on $\delta_\group = \sqrt{1-\min_\theta f_\theta^2}$, we need an upper bound of $\min_\theta f_\theta$ which is smaller than $1$. To this end, we analyze $c_{\eta_*}$  in detail. In particular, we consider two situations: 
\begin{enumerate}[(1),wide, labelindent=0pt]
\item $c_{\eta_*} \leq 1/2$ and a constant upper bound on $\min_\theta f_\theta$ exists. We can always find a subset of $\{\eta\}$ denoted by $\is$ such that $ 1/4 \leq \sum_{\eta\in\is} c_\eta \leq 1/2$. To find such a set, we first include $\eta_*$ in $\is$ and add new elements into $\is$ one by one until their sum is at least $1/4$. Then there is always a $\theta$ such that $(\sum_{\eta\in\is} c_\eta e^{-i\eta\theta})\cdot(\sum_{\eta\notin\is} c_\eta e^{-i\eta\theta})$ is imaginary, in which case $\min_\theta f_\theta \leq \sqrt{(1/4)^2+(3/4)^2} = \sqrt{5/8}$ and we have \begin{equation}\label{eq:delta-1-app} \delta_\group \geq \sqrt{1-\min_\theta f_\theta^2} \geq \sqrt{3/8}.\end{equation}
\item $c_{\eta_*} > 1/2$. Then $\sqrt{c_{\eta_*}^2 + (1-c_{\eta_*})^2}$ is a monotonically increasing function of $c_{\eta_*}$ and we only need to find an upper bound on $c_{\eta_*}$. 
We first note that $0 \leq \eta^* \leq \Delta H_S - \widetilde{\Delta H_L}$ because otherwise either $c_{\eta_*+\widetilde{\Delta H_L}}^0 = 0$ or $c_{\eta_*}^1 = 0$ which contradicts with $c_{\eta_*} > 1/2$. We have $\tilde\chi = \braket{\tilde\frakc_0|H_S|\tilde\frakc_0} - \braket{\tilde\frakc_1|H_S|\tilde\frakc_1}$, $\abs{\widetilde{\Delta H_L} - \tilde\chi} = \tilde\delta_\charge$ and 
\begin{gather*}
\sum_\eta \abs{c_\eta^0}^2 \eta = \sum_\eta \abs{c_\eta^1}^2 \eta + \tilde\chi, \\
~\Leftrightarrow~ 
\widetilde{\Delta H_L} - \tilde\chi = -\left(1-\abs{c_{\eta_*}^1}^2\right)\left(\eta_* \!-\! \frac{\sum_{\eta\neq\eta_*}{\abs{c_\eta^1}^2\eta}}{\sum_{\eta\neq\eta_*}{\abs{c_\eta^1}^2}}\right) 
+\left(1-\abs{c_{\eta_*+\widetilde{\Delta H_L}}^0}^2\right)\left(\eta_* \!+\! \widetilde{\Delta H_L} \!-\! \frac{\sum_{\eta\neq\eta_*+\widetilde{\Delta H_L}}{\abs{c_\eta^0}^2\eta}}{\sum_{\eta\neq\eta_*+\widetilde{\Delta H_L}}{\abs{c_\eta^0}^2}}\right).
\end{gather*}
Note that both $\bigg|\eta_* - \frac{\sum_{\eta\neq\eta_*}{\abs{c_\eta^1}^2\eta}}{\sum_{\eta\neq\eta_*}{\abs{c_\eta^1}^2}}\bigg|$ and $\bigg|\eta_* + \widetilde{\Delta H_L} - \frac{\sum_{\eta\neq\eta_*+\widetilde{\Delta H_L}}{\abs{c_\eta^0}^2\eta}}{\sum_{\eta\neq\eta_*+\widetilde{\Delta H_L}}{\abs{c_\eta^0}^2}}\bigg|$ are at most $\Delta H_S$. Therefore, $c_{\eta_*} \leq 1 - \tilde\delta_\charge/(2\Delta H_S)$ and 
\begin{align}
\delta_\group 
\geq \sqrt{1 - \min_\theta f_\theta^2} \geq \sqrt{2 c_{\eta_*}(1-c_{\eta_*})} 
\geq \frac{\sqrt{\tilde\delta_\charge\left(\Delta H_S - \frac{1}{2}\tilde\delta_\charge\right)}}{\Delta H_S}. \label{eq:KL-last-step-app}
\end{align}
\end{enumerate}
\propref{prop:global-charge-app} then follows from \eqref{eq:delta-1-app}, \eqref{eq:KL-last-step-app} and the fact that $\tilde \delta_\charge$ can be arbitrarily close to $\delta_\charge$. 
\end{proof}

\propref{prop:global-charge-app} shows that $\delta_\group$ is always lower bounded by a monotonic function of $\delta_\charge$. We also immediately see that such a relation cannot hold true for general non-isometric encodings because there are situations when $\delta_\charge > 0$ and $\delta_\group = 0$~\cite{cirstoiu2020robustness}. 
The proof of \propref{prop:global-charge-app} is essentially the same as the proof of \propref{prop:global-charge} except that we replace the pair of logical states $\ket{0_L}$ and $\ket{1_L}$ with another pair of logical states $\ket{\tilde 0_L}$ and $\ket{\tilde 1_L}$. Similar tricks can be used to obtain the refinements of \propref{prop:charge-KL} and \propref{prop:charge-metrology}. 
\begin{proposition}
\label{prop:charge-KL-app} 
\isometric Then the following inequality holds: 
\begin{equation}
\Delta((\mE_{\LtoS})^\dagger(H_S)) \leq 2\epsilon\frakJ.
\end{equation}
\end{proposition}
\begin{proposition}
\label{prop:charge-metrology-app}
\nonisometric 
Then the following inequality holds: 
\begin{equation}
\Delta((\mE_{\LtoS})^\dagger(H_S)) \leq 2\epsilon\left(  \sqrt{(1-\epsilon^2)\frakF} + \frakB\right). 
\end{equation} 
In particular, when $\epsilon \ll 1$ and $\frakB \ll \sqrt{\frakF}$, we have 
\begin{equation}
\Delta((\mE_{\LtoS})^\dagger(H_S)) \lesssim 2\epsilon\sqrt{\frakF}.
\end{equation}
\end{proposition}
We omit the proofs of \propref{prop:charge-KL-app} and \propref{prop:charge-metrology-app} here, as they follow directly from the proofs of \propref{prop:charge-KL} and \propref{prop:charge-metrology} where we replace $\ket{0_L}$ and $\ket{1_L}$ with orthonormal $\ket{\tilde{\tilde{0}}_L}$ and $\ket{\tilde{\tilde{1}}_L}$ which satisfy
\begin{equation}
    \Delta((\mE_{\LtoS})^\dagger(H_S))  = \braket{\tilde{\tilde{0}}_L|(\mE_{\LtoS})^\dagger(H_S)|\tilde{\tilde{0}}_L} - \braket{\tilde{\tilde{1}}_L|(\mE_{\LtoS})^\dagger(H_S)|\tilde{\tilde{1}}_L}.  
\end{equation}

\section{An inequality between \texorpdfstring{$\frakJ(\mN_S,H_S)$}{J} and \texorpdfstring{$\frakF(\mN_S,H_S)$}{F}}
\label{app:inequality}

Here, we prove a simple inequality of $\frakJ$ and $\frakF$. 
\begin{lemma}
Suppose the HKS condition is satisfied. Then it holds that $\frakJ(\mN_S,H_S)^2 \geq \frakF(\mN_S,H_S)$. 
\end{lemma}
\begin{proof}
Let $\vK_S = \begin{pmatrix}K_{S,1}\\ \vdots \\ K_{S,r}\end{pmatrix}$. According to the definitions \eqref{eq:def-frakJ} and \eqref{eq:def-frakF}, 
\begin{equation}
\frakJ(\mN_S,H_S)^2 = \min_{h:\vK_S^\dagger h\vK_S = H_S} (\Delta h)^2
= 4  \min_{h,\nu:\vK_S^\dagger h\vK_S = H_S - \nu \id}  \norm{h^2}, 
\end{equation}
and 
\begin{equation}
\frakF(\mN_S,H_S) = \barF(\mN_{S,\theta}) = 4 \min_{h:\vK_S^\dagger h\vK_S = H_S} \norm{\vK_S^\dagger h^2 \vK_S - H_S^2} = 4 \min_{h,\nu:\vK_S^\dagger h\vK_S = H_S -\nu \id} \norm{\vK_S^\dagger h^2 \vK_S - (H_S -\nu \id)^2}. 
\end{equation}
Then it is clear that $\frakF(\mN_S,H_S) \leq \frakJ(\mN_S,H_S)^2$ because \begin{equation}
    \norm{\vK_S^\dagger h^2 \vK_S^\dagger - (H_S -\nu \id)^2} \leq \norm{\vK_S^\dagger h^2 \vK_S} \leq \norm{h^2}\norm{\vK_S\vK_S^\dagger} = \norm{h^2}. 
\end{equation} 
\end{proof}

\section{Upper bounds on \texorpdfstring{$\frakJ(\mN_S,H_S)$}{J(N\_S,H\_S)} for random local noise and independent noise}
\label{app:frakJ}

In this section, we prove two useful upper bounds (\eqref{eq:frakJ-local} and \eqref{eq:frakJ-independent}) on $\frakJ$ for random local noise and independent noise, respectively.

\begin{lemma}
Let $\mN_S = \sum_{i=1}^m q_i \mN_S^{(i)}$ and $H_S = \sum_i^m H_S^{(i)}$ where $q_i > 0$ and $\sum_i^m q_i = 1$. Suppose the HKS condition is satisfied for each pair of $\mN_S^{(i)}$ and $H_S^{(i)}$. Then \eqref{eq:frakJ-local} holds, i.e.,  
\begin{equation}
\frakJ(\mN_S,H_S) \leq \max_i \frac{1}{q_i} \frakJ(\mN_S^{(i)},H_S^{(i)}). 
\end{equation}
\end{lemma}
\begin{proof}
Let $\mN_{S}^{(i)}(\cdot) = \sum_{j=1}^{r_i} K_{S,j}^{(i)} (\cdot) K_{S,j}^{(i)\dagger}$. 
Then we could take $\vK_S = \begin{pmatrix}\sqrt{q_1}\vK^{(1)}\\\sqrt{q_2}\vK^{(2)}\\\vdots\\\sqrt{q_m}\vK^{(m)}\end{pmatrix}$, where $\vK^{(i)} = \begin{pmatrix}K_{S,1}^{(i)}\\\vdots \\K_{S,r_i}^{(i)}\end{pmatrix}$. Recall that 
\begin{align}
\frakJ(\mN_S,H_S) &= \min_{h:\vK_S^\dagger h\vK_S = H_S} (\Delta h)
= 4  \min_{h,\nu:\vK_S^\dagger h\vK_S = H_S - \nu \id}  \norm{h},\\
\frakJ(\mN_S^{(i)},H_S^{(i)}) &= \min_{h^{(i)}:\vK^{(i)\dagger} h^{(i)} \vK^{(i)} = H_S^{(i)}} (\Delta h^{(i)})
= 4 \min_{h^{(i)},\nu^{(i)}:\vK^{(i)\dagger} h^{(i)} \vK^{(i)} = H_S^{(i)}-\nu^{(i)}\id}  \norm{h^{(i)}}. 
\end{align}
Suppose $(h^{(i)}_*,\nu^{(i)}_*)$ is optimal for $\frakJ(\mN_S^{(i)},H_S^{(i)})$, then let $\nu_* = \sum_i \nu^{(i)}_*$ and 
\begin{equation}
h = \begin{pmatrix}
\frac{h^{(1)}_*}{q_1} & & & \\
& \frac{h^{(2)}_*}{q_2} & & \\
& & \ddots & \\
& & & \frac{h^{(m)}_*}{q_m}, 
\end{pmatrix}, 
\end{equation}
We must have $\vK_S^\dagger h_* \vK_S = \sum_i \vK_S^{(i)\dagger} h_* \vK_S^{(i)} = H_S - \nu_*\id$. Therefore, 
\begin{equation}
\frakJ(\mN_S,H_S) = 4  \min_{h,\nu:\vK_S^\dagger h\vK_S = H_S - \nu \id}  \norm{h} \leq 4 \norm{h_*} = 4 \max_i \frac{1}{q_i}\norm{h^{(i)}_*} = \max_i \frac{1}{q_i} \frakJ(\mN_S^{(i)},H_S^{(i)}). 
\end{equation}
\end{proof}

\begin{lemma}
Let $\mN_S = \bigotimes_{l=1}^n \mN_{S_l}$ and $H_S = \sum_{l=1}^n H_{S_l}$. Suppose the HKS condition is satisfied for each pair of $\mN_{S_l}$ and $H_{S_l}$. Then \eqref{eq:frakJ-independent} holds, i.e.,  
\begin{equation}
\frakJ(\mN_S,H_S) \leq \sum_{l=1}^n \frakJ(\mN_{S_l},H_{S_l}). 
\end{equation}
\end{lemma}
\begin{proof}
Let $\mN_{S_l}(\cdot) = \sum_{j=1}^{r_l} K_{S_l,j} (\cdot) K_{S_l,j}^{\dagger}$ and $\vK_{S_l} = \begin{pmatrix}K_{S_l,1}\\\vdots\\ K_{S_l,r_l}\end{pmatrix}$. 
Then we could take $\vK_S = \bigotimes_{l=1}^n \vK_{S_l}$.

Recall that 
\begin{align}
\frakJ(\mN_S,H_S) &= \min_{h:\vK_S^\dagger h\vK_S = H_S} (\Delta h)
= 4  \min_{h,\nu:\vK_S^\dagger h\vK_S = H_S - \nu \id}  \norm{h},\\
\frakJ(\mN_{S_l},H_{S_l}) &= \min_{h_l:\vK_{S_l}^{\dagger} h_l \vK_{S_l} = H_{S_l}} (\Delta h_l)
= 4 \min_{h_l,\nu_l:\vK_{S_l}^{\dagger} h_l \vK_{S_l} = H_{S_l}-\nu_l\id}  \norm{h_l}. 
\end{align}
Suppose $(h_{l,*},\nu_{l,*})$ is optimal for $\frakJ(\mN_{S_l},H_{S_l})$, then let $\nu_* = \sum_l \nu_{l,*}$ and 
\begin{equation}
    h_* = \sum_{l=1}^n \id \otimes \cdots \underbrace{h_l}_{\text{$l$-th}} \otimes \cdots \otimes \id. 
\end{equation}
We must have $\vK_S^\dagger h_* \vK_S = \sum_l H_{S_l} - \nu_{l,*}\id = H_S - \nu_*\id$. Therefore, 
\begin{equation}
\frakJ(\mN_S,H_S) = 4  \min_{h,\nu:\vK_S^\dagger h\vK_S = H_S - \nu \id}  \norm{h} \leq 4 \norm{h_*} = 4 \sum_l \norm{h_{l,*}} = \sum_{l=1}^n \frakJ(\mN_{S_l},H_{S_l}).
\end{equation}
\end{proof}

\section{Computing the QEC inaccuracy of modified thermodynamic codes}
\label{app:thermo}

Let $\mN_S = \sum_{l=1}^n \frac{1}{n} \mN_{S_l}$, where $\mN_{S_l}(\cdot) = \ket{\vac}\bra{\vac}_{S_l} \otimes \trace_{S_l}(\cdot)$. 
Since for erasure channels the noise locations are effectively known (simply by measuring $\ket{\vac}$), to compute $\epsilon = \min_{\mR_\StoL} P(\mR_\StoL\circ\mN_S\circ\mE_\LtoS,\id_L)$, we could equivalently replace $\mN_S$ with the completely erasure noise on the first qubit $\mN_{S_1}$ and then  $\widehat\mN_{S\rightarrow B}(\cdot) = \sum_{i,j=0}^1 \trace(\braket{i_{S_1}|(\cdot)|j_{S_1}}) \ket{i_B}\bra{j_B}$. Note that 
\begin{align}
\ket{0_L} &=  \sqrt{\frac{n+m}{2(n+qm)}}\ket{1}\ket{(m-1)_{n-1}} +  \ket{0} \left( \sqrt{\frac{n-m}{2(n+qm)}} \ket{(m+1)_{n-1}}+ \sqrt{\frac{qm}{n+qm}}\ket{0^{\otimes n-1}} \right)\\
\ket{1_L} &= \ket{1}\left( \sqrt{\frac{n-m}{2(n+qm)}} \ket{(-m-1)_{n-1}}+ \sqrt{\frac{qm}{n+qm}}\ket{1^{\otimes n-1}} \right) + \sqrt{\frac{n+m}{2(n+qm)}}\ket{0}\ket{(-m+1)_{n-1}}
\end{align}
Then (for simplicity, we sometimes omit the subscripts $L$, $S$, $\LtoS$ and $\StoL$)
\begin{gather}
\widehat{\mN\circ\mE}(\ket{0_L}\bra{0_L}) = \widehat{\mN}(\ket{\frakc_0^q}\bra{\frakc_0^q}) = \frac{n+(2q-1)m}{2(n+qm)}\ket{0_B}\bra{0_B} + \frac{n+m}{2(n+qm)}\ket{1_B}\bra{1_B} =: \rho_0,\\
\widehat{\mN\circ\mE}(\ket{1_L}\bra{1_L}) = \widehat{\mN}(\ket{\frakc_1^q}\bra{\frakc_1^q}) = \frac{n+(2q-1)m}{2(n+qm)}\ket{1_B}\bra{1_B} + \frac{n+m}{2(n+qm)}\ket{0_B}\bra{0_B} =: \rho_1,\\
\widehat{\mN\circ\mE}(\ket{0_L}\bra{1_L}) = \widehat{\mN\circ\mE}(\ket{1_L}\bra{0_L}) = 0.
\end{gather}
Then let $\ket{\psi} = \psi_0\ket{0_L}\ket{\psi^0_R} + \psi_1\ket{1_L}\ket{\psi^1_R}$ be an arbitrary pure state on $L\otimes R$. We have 
\begin{align}
\epsilon &= 
\min_{\mR'} P(\widehat{\mN\circ\mE},\mR'\circ \widehat{\id})\\
& = \min_{\zeta} \max_{\ket{\psi}} P\left(|\psi_0|^2\rho_0 \otimes \ket{\psi^0_R}\bra{\psi^0_R}+|\psi_1|^2\rho_1 \otimes \ket{\psi^1_R}\bra{\psi^1_R},\zeta \otimes  (|\psi_0|^2\ket{\psi^0_R}\bra{\psi^0_R}+|\psi_1|^2\ket{\psi^1_R}\bra{\psi^1_R})\right)\\
& \geq \min_{\zeta}\max_{|\psi_0|^2+|\psi_1|^2 = 1} P\left(|\psi_0|^2\rho_0+|\psi_1|^2\rho_1,\zeta\right) \\&\geq \min_{\zeta}\frac{1}{2}\left(P(\rho_0,\zeta)+P(\rho_1,\zeta)\right) \\ &\geq \frac{1}{2}P(\rho_0,\rho_1)  \\ &= \frac{(1-q)m}{2(n+qm)}=\frac{(1-q)m}{2n} + O\left(\frac{m^2}{n^2}\right), 
\end{align}
where we used the monotonicity of purified distance for the third line and the triangular inequality for the fifth line. On the other hand, we have the following upper bound:  
\begin{align}
\label{eq:upper} 
\epsilon &= 
\min_{\mR'} P(\widehat{\mN\circ\mE},\mR'\circ \widehat{\id})\\
& = \min_{\zeta} \max_{\ket{\psi}} P\left(|\psi_0|^2\rho_0 \otimes \ket{\psi^0_R}\bra{\psi^0_R}+|\psi_1|^2\rho_1 \otimes \ket{\psi^1_R}\bra{\psi^1_R},\zeta \otimes  (|\psi_0|^2\ket{\psi^0_R}\bra{\psi^0_R}+|\psi_1|^2\ket{\psi^1_R}\bra{\psi^1_R})\right)\\
& \leq \max_{\ket{\psi}} P\left(|\psi_0|^2\rho_0 \otimes \ket{\psi^0_R}\bra{\psi^0_R}+|\psi_1|^2\rho_1 \otimes \ket{\psi^1_R}\bra{\psi^1_R},\zeta \otimes  (|\psi_0|^2\ket{\psi^0_R}\bra{\psi^0_R}+|\psi_1|^2\ket{\psi^1_R}\bra{\psi^1_R})\right)\big|_{\zeta = \frac{\rho_0+\rho_1}{2}}\\
& \leq \max_{|\psi_0|^2+|\psi_1|^2 = 1} \sqrt{1 - \left(\abs{\psi_0}^2 f(\rho_0,\zeta) + \abs{\psi_1}^2 f(\rho_1,\zeta)\right)^2} \bigg|_{\zeta = \frac{\rho_0+\rho_1}{2}}\\ &= P\left(\rho_1,\frac{1}{2}(\rho_0+\rho_1)\right) \\&= \sqrt{\frac{1}{2} - \frac{\sqrt{(n+m)(n+(2q-1)m)}}{2(n+qm)}}= \frac{(1-q)m}{2n} + O\left(\frac{m^2}{n^2}\right), 
\end{align}
where we used the joint concavity of the fidelity~\cite{nielsen2002quantum} for the fourth line and the fact that $P(\rho_0,\zeta) = P(\rho_1,\zeta)$ when $\zeta = \frac{\rho_0+\rho_1}{2}$ for the fifth line. Therefore, we conclude that 
\begin{equation}
\epsilon(\mN_{S},\mE_{\LtoS}) = \frac{(1-q)m}{2n} + O\left(\frac{m^2}{n^2}\right). 
\end{equation}

We claim that an optimal recovery channel (up to the leading order in $m/n$) 
achieving the smallest $\epsilon$ is 
\begin{equation}
\begin{split}
\mR_{\StoL}^{\optL}
=\;& \sum_{l=1}^n\left(\ket{0_L}\bra{(m-1)_{n-1}}_{S\backslash S_l} + \ket{1_L}\bra{\varphi_1}_{S\backslash S_l}\right)\bra{\vac}_{S_l}(\cdot)\ket{\vac}_{S_l}\left(\ket{(m-1)_{n-1}}_{S\backslash S_l}\bra{0_L}+\ket{\varphi_1}_{S\backslash S_l}\bra{1_L}\right) \\
&+  \left(\ket{0_L}\bra{\varphi_{0}}_{S\backslash S_l} + \ket{1_L}\bra{(-m+1)_{n-1}}_{S\backslash S_l}\right)\bra{\vac}_{S_l}(\cdot)\ket{\vac}_{S_l}\left(\ket{\varphi_{0}}_{S\backslash S_l}\bra{0_L}+\ket{(-m+1)_{n-1}}_{S\backslash S_l}\bra{1_L}\right)\\
&+ \ket{0_L}\bra{0_L} \trace\left((\cdot)\tPi^\perp\right), 
\end{split}
\end{equation}
where 
\begin{align}
\ket{\varphi_1} &= \sqrt{\frac{n-m}{n+(2q-1)m}} \ket{(-m-1)_{n-1}}+ \sqrt{\frac{2qm}{n+(2q-1)m}}\ket{1^{\otimes n-1}}, \\
\ket{\varphi_0} &=  \sqrt{\frac{n-m}{n+(2q-1)m}} \ket{(m+1)_{n-1}}+ \sqrt{\frac{2qm}{n+(2q-1)m}}\ket{0^{\otimes n-1}},
\end{align} $\tPi$ is the projection onto ${\rm span}\{\ket{\vac}_{S_l}\ket{(m-1)_{n-1}}_{S\backslash S_l},\ket{\vac}_{S_l}\ket{(-m+1)_{n-1}}_{S\backslash S_l},\ket{\vac}_{S_l}\ket{\varphi_1}_{S\backslash S_l},\ket{\vac}_{S_l}\ket{\varphi_0}_{S\backslash S_l},\forall l\}$ and $\tPi^\perp$ is its orthogonal projector. We have 
\begin{gather}
\mR^{\optL}\circ\mN\circ\mE(\ket{0_L}\bra{0_L}) = \ket{0_L}\bra{0_L},\\
\mR^{\optL}\circ\mN\circ\mE(\ket{1_L}\bra{1_L}) = \ket{1_L}\bra{1_L},\\
\mR^{\optL}\circ\mN\circ\mE(\ket{0_L}\bra{1_L}) = \sqrt{\frac{(n+m)(n+(2q-1)m)}{(n+qm)^2}}\ket{0_L}\bra{1_L}. 
\end{gather}
That is, $\mR^{\optL}\circ\mN\circ\mE$ is a dephasing channel. Using \lemmaref{lemma:dephasing}, we obtain
\begin{equation}
\label{eq:attainable}
P(\mR^{\optL}\circ\mN\circ\mE,\id) 
= \sqrt{\frac{1}{2} - \frac{\sqrt{(n+m)(n+(2q-1)m)}}{2(n+qm)}}
= \frac{(1-q)m}{2n} + O\left(\frac{m^2}{n^2}\right),
\end{equation}
and
\begin{equation}
D_\diamond(\mR^{\optL}\circ\mN\circ\mE,\id) 
= \frac{1}{2} - \frac{\sqrt{(n+m)(n+(2q-1)m)}}{2(n+qm)}
= \left(\frac{(1-q)m}{2n}\right)^2 + O\left(\frac{m^3}{n^3}\right). 
\end{equation}
Since 
\begin{equation}
\epsilon^2 \leq \epsilon_\diamond \leq D_\diamond(\mR^{\optL}\circ\mN\circ\mE,\id),
\end{equation}
we have $\epsilon_\diamond = \big(\frac{(1-q)m}{2n}\big)^2 + O\big(\frac{m^3}{n^3}\big)$.

\section{Refinement of 
{Theorem~23}
}
\label{app:refine}

Here, we provide a refinement of \thmref{thm:local-point} by replacing the local covariance violation $\delta_\point$ with a quantity $\delta_\point^\star$ which characterizes the local covariance violation under noise and recovery. Then we show that $\delta_\point^\star$ could be much more smaller than $\delta_\point$, which provides one explanation of the looseness of \thmref{thm:local-point}. Specifically, 
\begin{equation}
\delta^\star_\point
:= \min_{\substack{\mR_{\StoL}:\\P(\mR_{\StoL}\circ\mN_S\circ\mE_{\LtoS},\id_L) = \epsilon}}
\sqrt{F(\mR_{\StoL}\circ\mN_{S}\circ\mU_{S,\theta}\circ\mE_{\LtoS}\circ\mU_{L,\theta}^\dagger)\big|_{\theta = 0}}.
\end{equation}
Clearly, $\delta^\star_\point \leq \delta_\point$ because of the monotonicity of the QFI. And we have the following refinement of \thmref{thm:local-point}: 
\begin{theorem}
\label{thm:local-refine} 
\nonisometric
When $1 - 2\epsilon^2>\delta^\star_\point/\Delta H_L$, it holds that
\begin{equation}
\frac{\epsilon \sqrt{1-\epsilon^2}}{1 - 2\epsilon^2 - \delta^\star_\point/(\Delta H_L)} \geq \frac{\Delta H_L}{\sqrt{4 \frakF}}.
\end{equation}
When $\epsilon \ll 1$, we have 
\begin{equation}
\delta_\point^\star + 2\epsilon\sqrt{\frakF} \gtrsim \Delta H_L. 
\end{equation} 
\end{theorem}
\begin{proof}
In order to prove \thmref{thm:local-refine}, we only need to prove a refinement of \lemmaref{lemma:local} where $\delta_\point$ is replaced by $\delta^\star_\point$, i.e., $\xi_\theta = \bra{0_C}\mN_{C,\theta}(\ket{0_C}\bra{1_C})\ket{1_C}$ satisfies 
\begin{equation}
\label{eq:refine}
\abs{\partial_\theta\xi_\theta}^2\big|_{\theta = 0}  \geq  ((1-2\epsilon^2)\Delta H_L - \delta^\star_\point)^2,
\end{equation}
when $(1-2\epsilon^2)\Delta H_L \geq \delta_\point^\star$. The rest of the proof follows exactly the same from the proof of \thmref{thm:local-point}. Here 
\begin{align}
\mN_{C,\theta} = \mR^{\rep}_{\SAtoC} \circ ( \mR_{\StoL}^{\optL} \circ\mN_{S} \otimes \id_A) \circ  (\mU_{S,\theta}\circ \mE_{\LtoS} \otimes \id_A) \circ \mE_{\CtoLA}^{\rep},
\end{align}
where $\mR^{\optL}_{\StoL}$ is a recovery channel such that 
$\epsilon =  P(\mR^{\optL}_{\StoL}\circ \mN_{S}\circ\mE_{\LtoS}, \id_{L})$. To see that \eqref{eq:refine} must be true, we only need to revisit the proof of \lemmaref{lemma:local} in \appref{app:local-proof} and note that $\delta_\point$ in \eqref{eq:refine-1} and then in \eqref{eq:refine-2} can be replaced by $\delta_\point^\star$ because 
\begin{equation}
F(\mD_{\theta,C})\big|_{\theta = 0} \leq F(\mR^{\optL}_{\StoL}\circ \mN_{S,\theta} \circ \mE_{\LtoS} \circ \mU_{L,\theta}^\dagger)\big|_{\theta = 0} = (\delta^\star_\point)^2,
\end{equation}
where $\mD_{C,\theta} = \mR^{\rep}_{\SAtoC} \circ ( \mR_{\StoL}^{\optL} \circ\mN_{S} \otimes \id_A) \circ  (\mU_{S,\theta}\circ \mE_{\LtoS} \circ \mU^\dagger_{L,\theta}\otimes \id_A) \circ \mE_{\CtoLA}^{\rep}$. 

\end{proof}

Now we show that $\delta_\point^*$ could be much more smaller than $\delta_\point$. In \secref{sec:case-study}, we had two exact QEC code examples (the modified thermodynamic code at $q=1$, and the $[[n = 2^t - 1,1,3]]$ quantum Reed--Muller code) where $\delta_\point = \Theta(\sqrt{n})$ and $\delta_\charge = O(1)$. The following proposition, combined with \propref{prop:point-charge}, indicates that $\delta_\point^\star = O(1)$ and has at least a quadratic gap to $\delta_\point$ in these examples. 

\begin{proposition}
\isometric When $\epsilon = 0$, i.e., the code is exactly error-correcting, it holds that $\delta_\point^\star \leq \delta_\charge$. 
\end{proposition}
\begin{proof}

Suppose $\mE_{\LtoS}(\cdot) = W(\cdot)W^\dagger$ where $W$ is isometric. Then $\delta_\charge = \Delta\left({H_L - W^\dagger H_S W}\right)$. 
When $\epsilon = 0$, we must have $\mR_{\StoL}^{\optL} \circ \mN_S \circ \mE_{\LtoS} = \id_L$. Let $\Pi$ be the projection onto the code subspace and $\Pi^\perp = \id - \Pi$. Then according to Theorem 10.1 in Ref.~\cite{nielsen2002quantum}, we see that there exists a recovery channel $\mR_{\StoL}^{\optL}$ such that $\mR_{\StoL}^{\optL} \circ \mN_S$ have the following form: 
\begin{equation}
\mR_{\StoL}^{\optL} \circ \mN_S(\cdot) = W^\dagger(\cdot)W + \mR^\perp(\Pi^\perp(\cdot)\Pi^\perp), 
\end{equation}
for some CPTP map $\mR^{\perp}$. 
Let $\ket{\psi_\theta} = U_{S,\theta} W U_{L,\theta}^\dagger \ket{\psi} = e^{-iH_S\theta}We^{iH_L\theta}\ket{\psi}$, then $
\ket{\partial_\theta\psi_\theta} = e^{-iH_S\theta}We^{iH_L\theta}iH_L\ket{\psi} - iH_S e^{-iH_S\theta}We^{iH_L\theta}\ket{\psi}. $ Then 
\begin{gather}
\rho_\theta = \mR_{\StoL}^{\optL} \circ \mN_{S,\theta} \circ \mE_{\LtoS} \circ \mU_{L,\theta}^\dagger (\ket{\psi}\bra{\psi})= W^\dagger \ket{\psi_\theta}\bra{\psi_\theta} W + \mR^\perp(\Pi^\perp(\ket{\psi_\theta}\bra{\psi_\theta})\Pi^\perp), \\
\rho_\theta|_{\theta=0} = W^\dagger \ket{\psi_\theta}\bra{\psi_\theta} W = \ket{\psi}\bra{\psi}, \\
\partial_\theta\rho_\theta|_{\theta=0} = W^\dagger \ket{\partial_\theta\psi_\theta}\bra{\psi_\theta} W + W^\dagger \ket{\psi_\theta}\bra{\partial_\theta\psi_\theta} W = \ket{\psi'}\bra{\psi} + \ket{\psi}\bra{\psi'}, 
\end{gather}
where $\ket{\psi'} = i(H_L - W^\dagger H_S W)\ket{\psi}. $
Clearly, $\braket{\psi|\psi'} + \braket{\psi'|\psi} = 0$. 
Then 
\begin{align}
(\delta^\star_\point)^2 &\leq F(\mR_{\StoL}^{\optL} \circ \mN_{S,\theta} \circ \mE_{\LtoS} \circ \mU_{L,\theta}^\dagger)\big|_{\theta = 0} \\ 
&= \max_{\psi} F(\rho_\theta)\big|_{\theta = 0} = \max_{\psi} 4 (\abs{\braket{\psi'|\psi'}}^2 - \abs{\braket{\psi|\psi'}}^2) = (\delta_\charge)^2. 
\end{align}

\end{proof}

\bibliographystyle{naturemag}
\bibliography{refs-approx-cov-nc}

\end{document}